\documentclass[3p]{elsarticle}
\pdfoutput=1
\usepackage[colorlinks,urlcolor=black]{hyperref}
\usepackage{amsthm}
\usepackage{mdwlist} 
\usepackage[normalem]{ulem}  

\usepackage{amssymb,stmaryrd,amsmath,amsthm,xspace}
\usepackage{nccmath}
\usepackage{prooftree}
\usepackage{graphicx}
\usepackage[colorlinks,urlcolor=black]{hyperref}
\usepackage{soul}
\usepackage[all]{xy}

\usepackage[usenames,dvipsnames]{xcolor}

\usepackage{etex} 
\usepackage{tikz}
\definecolor{shade}{HTML}{F6F6FF}
\newcommand\shademath[1]
{
\raisebox{-12pt}{\begin{tikzpicture}
\node [fill=shade,rounded corners=3pt]
{$#1$};
\end{tikzpicture}}
}

\usepackage{microtype}
\usepackage{tgtermes}

\usepackage{float}
\floatstyle{boxed} \restylefloat{figure}

\usepackage{fancybox}
\newlength{\mylength}
\newenvironment{frameqn}%
{\setlength{\fboxsep}{4pt}
\setlength{\mylength}{\linewidth}%
\addtolength{\mylength}{-2\fboxsep}%
\addtolength{\mylength}{-2\fboxrule}%
\Sbox
\minipage{\mylength}%
$$}%
{$$\endminipage\endSbox
{
\[\fbox{\TheSbox}\]
}}

\newenvironment{frametxt}%
{
\setlength{\fboxsep}{4pt}
\setlength{\mylength}{\linewidth}%
\addtolength{\mylength}{-2\fboxsep}%
\addtolength{\mylength}{-2\fboxrule}%
\Sbox
\minipage{\mylength}%
}%
{\endminipage\endSbox
{
\[\fbox{\TheSbox}\]
}}

{
\begin{figure}[H]
}
{\end{figure}
}

\newcommand\XD{X_{\hspace{-.15ex}\scalebox{.6}{$D$}}}
\newcommand\YD{Y_{\hspace{-.25ex}\scalebox{.6}{$D$}}}

\newcommand\fa{\f{fa}}
\newcommand\fv{\f{fv}}
\newcommand\fU{\f{fU}}
\newcommand\bigact[0]{\raisebox{0.3ex}{\scalebox{.5}{$\hspace{1pt}\bullet$}}}

\newcommand\nontriv{\f{nontriv}}
\newcommand\dom{\f{dom}}
\newcommand\img{\f{img}}
\newcommand\nopi{{\scalebox{.6}{\sout{$\pi$}}}}
\newcommand\nopicent{\mathrel{\cent^{\hspace{-.95ex}\raisebox{.5pt}{\nopi}}}} 
\newcommand\nopiment{\mathrel{\ment^{\hspace{-.95ex}\raisebox{1pt}{\nopi}}}} 
\newcommand\holcent{\mathrel{\cent^{\hspace{-.95ex}\raisebox{.5pt}{\scalebox{.55}{$\lambda$}}}}}
\newcommand\holment{\mathrel{\ment^{\hspace{-.95ex}\raisebox{1.5pt}{\scalebox{.55}{$\lambda$}}}}}
\newcommand\ns[1]{\mathsf{#1}^{\scalebox{.5}{$\circlearrowright$}}}
\newcommand\rs[1]{\mathsf{#1}^{\scalebox{.5}{$\rightrightarrows$}}}

\newcommand\ren{\f{ren}}
\newcommand\Ren[1]{\ren(#1)}
\newcommand\supp{\f{supp}}
\def\equiv{=}
\newcommand\raws{{{s}}}
\newcommand\rawr{{{r}}}
\newcommand\rawt{{{t}}}
\newcommand\rawu{{{u}}}
\newcommand\rawphi{{\phi}}
\newcommand\rawpsi{{\psi}}

\newcommand{\hol}[2]{{\lfloor} #2 {\rfloor}^{\hspace{-.25ex}\raisebox{-.25ex}{\scalebox{.4}{$#1$}}}}
\newcommand{\denot}[3]{\llbracket #3 \rrbracket_{\scalebox{.6}{$#2$}}^{\hspace{-.1ex}\raisebox{.05ex}{\scalebox{.4}{$#1$}}}}

\newcommand\GammaX{{D\cap\pmss(X)}}
\newcommand\GammaY{{D\cap\pmss(Y)}}
\newcommand\smtf[1]{{\scalebox{.45}{$\tf{#1}$}}}

\newcommand\hiden{{\scalebox{.4}{$\mathcal H$}}}
\newcommand\iden{{\scalebox{.4}{$\mathcal I$}}}
\newcommand\hden{{\scalebox{.4}{$\mathcal H$}}}

\newcommand\pmss{\f{pms}}
\newcommand\sort{\f{sort}}
\newcommand\type{\f{type}}
\newcommand\somerel{\mathrel{\mathcal R}}
\newcommand\theory[1]{\ensuremath{\mathsf{#1}}}

\newcommand\Chi{\raisebox{.15em}{\large$\chi$}}    
\newbox\tempa
\newbox\tempb
\newdimen\tempc
\def\mud#1{\hfil $\displaystyle{\mathstrut #1}$\hfil}
\def\rig#1{\hfil $\displaystyle{#1}$}
\def\irulehelp#1#2#3{\setbox\tempa=\hbox{$\displaystyle{\mathstrut #2}$}%
		        \setbox\tempb=\vbox{\halign{##\cr
	\mud{#1}\cr
	\noalign{\vskip\the\lineskip}%
	\noalign{\hrule height 0pt}%
	\rig{\vbox to 0pt{\vss\hbox to 0pt{${\; #3}$\hss}\vss}}\cr
	\noalign{\hrule}%
	\noalign{\vskip\the\lineskip}%
	\mud{\copy\tempa}\cr}}%
		      \tempc=\wd\tempb
		      \advance\tempc by \wd\tempa
		      \divide\tempc by 2 }
\def\irule#1#2#3{{\irulehelp{#1}{#2}{#3}%
		     \hbox to \wd\tempa{\hss \box\tempb \hss}}}

\newcommand\basesort{\tau}
\newcommand\basetype{\mu}

\newcommand{\model}[1]{\denot{\mathcal I}{}{#1}}
\newcommand{\holmodel}[1]{\denot{\mathcal H}{}{#1}}

\newcommand\deffont[1]{{\bf #1}}
\newcommand\tf[1]{{\mathsf{#1}}}
\newcommand\f[1]{\mathit{#1}}

\newcommand\act{{\cdot}}
\newcommand\liff{\mathrel{\Leftrightarrow}}
\newcommand\limp{\Rightarrow}
\newcommand\Forall[1]{\forall #1.}

\newcommand\lam[1]{\lambda #1.}
\newcommand\aeq{\mathrel{=_{\alpha}}}
\newcommand\abeq{\mathrel{=_{\alpha\beta}}}
\newcommand\id{\f{id}}
\newcommand\Id{\f{Id}}
\newcommand\cent{\vdash}
\newcommand\ment{\vDash}
\newcommand\sm{{\mapsto}}
\newcommand\ssm{{{:}{=}}}
\newcommand\mone{{\text{-}1}}
\newcommand\rulefont[1]{\ensuremath{{\bf (#1)}}\xspace}

\newcommand\atomsup{\mathbb A^{\hspace{-.25ex}\scalebox{.6}{$>$}}}
\newcommand\atomsdown{\mathbb A^{\hspace{-.25ex}\scalebox{.6}{$<$}}}
\newcommand\atoms{{\mathbb A}}

\newtheoremstyle{jamiestyle}
  {4pt}
  {0pt}
  {\it}
  {0pt}
  {\bf}
  {.}
  { }
  {}
\theoremstyle{jamiestyle}
\newtheorem{thrm}{Theorem}[section]
\newtheorem{prop}[thrm]{Proposition}
\newtheorem{lemm}[thrm]{Lemma}
\newtheorem{corr}[thrm]{Corollary}
\newtheoremstyle{jamienfstyle}
  {4pt}
  {0pt}
  {\normalfont}
  {0pt}
  {\bf}
  {.}
  { }
  {}
\theoremstyle{jamienfstyle}
\newtheorem{nttn}[thrm]{Notation}
\newtheorem{defn}[thrm]{Definition}
\newtheorem{xmpl}[thrm]{Example}
\newtheorem{rmrk}[thrm]{Remark}

\usepackage{array}
\newcolumntype{L}[1]{>{$}p{#1}<{$}}
\newcolumntype{C}[1]{>{\centering$}p{#1}<{$}}
\newcolumntype{R}[1]{>{\raggedleft$}p{#1}<{$}}
\newcommand\maketab[2]
   {\newenvironment{#1}
      {\begin{quote}\noindent\begin{tabular}{#2}}
      {\end{tabular}\end{quote}}
    \newenvironment{#1noquote}{\noindent\begin{tabular}{#2}}{\end{tabular}}
    }

\journal{Version of record published in Theoretical Computer Science, Volume 451, 14 September 2012, Pages 38-69. \href{https://doi.org/10.1016/j.tcs.2012.06.007}{DOI: 10.1016/j.tcs.2012.06.007}}

\begin{document}

\maketab{tab1}{@{\hspace{-0em}}L{6em}R{7em}@{\ }L{6em}@{\ }L{7em}}
\maketab{tab3}{@{\hspace{-4em}}R{10em}@{\ }L{15em}@{\ }L{15em}}
\maketab{tab7}{@{\hspace{-2em}}R{10em}@{\ }L{12em}L{14em}}

\title{PNL to HOL: from the logic of nominal sets to the logic of higher-order functions\tnoteref{vor}}
\tnotetext[vor]{This is the authors' accepted manuscript, shared as per the journal's \href{https://www.elsevier.com/about/policies-and-standards/sharing}{\emph{sharing policy}} (\href{https://web.archive.org/web/20231109022149/https://www.elsevier.com/about/policies-and-standards/sharing}{permalink}).  Details of the version of record are below.}
\author[1]{Gilles Dowek}
\ead[url]{http://www-roc.inria.fr/who/Gilles.Dowek/} 
\author[2]{Murdoch J. Gabbay}
\ead[url]{http://www.gabbay.org.uk} 
\begin{abstract}
Permissive-Nominal Logic (PNL) extends first-order predicate logic with term-formers that can bind names in their arguments.
It takes a semantics in (permissive-)nominal sets.
In PNL, the $\forall$-quantifier or $\lambda$-binder are just term-formers satisfying axioms, and their denotation is functions on nominal atoms-abstraction.

Then we have higher-order logic (HOL) and its models in ordinary (i.e. Zermelo--Fraenkel) sets; the denotation of $\forall$ or $\lambda$ is functions on full or partial function spaces.

This raises the following question: how are these two models of binding connected?
What translation is possible between PNL and HOL, and between nominal sets and functions?

We exhibit a translation of PNL into HOL, and from models of PNL to certain models of HOL.
It is natural, but also partial: we translate a restricted subsystem of full PNL to HOL. 
The extra part which does not translate is the symmetry properties of nominal sets with respect to permutations.
To use a little nominal jargon: we can translate names and binding, but not their nominal equivariance properties. 
This seems reasonable since HOL---and ordinary sets---are not equivariant.

Thus viewed through this translation, PNL and HOL and their models do different things, but they enjoy non-trivial and rich subsystems which are isomorphic.
\end{abstract}


\maketitle


\tableofcontents

\section{Introduction}

Permissive-Nominal Logic (PNL) extends first-order predicate logic with name-binding term-formers.
For instance first-order logic, set theory, and the untyped $\lambda$-calculus axiomatise in PNL; their binders $\forall$, comprehension, and $\lambda$ are just modelled as binding PNL term-formers.
The canonical semantics of PNL is in nominal sets, and it is first-order.

Higher-order logic (HOL) also has binding \cite{miller:logho,farmer:sevvst}. 
This has been used to encode other binders, e.g. the Church encoding of quantifiers as constants of higher type such as $\forall:(\iota{\to} o){\to} o$ \cite{andrews:intmlt,church:forstt}; higher-order abstract syntax (HOAS) encoding term-formers of an encoded syntax with binders as constants of higher type such as $\forall:(\iota{\to} \rho){\to}\rho$ or $\forall:(\nu{\to}\rho){\to}\rho$ (strong vs. weak HOAS)\footnote{A word of clarification here: we take $o$ to be a type of truth-values, $\iota$ to be a type of terms, and $\rho$ to be a type of predicates.  $\forall$-the-quantifier generates truth-values, whence the type headed by $o$, namely $\forall:(\iota{\to} o){\to} o$.  $\forall$-the-syntax-building-constant in HOAS generates \emph{terms}, whence the types headed by $\rho$, namely $\forall:(\iota{\to} \rho){\to}\rho$ or $\forall:(\nu{\to}\rho){\to}\rho$.  Do not confuse a HOL constant for a HOAS-style binder (a way to give meaning to building syntax with binding) with a HOL constant for the corresponding quantifier (a way to give meaning to what that syntax is intended to denote; namely, actual quantification).} 
 \cite{despeyroux94higherorder,pfenning:hoas}; and higher-order rewrite systems \cite{mayr:horwc}.

Since PNL is first-order and has a sound and complete semantics (so expressivity and models are fairly `small'), whereas HOL is higher-order (so expressivity and models are fairly `large'), the natural direction for a translation is from nominal sets and PNL, to functions and HOL (a \emph{shallow embedding} of PNL into HOL).\footnote{A \emph{deep embedding} e.g. of HOL in PNL is an answer to a different question; for more on this direction, see \cite{gabbay:unialt}.} 

In this paper we translate a subsystem of PNL into HOL and prove it sound and complete using arguments on nominal sets and nominal renaming sets models \cite{gabbay:nomrs}. 
The proof of completeness involves giving a functional semantics to nominal terms, and a nominal semantics to $\lambda$-terms in the spirit of Henkin models \cite{andrews:intmlt,benzmuller:higose}. 
This involves a construction on nominal sets models corresponding to a free extension to \emph{nominal renaming sets}, as previously considered by the second author with Hofmann \cite{gabbay:nomrs}.

The partiality of the translation of PNL seems to be inherent and reflects natural differences in structure between nominal and `ordinary' sets; nominal sets are subject to the action of a \emph{symmetry group} of atoms-permutations, which cannot be naturally represented in HOL or its `ordinary' sets semantics.
That is, it is not the case that nominal techniques are `just' a concise presentation of HOL with a weakened $\beta$-equivalence (e.g. higher-order patterns \cite{miller:logpll}).
There is that, but there is also more. 
The nominal and functional models of binding are distinct, but they do have non-trivial and rich subsystems which are isomorphic in a sense made precise in this paper.

\subsection{Some background on PNL}

We study PNL for its own sake in this paper, but the interested reader can find example nominal theories in the literature: for substitution, $\beta$-equivalence, and first-order logic \cite{gabbay:capasn,gabbay:capasn-jv,gabbay:nomalc,gabbay:oneaah,gabbay:oneaah-jv}.

These axiomatisations are in nominal algebra (which can be viewed as the equality fragment of PNL) and are accompanied by proofs of correctness in the respective papers.

Not all PNL theories are expressed in the equality fragment.
For instance, in the papers which introduced PNL \cite{gabbay:pernl,gabbay:pernl-jv} we included theories of first-order logic and arithmetic which put universal quantification to the left of an implication.

To give some idea of what this family of logics looks like in practice, assume a name-sort $\nu$ and a base sort $\iota$ and term-formers $\tf{lam}:([\nu]\iota)\iota$,\ \ $\tf{app}:(\iota,\iota)\iota$,\  and $\tf{var}:(\nu)\iota$.
(Full definitions are in the body of the paper.)
We sugar $\tf{lam}([a]r)$ to $\lam{a}r$ and $\tf{app}(r',r)$ to $r'r$ and $\tf{var}(a)$ to $a$.
Atoms in PNL are a form of data and populate their own sort $\nu$; so $\tf{var}$ serves to map them into the sort $\iota$, where they represent object-level variables.

Here is $\eta$-equivalence, written out as it would be informally:
$$
\lam{x}(tx)=t\text{\ \ \ if $x$ is not free in $t$}
$$
Here is a PNL axiom for $\eta$-equivalence, written out formally:
$$
\Forall{Z}(\lam{a}(Za)=Z)\quad (a\not\in\pmss(Z))
$$
(See \cite{gabbay:nomalc} for a detailed study of this axiom in a nominal context.)

$a$ is an \emph{atom} and corresponds to the \emph{object-level variable} $x$; $a$ is not a PNL variable but it \emph{represents} a variable of the object level system being axiomatised. 
$Z$ is an \emph{unknown} and correspond to the \emph{meta-level variable} $t$; $Z$ is a variable in PNL and may be instantiated. 

The reader can see how similar the two axioms look. 
Their status is different in the following sense: whereas $t$ is typically taken to range over terms, $Z$ ranges over elements of nominal sets (via a valuation; see Definition~\ref{defn.valuation}).
This is possible because nominal sets have a notion of \emph{supporting set of atoms} which mirrors the free variables of a term.

The condition $a\not\in\pmss(Z)$ is a \emph{typing condition} in PNL.
The types, or \emph{permission sets} as we call them, restrict the support of denotations associated to $Z$ by a valuation.
They correspond to freshness side-conditions in nominal terms from \cite{gabbay:nomu-jv} and to informal freshness conditions of the form `$x$ not free in $t$' in informal practice.
To see this intuition made formal see a translation from nominal terms to permissive-nominal terms in \cite{gabbay:perntu-jv}.

There is no requirement to axiomatise $\alpha$-equivalence because this is done automatically by the PNL system.

For instance, axioms for $\beta$-equivalence \cite{gabbay:nomalc} are:
$$
\begin{array}{l@{\ }l@{\ =\ }l@{\ }l}
\Forall{Y}&(\lam{a}a)Y&Y
\\
\Forall{Z,X}&(\lam{a}Z) X&Z &(a\not\in\pmss(Z))
\\
\Forall{X',X,Y}&(\lam{a}(X'X))Y&((\lam{a}X') Y)((\lam{a}X)Y)
\\
\Forall{X,Z}&(\lam{b}(\lam{a}X))Z&\lam{a}((\lam{b}X) Z) &(a\not\in\pmss(Z))
\\
\Forall{X}&(\lam{a}X)a&X 
\end{array}
$$ 
See \cite{gabbay:nomalc} for a proof that these axioms really \emph{do} axiomatise the $\lambda$-calculus.\footnote{These axioms first appeared in \cite{gabbay:capasn,gabbay:oneaah} where a slightly different version of the final axiom was used.  They are equivalent; see part~5 of Example~\ref{xmpl.aeq} in this paper.}

It is important to appreciate that most models of the axioms above are abstract and algebraic, not concrete and syntactic.
So for instance the final axiom is not admissible because the objects $X$ ranges over do not necessarily have inductive structure and we cannot necessarily push the $\beta$-redex down to the atoms.

In particular, this is not a paper about representing syntax-with-binding, unlike the first applications of nominal techniques to syntax-with-binding \cite{gabbay:newaas-jv}.
True, we wrote above that atoms represent object-level variables.
But indeed, they represent \emph{variables}, not variable \emph{symbols}.
Nominal sets admit open elements and we can write axioms about names and binding in PNL, so that the full behaviour of variables is in PNL susceptible to algebraic axiomatisation.
For instance the equality axioms above give atoms the behaviour of $\beta$-convertible $\lambda$-calculus variables; see \cite{gabbay:nomalc} for a proof.
Most models of the theory above are not built out of syntax. 

Several nominal algebraic theories have been developed, with proofs of correctness: see \cite{gabbay:capasn-jv} (substitution), \cite{gabbay:nomalc} ($\lambda$-calculus) or \cite{gabbay:pernl,gabbay:pernl-jv} (first-order logic and arithmetic). 
For examples of natural non-syntactic models of nominal-style theories see \cite[Subsection~5.2]{gabbay:pernl-jv} and \cite{gabbay:stodfo}.

Thus, the design philosophy of PNL is that axioms should look like what we would write informally anyway, where variables map to atoms, meta-variables to unknowns, binding to atoms-abstraction, and capture-avoidance conditions to choice of permission sets.
In particular, open terms and predicates map to elements and subsets of nominal sets with nonempty support.

\subsection{On symmetry}

The thing that goes missing in the translation from PNL to HOL is \emph{name-symmetry}.
Nominal sets are sets with an action of permutations of atoms; that is, nominal sets are sets with a symmetry action.

It turns out that this symmetry is key. 
For instance, we can define atoms-abstraction as a symmetric equivalence class (Definition~\ref{defn.abstraction.sets}, in this paper).
This class is `first-order' in flavour and does not involve the construction of a full function-space.

In PNL syntax permutative symmetry is reflected in the use of atoms and atoms-abstraction and by the permutations in terms (the interested reader could look at the `symmetry-based' definition of $\alpha$-equivalence in Definition~\ref{defn.aeq}).

In PNL \emph{derivation} this is reflected in the axiom rule (rule \rulefont{Ax} of Figure~\ref{Seq} in this paper).
This rule gives that $\phi\limp\pi\act\phi$ for all PNL predicates; so PNL predicates are fully symmetric up to permuting atoms.
This is impossible to model in HOL because functional abstraction is asymmetric: obviously, $\lam{x}\lam{y}x$ is not equal to $\lam{y}\lam{x}x$.
 
Truth in \emph{restricted PNL} (Figure~\ref{rSeq}) is also asymmetric; restricted PNL is therefore a weaker system.
It still has nominal semantics, but its predicates are not symmetric.
This is the logic that we translate to HOL.

The symmetry of full PNL might seem counterintuitive---especially if the reader is used to modelling names as functional arguments (or indeed as numbers).
We do not expect $\lam{x,y}P(x,y)$ to be symmetric with $\lam{y,x}P(x,y)$, or (for numbers) $x\leq y$ to be symmetric with $y\leq x$.

Consider the predicate $a=b$ in PNL (assume an equality, for the sake of argument).
This is \emph{false} because $a$ is a distinct atom from $b$.
If the reader is used to thinking of variables as things that `vary' then $x=y$ might be either false or true depending the values associated to $x$ and $y$.
Not so in PNL: at the level of PNL syntax atoms are not variables; $a$ and $b$ are distinct; and their distinctness is symmetric up to permuting atoms so that $b=a$, $c'=a$, and $d=e$ are also false.
If the reader is used to thinking of variables as numbers then there might exist some predicate $\leq$ that puts them in order.
Not so in PNL: such a predicate is forbidden.\footnote{This is because we used all finite permutations of atoms as our symmetry group.  Generalisations of this as suggested e.g. in \cite[Subsection~3.1]{klin:townc} are possible.  See also Remark~2.1.8 of \cite{gabbay:nomtnl}.}
The unknowns $X$ and $Y$ in PNL do vary, and they are variables.
More on this in Remark~\ref{rmrk.asym} and Subsection~\ref{subsect.pnl.ax}.

\subsection{Map of the paper}

This paper has a lot of technical ground to cover.
This is unavoidable, because we need to deal with two logics (restricted PNL and HOL) and two semantics (nominal sets, and the hand-crafted Henkin models in nominal renaming sets used in the completeness proof), as well as two translations (from logic to logic, and from models to models).

For the reader's convenience, we provide an overview of the main technical points with brief justifications for their design:
\begin{itemize*}
\item
Section~\ref{sect.pnl} introduces permissive-nominal logic.
This comes from previous work into `nominal' axiomatisations of systems with binding \cite{gabbay:pernl,gabbay:pernl-jv}.\footnote{Note that PNL is not only about nominal abstract syntax as considered in e.g. \cite{gabbay:newaas-jv,gabbay:fountl}.  Nominal abstract syntax is a denotation for syntax with binding.  PNL and its models are a (more general) syntax and semantics for denotations with binding in general, which are not all necessarily datatypes of abstract syntax.}

In fact, we need to introduce two logics: full PNL and also a \emph{restricted} version which has a weaker non-equivariant axiom rule.
We write the entailment relations $\cent$ and $\nopicent$ respectively.
It is the restricted version that we will eventually translate to HOL.
\item
Section~\ref{sect.hol} introduces higher-order logic as a theory over the syntax of the simply-typed $\lambda$-calculus.
We write the entailment relation $\holcent$.
\item
Section~\ref{sect.translation.sound} defines the translation from restricted PNL to HOL, and proves it sound using arguments on syntax.
In order to do the translation, we need to introduce a \emph{capture typing} ${D\cent r:A}$ which is a measure of how many functional abstractions are required to translate a given nominal term without losing information; that is, of the functional complexity of a nominal term. 
\item
Our goal is then to prove completeness of the translation.
We do this by transforming models of PNL into models of HOL.
So Section~\ref{sect.semantics} introduces two categories: \theory{PmsPrm} of permissive-nominal sets and \theory{PmsRen} of permissive-nominal renaming sets.
We also give a \emph{free} construction, transforming a permissive-nominal set into a permissive-nominal renaming set.
\item
In Section~\ref{sect.permissive-nominal.sets} we interpret full and restricted PNL in \theory{PmsPrm}.
In Section~\ref{sect.interp.hol} we interpret HOL in \theory{PmsRen}.
\item
Finally, in Section~\ref{sect.pnl.hol.complete} we use the free construction of Section~\ref{sect.semantics} to map a model of PNL in \theory{PmsPrm} to a model in \theory{PmsRen}, and because the free construction does not `make anything equal' this is sufficient to prove completeness.
\item
As one further mathematical note, the results in the literature concern full PNL and not restricted PNL.
So in Appendix~\ref{sect.completeness} we sketch proofs of soundness, cut-elimination, and completeness of restricted PNL with respect to non-equivariant models in \theory{PmsPrm}.
These are modest, if not entirely direct, modifications of the existing definitions and proofs for full PNL and equivariant models in \theory{PmsPrm}.
\end{itemize*}

Quite a number of new ideas are required to make this all work.
The highlights are: permissive-nominal renaming sets and their application to give non-standard `nominal' Henkin models for higher-order logic; restricted PNL and its semantics; the free construction; and the technical arguments as discussed in Section~\ref{sect.pnl.hol.complete}.

\subsection{Review of motivation}

Given that the proofs and constructions in this paper are non-trivial and involve an effort to extend existing machinery, we should pause to ask again why doing this is justified, even necessary.

Nominal techniques were designed originally to reason on syntax-with-binding (see the original journal paper \cite{gabbay:newaas-jv} or a recent survey paper \cite{gabbay:fountl}).
But since then this remit has expanded to reasoning about denotations with binding more generally (an overview of which is in \cite{gabbay:nomtnl}).
In doing this, we have created a whole new syntax and semantics for meta-mathematics.

We will not argue for or against either the nominal foundation or the higher-order foundation for mathematics.\footnote{There has been more than enough of that already, and anyway, because truth is free, proving theorems is never a zero sum game.}
Our question is: given that these two foundations exist, how do they relate?

In fact, questions have been asked about how nominal names and binding are related to functions, ever since nominal techniques were conceived in the second author's thesis.
Since then, the development of PNL \cite{gabbay:pernl-jv} and nominal renaming sets \cite{gabbay:nomrs} has given us two powerful new tools with which to address these questions: a proof-theory for a logic in which nominal reasoning so far can be formalised, and a visibly nominal semantics which is not based on permutations but on possibly non-bijective renamings on atoms, so that atoms-abstraction can be considered as a function in that semantics. 

In this paper, we leverage this to give a precise, concrete, and mathematically detailed account of how these two worlds really stand in relation to one another---and how they differ.
In conclusion we speculate that there is some potential (not explored in this paper) that our translations might be used to piggyback nominal techniques on the substantial implementational efforts that have gone into developing HOL over the past seventy years.


\section{Permissive-Nominal Logic}
\label{sect.pnl}

Permissive-nominal logic is a first-order logic for nominal terms quotiented by $\alpha$-equivalence.
Doing this is not entirely trivial; the interested reader can find more on this elsewhere \cite{gabbay:nomu-jv,gabbay:pernl,gabbay:pernl-jv,gabbay:nomtnl}.

\subsection{Syntax}

\begin{defn}
\label{defn.sort.sig}
A \deffont{sort-signature} is a pair $(\mathcal A,\mathcal B)$ of \deffont{name} and \deffont{base sorts}.
$\nu$ will range over name sorts; $\basesort$ will range over base sorts.
A \deffont{sort language} is then defined by
\begin{frameqn}
\alpha ::= \nu \mid (\alpha,\dots,\alpha) \mid [\nu]\alpha \mid \basesort 
.
\end{frameqn}
\end{defn}

\begin{rmrk}
Examples of base sorts are: `$\lambda$-terms',\ `formulae',\ `$\pi$-calculus processes',\ and `program environments', `functions', `truth-values', `behaviours',\ and `valuations'.

Examples of name sorts are `variable symbols',\ `channel names',\ or `memory locations'.

$[\nu]\alpha$ is an \emph{abstraction sort}.
This does a similar job to function-types in higher-order logic but note that $\nu$ must always be a name-sort.
The behaviour of a term of sort $[\nu]\alpha$ corresponds to `bind a name of sort $\nu$ in a term of sort $\alpha$'.
Such a term does not denote a function, though later on in our completeness proof we will deliberately undermine that intuition to obtain our completeness result. 
\end{rmrk}

\begin{defn}
\label{defn.atoms}
For each $\nu$ fix a disjoint countably infinite set of \deffont{atoms} $\atoms_\nu$, and an arbitrary bijection $f_\nu$ between $\atoms_\nu$ and the integers $\mathbb Z=\{0,\text{-}1,1,\text{-}2,2,\ldots\}$.
Write 
$$
\atomsdown_\nu=\{f_\nu(i)\mid i<0\}
\qquad
\atomsup_\nu=\{f_\nu(i)\mid i\geq 0\}.
$$
Finally, write 
$$
\atomsdown=\bigcup\atomsdown_\nu
\qquad
\atomsup=\bigcup\atomsup_\nu
\qquad
\mathbb A=\bigcup \mathbb A_\nu
$$ 
$a,b,c,\ldots$ will range over \emph{distinct} atoms (we call this the \deffont{permutative} convention).

A \deffont{permission set} has the form $(\atomsdown \cup A)\setminus B$ where $A\subseteq\atomsup$ and $B\subseteq\atomsdown$ are finite (and a permission set may be finitely represented by the pair $(A,B)$).
$S$, $T$, and $U$ will range over permissions sets. 
\end{defn}

The use of $\atomsdown$ and $\atomsup$ ensures that permission sets are infinite and also coinfinite (their complement is also infinite).

\begin{rmrk}
Permission sets are simple, but surprisingly subtle.
 
$\atomsdown$ and $\atomsup$ are reminiscent of some treatments of syntax where a formal distinction is made between `names that exist to be bound' and `names that exist to be free'.
See for instance the \emph{freie} and \emph{gebundene Gegenstansvariable} of Gentzen \cite[Section~1]{Gentzen:1935}, and the \emph{individual variables} and \emph{parameters} of Prawitz \cite[Section~1]{prawitz:natdpt}, or Smullyan \cite[Chapter~IV, Section~1]{smullyan:firol}. 

However, for any given atom there is no \emph{fixed} sense in which it is either capturable or not capturable.
Each permission set defines a world of capturable/non-capturable atoms. 
Furthermore, even once a permission set is fixed, we shall see that permutations can shift an atom from one to the other.
See Example~\ref{xmpl.aeq} and Remark~\ref{rmrk.capturing.sub} for examples of permission sets in action, controlling $\alpha$-equivalence and substitution respectively. 

We make $\atomsdown$ and $\atomsup$ both infinite so that we have inexhaustible supplies of both kinds of atom.\footnote{For comparison, nominal terms have only a finite supply of fresh atoms.  The effect of this is that the nominal terms of \cite{gabbay:nomu-jv} cannot be quotiented by $\alpha$-equivalence as primitive, and the freshness context may need to be extended dynamically with fresh names.  This introduces an `impure' flavour of state and sequentiality into the theory of nominal terms which is absent from the permissive-nominal version.  In short, making permission sets infinite and coinfinite makes the whole theory noticeably more `pure'.}
Further technical discussion of the advantages of permissive-nominal techniques is in \cite{gabbay:perntu-jv}.
\end{rmrk}

\begin{defn}
\label{defn.term.signature}
A \deffont{term-signature} over a sort-signature $(\mathcal A,\mathcal B)$ is a tuple $(\mathcal F,\mathcal P,\f{ar},\mathcal X)$ where:
\begin{itemize*}
\item
$\mathcal F$ and $\mathcal P$ are disjoint sets of \deffont{term-} 
and \deffont{proposition-formers}.

$\tf f$ will range over term-formers.
$\tf P$ will range over proposition-formers.
\item 
$\f{ar}$ assigns to each ${\tf f\in\mathcal F}$ a
\deffont{term-former arity} $(\alpha)\tau$
and to each $\tf P\in\mathcal P$ a \deffont{proposition-former arity}
$\alpha$, where $\alpha$ and $\tau$ are in the sort-language
determined by $(\mathcal A,\mathcal B)$.

We will write $((\alpha_1,\ldots,\alpha_n))\tau$ just as $(\alpha_1,\ldots,\alpha_n)\tau$.
\item
$\mathcal X$ is a set of \deffont{unknowns} $X$, each of which has a sort $\sort(X)$ and a permission set $\pmss(X)$, such that for each sort $\alpha$ and permission set $S$ the set $\{X\in\mathcal X\mid \sort(X)=\alpha,\ \pmss(X)=S\}$ is countably infinite.
$X,Y,Z$ will range over distinct unknowns.
\end{itemize*}
\label{defn.signature}
A \deffont{signature} $\mathcal S$ is then a tuple $(\mathcal A,\mathcal B,\mathcal F,\mathcal P,\f{ar},\mathcal X)$.
\end{defn}
We write $\tf f:(\alpha)\tau$ for $\f{ar}(\tf f)=(\alpha)\tau$ and similarly we write $\tf P:\alpha$ for $\f{ar}(\tf P)=\alpha$. 

\begin{xmpl}
\label{xmpl.lam.sig}
The signature for the $\lambda$-calculus from the Introduction has a name-sort for $\lambda$-calculus object-level variables $\nu$, a base sort $\iota$ for $\lambda$-terms, and appropriate term-formers: 
\begin{itemize*}
\item
$\tf{var}:(\nu)\iota$ to form $\lambda$-calculus variables in $\iota$ out of names in $\nu$, 
\item
$\tf{app}:(\iota,\iota)\iota$ for application, and 
\item
$\tf{lam}:([\nu]\iota)\iota$ taking an abstraction in $[\nu]\iota$ and forming from it a $\lambda$-abstraction term in $\iota$. 
\end{itemize*}
\end{xmpl}

\begin{frametxt}
\begin{defn}
\label{defn.permutation}
A \deffont{permutation} is a bijection $\pi$ on $\mathbb A$ such that $a\in\mathbb A_\nu\liff \pi(a)\in\mathbb A_\nu$ and $\f{nontriv}(\pi)=\{a\mid \pi(a)\neq a\}$ is finite.
Write $\mathbb P$ for the set of permutations.

Given $a,b\in\mathbb A_\nu$ let a \deffont{swapping} $(a\ b)$ be the bijection on atoms that maps $a$ to $b$, $b$ to $a$, and all other $c$ to themselves.
\end{defn}
\end{frametxt}

\begin{nttn}
\label{nttn.permutations}
We use the following notation:
\begin{itemize*}
\item
Write $\pi\circ\pi'$ for \deffont{functional composition}, so $(\pi\circ\pi')(a)=\pi(\pi'(a))$).
\item
Write $\id$ for the \deffont{identity permutation}, so $\id(a)=a$ always. 
\item
Write $\pi^\mone$ for \deffont{inverse}, so $\pi\circ\pi^\mone=\id$.
\end{itemize*}
\end{nttn}

\begin{defn}
\label{defn.pnl.syntax}
For each signature $\mathcal S$, define \deffont{terms} and \deffont{propositions} over $\mathcal S$ by: 
\begin{frameqn}
\begin{array}{c@{\qquad}c@{\qquad}c}
\begin{prooftree}
(a\in\mathbb A_\nu)
\justifies
a:\nu
\end{prooftree}
&
\begin{prooftree}
\rawr_1:\alpha_1 \ \ldots\ \rawr_n:\alpha_n
\justifies
(\rawr_1,\ldots,\rawr_n):(\alpha_1,\ldots,\alpha_n)
\end{prooftree}
&
\begin{prooftree}
\rawr:\alpha\quad (\f{ar}(\tf f)=(\alpha)\tau)
\justifies
\tf f(\rawr):\tau
\end{prooftree}
\\[4ex]
\begin{prooftree}
\rawr:\alpha\quad (a\in\mathbb A_\nu)
\justifies
[a]\rawr:[\nu]\alpha
\end{prooftree}
&
\begin{prooftree}
(\sort(X)=\alpha)
\justifies
\pi\act X:\alpha
\end{prooftree}
\\[4ex]
\begin{prooftree}
\phantom{h}
\justifies
\bot\text{ prop.}
\end{prooftree}
&
\begin{prooftree}
\rawphi\text{ prop.}\ \ \rawpsi\text{ prop.}
\justifies
\rawphi\limp\rawpsi\text{ prop.}
\end{prooftree}
&
\begin{prooftree}
\rawr:\alpha\ \ (\f{ar}(\tf P)=\alpha)
\justifies
\tf P(\rawr)\text{ prop.}
\end{prooftree}
\\[4ex]
\begin{prooftree}
\rawphi\text{ prop.}
\justifies
\Forall{X}\rawphi\text{ prop.}
\end{prooftree}
\end{array}
\end{frameqn}
\end{defn}

\begin{xmpl}
Continuing Example~\ref{xmpl.lam.sig}, we have the following terms and propositions:
\begin{itemize*}
\item
$\tf{var}(a):\iota$ where $a\in\mathbb A_\nu$.
\item
$[a]X:[\nu]\iota$ where $a\in\mathbb A_\nu$ and $\sort(X)=\iota$.
\item
$\tf{lam}([a]X):\iota$ since $[a]X:[\nu]\iota$ and $\tf{lam}:([\nu]\iota)\iota$. 
\item
$\Forall{X}\tf P(\tf{lam}([a]X),X)$ is a proposition if $\tf P$ is a proposition-former and $\tf P:(\iota,\iota)$.
\end{itemize*}
\end{xmpl}

\subsection{Permutation, substitution, and so on}

These definitions are all needed for the rest of the paper, starting with $\alpha$-equivalence in Subsection~\ref{subsect.aeq}.
We need them at both levels; both for atoms and for unknowns.

\begin{defn}
\label{defn.permutation.action}
Define a (level 1) \deffont{permutation action} on syntax by:
$$
\begin{array}{r@{\ }l@{\qquad}r@{\ }l}
\pi\act a\equiv& \pi(a)
&
\pi\act (\rawr_1,\ldots,\rawr_n) \equiv&  (\pi\act \rawr_1,\ldots,\pi\act \rawr_n)
\\
\pi\act [a]\rawr \equiv&  [\pi(a)]\pi\act \rawr
&
\pi\act(\pi'\act X) \equiv&  (\pi{\circ}\pi')\act X
\\
\pi\act \tf f(\rawr) \equiv&  \tf f(\pi\act \rawr)
\\
\pi\act\bot \equiv&  \bot
&
\pi\act (\rawphi\limp\rawpsi)\equiv&  (\pi\act \rawphi)\limp(\pi\act \rawpsi)
\\
\pi\act \tf P(\rawr)\equiv&  \tf P(\pi\act \rawr)
&
\pi\act (\Forall{X}\rawphi) \equiv&  \Forall{X}\pi\act\rawphi
\end{array}
$$
\end{defn}

\begin{defn}
\label{defn.permutation.action.2}
Let $\Pi$ range over sort- and permission-set-preserving bijections on unknowns 
(so $\sort(\Pi(X)){=}\sort(X)$ and $\pmss(\Pi(X)){=}\pmss(X)$)
such that $\{X\mid \Pi(X)\neq X\}$ is finite.

Write $\Pi\circ\Pi'$ for functional composition,\ $\Id$ for the identity permutation, and $\Pi^\mone$ for inverse, much as in Notation~\ref{nttn.permutations}.

Define a (level 2) \deffont{permutation action} by:
{
$$
\begin{array}{r@{\ }l@{\qquad}r@{\ }l}
\Pi\act a\equiv&  a
&
\Pi\act (\rawr_1,\ldots,\rawr_n) \equiv&  (\Pi\act \rawr_1,\ldots,\Pi\act \rawr_n)
\\
\Pi\act [a]\rawr \equiv&  [a]\Pi\act \rawr
&
\Pi\act(\pi\act X) \equiv&  \pi\act(\Pi(X))
\\
\Pi\act \tf f(\rawr) \equiv&  \tf f(\Pi\act \rawr)
\\
\Pi\act\bot \equiv&  \bot
&
\Pi\act (\rawphi\limp\rawpsi)\equiv&  (\Pi\act \rawphi)\limp(\Pi\act \rawpsi)
\\
\Pi\act \tf P(\rawr)\equiv&  \tf P(\Pi\act \rawr)
&
\Pi\act (\Forall{X}\rawphi) \equiv&  \Forall{\Pi(X)}\Pi\act\rawphi
\end{array}
$$
}
\end{defn}

\begin{rmrk}
A curious asymmetry between Definitions~\ref{defn.permutation.action} and~\ref{defn.permutation.action.2} is that $\pi\act(\Forall{X}\rawphi)\equiv \Forall{X}\pi\act\rawphi$ but $\Pi\act(\Forall{X}\rawphi)\equiv \Forall{\Pi(X)}\Pi\act\rawphi$.
Note that $\pi$ not applied to the binding occurrence of $X$, but $\Pi$ is.

In fact, we \emph{could} take $\pi\act(\Forall{X}\rawphi)=\Forall{\pi\act X}\pi\act\rawphi$.
We would have to complicate Definition~\ref{defn.pnl.syntax} by introducing $\Forall{\pi\act X}\rawphi$ as well-formed syntax, but we could do this.
However, it turns out that this would make no difference.
It turns out that $\Forall{\pi\act X}\rawphi$ and $\Forall{X}\rawphi$ are logically equivalent.

To gain a quick intuitive understanding of why this is so, bear in mind that the substitution `$X$ maps to $r$' maps $\pi\act X$ to $\pi\act r$ (this is made formal later, in Definition~\ref{defn.subst.action}) and in fact `$\pi\act X$ maps to $\pi\act r$' would map $X$ to $r$ (the interested reader can find a wider discussion of this in and around \cite[Remark~3.4.7]{gabbay:nomtnl}).

So, rather curiously, $\Forall{\pi\act X}\rawphi$ means the same thing and would receive the same denotation, regardless of $\pi$: this would be the denotation of the PNL proposition $\Forall{X}\rawphi$.
We can therefore simplify our syntax and take 
$\pi\act(\Forall{X}\rawphi)\equiv \Forall{X}\pi\act\rawphi$.

Something similar happens in the much more abstract semantic context of two-level nominal sets; see \cite[Lemma~2.25]{gabbay:twolns}.
Further comments on asymmetries between atoms and unknowns in PNL will follow in Remark~\ref{rmrk.asym}.
\end{rmrk}

\begin{defn}
\label{defn.pointwise}
Suppose $f$ is a function on a set $X$ and $U\subseteq X$.
Define $f\act U$ by
$$
f\act U=\{f(x)\mid x\in U\} .
$$
This is the standard \deffont{pointwise} action of a function on a set.
We use this for $\pi$ acting on sets of atoms, $\Pi$ acting on sets of unknowns, and (from Definition~\ref{defn.renaming}
 onwards) $\rho$ acting on sets of atoms.
\end{defn}

\begin{defn}
\label{defn.fa}
Define \deffont{free atoms} $\fa(\rawr)$ and $\fa(\rawphi)$ by:
$$
\begin{array}{r@{\ }l@{\quad}r@{\ }l@{\quad}r@{\ }l}
\fa(\pi\act X)=& \pi\act\pmss(X) 
&
\fa([a]\rawr)=& \fa(\rawr)\setminus\{a\}
&
\fa(a)=& \{a\}
\\
\fa(\tf f(\rawr)) =&  \fa(\rawr)
&
\fa((\rawr_1,\ldots,\rawr_n)) =& 
\bigcup\fa(\rawr_i)
&&
\\[1.5ex]
\fa(\bot) =& \varnothing
&
\fa(\rawphi\limp\rawpsi)=& \fa(\rawphi)\cup \fa(\rawpsi)
\\
\fa(\tf P(\rawr)) =&  \fa(\rawr)
&
\fa(\Forall{X}\rawphi)=& \fa(\rawphi) 
\end{array}
$$
Define \deffont{free unknowns} $\fU(r)$ and $\fU(\rawphi)$ by:
$$
\begin{array}{r@{\ }l@{\quad}r@{\ }l@{\quad}r@{\ }l}
\fU(a)=& \varnothing
&
\fU(\pi\act X)=& \{X\}
&
\fU(\tf f(\rawr)) =&  \fU(\rawr)
\\
\fU([a]\rawr)=& \fU(\rawr)
&
\fU((\rawr_1,\ldots,\rawr_n)) =&  
\bigcup\fU(\rawr_i)
\\[1.5ex]
\fU(\bot) =& \varnothing
&
\fU(\rawphi\limp\rawpsi)=& \fU(\rawphi)\cup \fU(\rawpsi)
\\
\fU(\tf P(\rawr)) =&  \fU(\rawr)
&
\fU(\Forall{X}\rawphi)=& \fU(\rawphi)\setminus\{X\} 
\end{array}
$$
\end{defn}

\begin{lemm}
\label{lemm.fa.pi.r}
$\fa(\pi\act \rawr)=\pi\act \fa(\rawr)$ and $\fa(\pi\act\rawphi)=\pi\act\fa(\rawphi)$.

Also,
$\fU(\Pi\act \rawr)=\Pi\act \fU(\rawr)$ and $\fU(\Pi\act\rawphi)=\Pi\act\fU(\rawphi)$.
\end{lemm}
\begin{proof}
By routine inductions on $\rawr$.
\end{proof}

\subsection{$\alpha$-equivalence}
\label{subsect.aeq}

The use of permissive-nominal terms allows us to `just quotient' syntax by $\alpha$-equivalence.
We can do this for both level 1 variable symbols (atoms) and level 2 variable symbols (unknowns).

\begin{defn}
Call a relation $\somerel$ on terms and on propositions a \deffont{congruence} when it is closed under the following rules:\footnote{We do not assume a congruence is an equivalence relation.  We prefer to keep the notion of `being preserved by all term-formers' orthogonal to the notion of `being transitive, reflexive, and symmetric'.  For instance, we want a rewrite relation to have the first property but not the second.  In \cite{dowek:dedm} the second author considers combining sequents with rewriting.  This is \emph{deduction modulo}; a nominal version of deduction modulo is future work and is one of the background motivations for the creation of PNL.}
$$
\begin{array}{c@{\qquad}c}
\begin{prooftree}
\rawr_i\somerel \raws_i\quad 1\leq i\leq n
\justifies
(\rawr_1,\ldots,\rawr_n)\somerel (\raws_1,\ldots,\raws_n)
\end{prooftree}
&
\begin{prooftree}
\rawr\somerel \raws\ \ (\tf f:(\alpha)\tau,\ \rawr,\raws:\alpha)
\justifies
\tf f(\rawr)\somerel\tf f(\raws)
\end{prooftree}
\\[3ex]
\begin{prooftree}
\rawr\somerel \raws
\justifies
[a]\rawr\somerel [a]\raws
\end{prooftree}
&
\begin{prooftree}
\rawphi\somerel\rawphi'\quad \rawpsi\somerel\rawpsi'
\justifies
\rawphi\limp\rawpsi\somerel \rawphi'\limp\rawpsi'
\end{prooftree}
\\[3ex]
\begin{prooftree}
\rawr\somerel \raws\quad (\tf P:\alpha,\ \rawr,\raws:\alpha)
\justifies
\tf P(\rawr)\somerel \tf P(\raws)
\end{prooftree}
&
\begin{prooftree}
\rawphi\somerel \rawphi'
\justifies
\Forall{X}\rawphi\somerel \Forall{X}\rawphi'
\end{prooftree}
\end{array}
$$
\end{defn}

\begin{defn}
\label{defn.aeq}
Write $(a\ b)$ for the \deffont{(level 1) swapping} permutation which maps $a$ to $b$ and $b$ to $a$ and all other $c$ to themselves.
Similarly, provided $\sort(X)=\sort(Y)$ and $\pmss(X)=\pmss(Y)$, write $(X\ Y)$ for the \deffont{(level 2) swapping}.

Define \deffont{$\alpha$-equivalence} $\aeq$ on terms and propositions to be the least equivalence relation that is a congruence and is such that:
\begin{frameqn}
\begin{array}{c@{\qquad}c}
\begin{prooftree}
(a,b\not\in\fa(\rawr))
\justifies
(b\ a)\act \rawr \aeq r
\end{prooftree}
&
\begin{prooftree}
(X,Y\not\in\fU(\rawphi))
\justifies
(Y\ X)\act\rawphi\aeq \rawphi
\end{prooftree}
\end{array}
\end{frameqn}
\end{defn}

\begin{rmrk}
Definition~\ref{defn.aeq} is inductive, but the reader familiar with nominal terms from e.g. \cite{gabbay:nomu-jv,gabbay:perntu-jv,gabbay:pernl-jv} might be familiar with a more syntax-directed inductive characterisation whose characteristic rules for atoms would look like this in our current notation:
$$
\begin{prooftree}
(b\not\in\fa(\rawr))
\justifies
[a]r \aeq [b](b\ a)\act r
\end{prooftree}
\qquad
\begin{prooftree}
(\Forall{a}\pi(a)\neq\pi'(a)\limp a\not\in\pmss(X))
\justifies
\pi\act X\aeq \pi'\act X
\end{prooftree}
$$ 
The form of Definition~\ref{defn.aeq} is more compact and more abstract.
It was introduced in \cite{gabbay:capasn,gabbay:forcie} (in particular see \cite[Lemma~3.2]{gabbay:forcie} and the discussion surrounding it); 
a detailed proof of an equivalence of the two presentations is in \cite[Theorem~2.31]{gabbay:capasn-jv} 
\end{rmrk}

\begin{xmpl}
\label{xmpl.aeq}
We illustrate Definition~\ref{defn.aeq}. 
Suppose $a,b,c,d:\nu$ and $X,Y:\tau$ for some name sort $\nu$ and base sort $\tau$.
Also suppose $a,b,c,d\not\in\pmss(Y)$ and suppose $b\not\in\pmss(X)$.
\begin{enumerate*}
\item
\emph{We $\alpha$-convert $a$ and $b$ in $[a][b]a$.}\quad

First we note that $b,d\not\in\fa([b]a)$ so by symmetry $[b]a\aeq (b\ d)\act[b]a=[d]a$, and by congruence $[a][b]a\aeq [a][d]a$.
Next we note that $a,c\not\in\fa([a][d]a)$ so by symmetry $[a][d]a\aeq(c\ a)\act[a][d]a=[c][d]c$.
We use transitivity.
\item 
\emph{We $\alpha$-convert $[a][a]b$ to $[c][d]b$.}\quad

We reason as follows: $[a][a]b\aeq [a][d]b\aeq[c][d]b$.
\item
\emph{We $\alpha$-convert $((a\ b)\circ(c\ d))\act Y$ to $Y$.}\quad
First we note that $a,b\not\in\fa((c\ d)\act Y)=(c\ d)\act\pmss(Y)$, so $((a\ b)\circ(c\ d))\act Y=(a\ b)\act((c\ d)\act Y)\aeq (c\ d)\act Y$.
Then we note that $c,d\not\in\fa(Y)$ so $(c\ d)\act Y\aeq Y$.
We use transitivity.
\item
\emph{We $\alpha$-convert $X$ and $a$ in $\Forall{X}\tf P([a]X)$.}

Using $(a\ b)$ and $(X\ Y)$ we deduce:
$$
\Forall{X}\tf P([a]X) \stackrel{(a\ b)}\aeq \Forall{X}\tf P([b](b\ a)\act X) 
\stackrel{(X\ Y)}\aeq \Forall{Y}\tf P([b](b\ a)\act Y) . 
$$  
It is routine to convert this sketch into a full derivation-tree.
\item
In \cite{gabbay:capasn,gabbay:oneaah} an axiom of substitution/$\beta$-equivalence was stated which, translated to a permissive-nominal syntax, is expressed thus:
\begin{tab1}
\rulefont{{\sm}ren} & (\lambda([a]X))b=&(b\ a)\act X& (b\not\in\pmss(X))
\end{tab1} 
Since $a,b\not\in\fa([a]X)$ we have $[a]X\aeq [b](b\ a)\act X$.
Writing $X'$ for an unknown with $\pmss(X')=(b\ a)\act\pmss(X)$ we can therefore rephrase this axiom as follows:
\begin{tab1}
\rulefont{{\sm}id} & (\lambda([b]X'))b=& X'& (b\not\in\pmss(X'))
\end{tab1}
This is the form that the same axiom took in \cite{gabbay:nomalc}; that paper used the slightly less flexible nominal terms framework, so the equivalence noted above had to be the subject of a proof \cite[Lemma~2.16]{gabbay:nomalc}.
\end{enumerate*}
\end{xmpl}
It is not hard to prove by inductive arguments that the definition of $\alpha$-equivalence here gives the same result as that found for instance in \cite[Subsection~2.4]{gabbay:pernl-jv} or \cite[Definition~2.16]{gabbay:pernl}.\footnote{In fact, this was the characterisation used to \emph{design} the `exotic' multi-level $\alpha$-equivalences of permissive-nominal logic or two-level nominal sets \cite{gabbay:twolns}.}

\begin{frametxt}
\begin{defn}
\label{defn.terms.and.propositions}
For each signature $\mathcal S$, we take terms and propositions quotiented by $\alpha$-equivalence. 
\end{defn}
\end{frametxt}

\begin{rmrk}
\label{rmrk.asym}
Atoms and unknowns both have permutation actions and both participate in $\alpha$-conversion and both `look like' variable symbols.
Indeed, in the translation to HOL of Subsection~\ref{subsect.pnl.to.hol} atoms and unknowns \emph{both} are translated to variables.
However, in PNL syntax atoms and unknowns are very distinct:
\begin{itemize*}
\item
Atoms populate only their own special name-sorts $\nu$.
Unknowns populate any sort.
\item
Unknowns depend on atoms in the sense that $\pmss(X)$ is a set of atoms.
Atoms do not depend on anything.
\item
Atoms and unknowns both get permuted in syntax but only unknowns also have a substitution action (defined next in Subsection~\ref{subsect.pnl.sub}).
Conversely, only atoms-permutation is explicit in the syntax, as $\pi\act X$ (there is no $\Pi$ in any $r$, only an action of $\Pi$ \emph{on} $r$). 
\item
If $\phi$ is proposition then $\Forall{X}\phi$ is a proposition, but $\Forall{a}\phi$ is never well-formed syntax.
\item
Later on when we build the denotation in Subsection~\ref{subsect.interpret.pnl.terms}, atoms are interpreted as themselves (so $\denot{\mathcal I}{\varsigma}{a} = a$) whereas unknowns are interpreted via a valuation, in typical Tarski style (so $\denot{\mathcal I}{\varsigma}{X} = \varsigma(X)$).
\end{itemize*}
Why this asymmetry?
\begin{itemize*}
\item
PNL is a first-order logic.
Therefore it has a syntactic class of variables which we call \emph{unknowns} $X$, with an unknowns-substitution action $[X\ssm r]$ (Definition~\ref{defn.Xssmr}), a universal quantifier $\forall X$, $\alpha$-conversion, and a denotation using valuations for unknowns $\denot{\mathcal I}{\varsigma}{X}=\varsigma(X)$ (Definition~\ref{defn.interpret.terms}).
\item
PNL treats atoms (names) as a special kind of symmetric data (i.e. data with a group action).
It has sorts for atoms $\nu$ which are populated by atoms-as-terms, a nominal-terms style explicit permutation $\pi\act X$ to permute atoms in terms, atoms-abstraction $[a]r$ to bind atoms in terms, and a nominal sets style semantics for atoms as themselves; so that $\denot{\mathcal I}{\varsigma}{a}=a$ and $\denot{\mathcal I}{\varsigma}{[a]r}=[a]\denot{\mathcal I}{\varsigma}{r}$ (Definition~\ref{defn.interpret.terms}).
\end{itemize*}
PNL is expressive enough to axiomatise a substitution action for atoms, and even a universal quantifier for atoms.
Thus we can make atoms behave like first-order logic variables, if we want to do this.
For more on these and other nominal theories see \cite{gabbay:capasn-jv} or Figures~4 and~5 of \cite{gabbay:pernl,gabbay:pernl-jv}.

The asymmetry noted above reflects the standard design of first-order logic: unknowns `live in' $\phi$ in the sense that they are used to make universal assertions $\Forall{X}\phi$, and atoms `live in' $r$ in the sense that they interact with the term- and sort-system and we can form $[a]r$ and $\tf f(a)$.

What can be a little confusing is that can mimic variables, and indeed, we can and do use PNL to axiomatise logic. 
In short: PNL is a first-order logic for axiomatising logics.
\end{rmrk}

\begin{rmrk}
We can still ask to what extent it might be possible to go beyond PNL and to `fold' unknowns back down into the syntax, possibly recursively, to obtain some language and/or semantics in which atoms and unknowns are just two aspects of a single well-founded structure.

In fact this idea predates PNL: the Lambda Context Calculus (LamCC) \cite{gabbay:lamcc,gabbay:lamcce} features just such an infinite hierarchy, but it has no semantic theory, no primitive notion of proposition, and has a weak notion of $\alpha$-equivalence.\footnote{In the LamCC, atoms are level 1 variables, unknowns are level 2 variables and behave much like `holes' in a program context, and there are level 3 variables, and so on.  Many interesting program constructs can be expressed in this language.} 
Concurrently with PNL, in two recent papers \cite{gabbay:twolns,gabbay:metvil} the first author has investigated these questions from a semantic perspective.
In \cite{gabbay:twolns} we develop an abstract notion of two-level nominal set that can directly interpret unknowns $X$ and atoms $a$ just as the nominal sets of this paper directly interpret atoms $a$.
In \cite{gabbay:metvil} we develop a concrete model of unknowns as infinite lists of distinct atoms. 
It would be future work to integrate these semantics into a new logic.\footnote{The second author is most interested in PNL as a first-order basic for specifying logic and computation in theorem-proving and also in undergraduate teaching.  Fancier logics and semantics than PNL can certainly be imagined, and perhaps created, but this does not undermine the interest of a `simple' logic, like PNL.} 
A PNL-like logic with more than two levels of variable and a perfect symmetry across levels, is certainly imaginable.
\end{rmrk}

\subsection{Substitution}
\label{subsect.pnl.sub}

\begin{frametxt}
\begin{defn}
A (level 2) \deffont{substitution} $\theta$ is a function from unknowns to terms such that:
\begin{itemize*}
\item
For all $X$, $\theta(X):\sort(X)$ and $\fa(\theta(X))\subseteq \pmss(X)$.
\item
$\theta(X)\equiv \id\act X$ for all but finitely many $X$.
\end{itemize*}
$\theta$ will range over substitutions.
\end{defn}
\end{frametxt}

\begin{defn}
Define $\f{nontriv}(\theta)$ by:
$$
\f{nontriv}(\theta)\equiv 
\{X\mid \theta(X){\not\equiv} \id\act X \text{ or } X{\in}\fU(\theta(Y))\text{ for some }Y\} 
$$
\end{defn}
$\f{nontriv}(\theta)$ is unknowns that can be produced or consumed by $\theta$, other than in the trivial manner that $\theta(X)\equiv\id\act X$.

\begin{defn}
\label{defn.subst.action}
Define a \deffont{substitution action} by:
\begin{frameqn}
\begin{array}{r@{\ }l@{\qquad}r@{\ }l}
a\theta\equiv&  a
&
(r_1,\ldots,r_n)\theta\equiv&  (r_1\theta,\ldots,r_n\theta)
\\
([a]r)\theta\equiv&  [a](r\theta)
&
(\pi\act X)\theta\equiv&  \pi\act \theta(X)
\\
\tf f(r)\theta\equiv&  \tf f(r\theta)
\\
\bot\theta\equiv& \bot
&
(\phi\limp\psi)\theta\equiv&  (\phi\theta)\limp\psi\theta
\\
(\tf P(r))\theta\equiv&  \tf P(r\theta)
&
(\Forall{X}\phi)\theta \equiv&  \Forall{X}(\phi\theta) \quad (X\not\in\f{nontriv}(\theta))
\end{array}
\end{frameqn}
\end{defn}

One kind of substitution will be particularly useful, starting with \rulefont{\forall L} in Figure~\ref{Seq}:
\begin{defn}
\label{defn.Xssmr}
Suppose $X:\alpha$ and $r:\alpha$ and $\fa(r)\subseteq\pmss(X)$.
Define $[X\ssm r]$ by:
$$
[X\ssm r](X)=r
\qquad\qquad
[X\ssm r](Y)=Y\quad\text{all other}\ Y
$$
\end{defn}
It is easy to verify that $[X\ssm r]$ is indeed a substitution.

\begin{rmrk}
\label{rmrk.capturing.sub}
Famously, nominal terms substitution is capturing \cite[Definition~2.13]{gabbay:nomu-jv}.
We spell out how this works in our permissive-nominal context:
Suppose $a\in\pmss(X)$ and $b\not\in \pmss(X)$ (where we assume appropriate sorts).
Then:
\begin{itemize*}
\item
$([a]X)[X\ssm a]=[a]a$.
The $a$ in the substitution $[X\ssm a]$ has been captured by the $[a]X$.
\item
$([b]X)[X\ssm a]=[b]a$.
\item
It is impossible to even ask what $([b]X)[X\ssm b]$ is equal to because $[X\ssm b]$ is not a substitution, since $b\not\in \pmss(X)$.
So $b\not\in \pmss(X)$ cannot be captured by a substitution $[X\ssm b]$, because that substitution does not exist.
\item
Also, $[b](b\ a)\act X=[a]X$.
By construction, 
$$
([b](b\ a)\act X)[X\ssm a]=[b](b\ a)\act a=[b]b=[a]a.
$$
So the choice of representative of $[a]X$ does not matter for capture to occur.
\end{itemize*}
In \cite{gabbay:metvil} we propose a view of $X$ as a well-ordering on its permission set; that is, we identify $X$ literally with an infinite list of atoms.
Viewed from this perspective, the nominal substitution action is not capturing at all: it is simply a compact way to present an `infinite raising' or `infinite Skolemisation' (cf. Remark~\ref{rmrk.skolem}), or a de Bruijn index \cite{bruijn:lamcnn}.
This idea underlies also the translation to HOL which we construct later in Definition~\ref{defn.translation}.
\end{rmrk}

\subsection{Sequents and derivability}

\begin{defn}
\label{defn.seq}
$\Phi$ and $\Psi$ will range over sets of propositions.
We may write $\phi,\Phi$ and $\Phi,\phi$ as shorthand for $\{\phi\}\cup\Phi$ (where we do not insist that $\phi\not\in\Phi$, that is, the union need not be disjoint). 
\begin{itemize*}
\item
A \deffont{sequent} of restricted PNL is a pair $\Phi\nopicent\Psi$.
\item
A \deffont{sequent} of full PNL is a pair $\Phi\cent\Psi$.
\end{itemize*}
Write 
$\fU(\Phi,\Psi)=\bigcup\{\fU(\phi)\mid \phi\in\Phi\}\cup\bigcup\{\fU(\psi)\mid\psi\in\Psi\}$.
\end{defn}

\begin{frametxt}
\begin{defn}[Derivable sequents]
Define the \deffont{derivable sequents} of full PNL and restricted PNL by the rules in Figures~\ref{Seq} and~\ref{rSeq} respectively.
\end{defn}
\end{frametxt}
\noindent The sole difference between Figures~\ref{Seq} and~\ref{rSeq} is in the axiom rule, and is highlighted with a light blue rectangle. 

\begin{figure*}[t!]
$$
\begin{array}{c@{\qquad}c}
\shademath{\begin{prooftree}
\phantom{h}
\justifies
\Phi,\,\phi\cent \pi\act\phi,\,\Psi
\using\rulefont{Ax}
\end{prooftree}}
&
\begin{prooftree}
\phantom{h}
\justifies
\Phi,\,\bot\cent \Psi
\using\rulefont{\bot L}
\end{prooftree}
\\[4ex]
\begin{prooftree}
\Phi\cent \phi,\,\Psi
\quad
\Phi,\,\psi\cent \Psi
\justifies
\Phi,\,\phi\limp\psi\cent\Psi
\using\rulefont{{\limp}L}
\end{prooftree}
&
\begin{prooftree}
\Phi,\,\phi\cent \psi,\,\Psi
\justifies
\Phi\cent \phi\limp\psi,\,\Psi
\using\rulefont{{\limp}R}
\end{prooftree}
\\[4ex]
\begin{prooftree}
{
\begin{array}{c}
\Phi,\,\phi[X\ssm r]\cent \Psi
\\
(\fa(r){\subseteq}\pmss(X), 
\ r{:}\sort(X))
\end{array}
}
\justifies
\Phi,\,\Forall{X}\phi\cent \Psi
\using\rulefont{{\forall}L}
\end{prooftree}
&
\begin{prooftree}
\Phi\cent \phi,\,\Psi\quad {\small (X\not\in\fU(\Phi,\Psi))}
\justifies
\Phi\cent \Forall{X}\phi,\,\Psi
\using\rulefont{{\forall}R}
\end{prooftree}
\\[5ex]
\end{array}
$$
\caption{Sequent derivation rules of full Permissive-Nominal Logic}
\label{Seq}
\end{figure*}
\begin{figure*}[t!]
$$
\begin{array}{c@{\qquad}c}
\shademath{\begin{prooftree}
\phantom{h}
\justifies
\Phi,\,\phi\nopicent \phi,\,\Psi
\using\rulefont{Ax^{\nopi}}
\end{prooftree}}
&
\begin{prooftree}
\phantom{h}
\justifies
\Phi,\,\bot\nopicent \Psi
\using\rulefont{\bot L}
\end{prooftree}
\\[4ex]
\begin{prooftree}
\Phi\nopicent \phi,\,\Psi
\quad
\Phi,\,\psi\nopicent \Psi
\justifies
\Phi,\,\phi\limp\psi\nopicent\Psi
\using\rulefont{{\limp}L}
\end{prooftree}
&
\begin{prooftree}
\Phi,\,\phi\nopicent \psi,\,\Psi
\justifies
\Phi\nopicent \phi\limp\psi,\,\Psi
\using\rulefont{{\limp}R}
\end{prooftree}
\\[4ex]
\begin{prooftree}
{
\begin{array}{c}
\Phi,\,\phi[X\ssm r]\nopicent \Psi
\\
(\fa(r){\subseteq}\pmss(X), 
\ r{:}\sort(X))
\end{array}
}
\justifies
\Phi,\,\Forall{X}\phi\nopicent \Psi
\using\rulefont{{\forall}L}
\end{prooftree}
&
\begin{prooftree}
\Phi\nopicent \phi,\,\Psi\quad {\small (X\not\in\fU(\Phi,\Psi))}
\justifies
\Phi\nopicent \Forall{X}\phi,\,\Psi
\using\rulefont{{\forall}R}
\end{prooftree}
\\[5ex]
\end{array}
$$
\caption{Sequent derivation rules of restricted Permissive-Nominal Logic}
\label{rSeq}
\end{figure*}


\begin{nttn}
We may write $\Phi\nopicent\Psi$ as shorthand for `$\Phi\nopicent\Psi$ is a derivable sequent'.
We may write $\Phi\not\nopicent\Psi$ as shorthand for `$\Phi\nopicent\Psi$ is not a derivable sequent'.

Similarly for $\Phi\cent\Psi$ and $\Phi\not\cent\Psi$.
\end{nttn}

\subsection{Discussion of the PNL axiom rule}
\label{subsect.pnl.ax}

Figure~\ref{Seq} is the logic of \cite{gabbay:pernl-jv,gabbay:nomtnl}.
Figure~\ref{rSeq} is the logic we translate to HOL in this paper.
The only difference is the `$\pi$' in the axiom rule: full PNL has it (see \rulefont{Ax}), and restricted PNL does not (see \rulefont{Ax^\nopi}).
Restricted PNL is a subset of full PNL, in the sense that (obviously) $\Phi\nopicent\Psi$ implies $\Phi\cent\Psi$ (this suggests that the models of restricted PNL should be a superset of those of full PNL, which will indeed turn out to be the case; see Appendix~\ref{sect.completeness}).

Why the difference?  
Because the translation to HOL identifies atoms with functional arguments.
Atoms are symmetric up to permutation in full PNL; this is built into \rulefont{Ax} in Figure~\ref{Seq}.
Functional arguments are typically not symmetric.

We might try to translate full PNL to HOL by translating $n!$ permutation instances of each $r$ or $\phi$, where $n$ is some notion of the number of atoms in $r$ or $\phi$ (cf. \emph{capture typings} in Definition~\ref{defn.capture.typing}); but that would be `cheating' in the sense that most of the syntax would then be generated by a meta-level `macro' which does $n!$ amount of work.
The issue here is not whether PNL can be encoded in HOL; the issue is whether it can be cleanly translated into HOL. 
These are related but distinct questions.

To quickly see the difference in derivational power between full and restricted PNL, assume a name sort $\nu$, a proposition-former $\tf P:\nu$, and two atoms $a,b:\nu$.
Then the difference in the entailment relations of PNL and restricted PNL can be summed up as follows:
\begin{itemize*}
\item
$\tf P(a)\cent \tf P(a)$\ and\ $\tf P(a)\cent \tf P(b)$.
\item
$\tf P(a)\nopicent \tf P(a)$\ but\ \sout{$\tf P(a)\nopicent \tf P(b)$}.
\end{itemize*}
Note that not even full PNL can derive that $Q(a,b)$ entails $Q(a,a)$; we can permute, but we have to permute in the \emph{entire} proposition.
So for instance if we axiomatise the syntax and derivability of first-order logic in PNL then we would have a predicate $\tf{entails}$ and we might prove $\tf{entails}(\tf P(a),\tf P(a))$.
By equivariance of full PNL we could also derive $\tf{entails}(P(b),\tf P(b))$.\footnote{This has actually been used in real proofs; starting with \cite{gabbay:frelog} we used what amounts to \rulefont{Ax} to rename variable symbols in inductive proofs. See \cite[Subsection~4.2]{gabbay:fountl} for this \emph{equivariance} principle discussed as a practical tool to obtain one-line proofs at the rigorous but informal meta-level of papers.}
However, we still cannot derive \sout{$\tf{entails}(\tf P(a),\tf P(b))$}, and if we could then that would be wrong.

In Appendix~\ref{sect.completeness} we see that this difference corresponds in models to proposition-formers being interpreted by equivariant functions (for full PNL) or not necessarily equivariant functions (for restricted PNL).

It has to be this way:
Definition~\ref{defn.translation} translates PNL terms and predicates to HOL terms and predicates.
In Lemma~\ref{lemm.it.has.to.be} we illustrate why only restricted PNL can be translated to HOL by our translation: the derivability of full PNL is too strong for HOL derivability and the translation would not be sound.

Note that this does not prove that other translations to HOL do not exist, but (as the discussion of $n!$ above suggests) we speculate that they would be significantly less natural.

\section{HOL syntax and derivability}
\label{sect.hol}

Higher-order logic (HOL) syntax and derivability should be familiar \cite{miller:logho,farmer:sevvst,andrews:intmlt,church:forstt}.
We give the basics.

\subsection{Syntax}

We present HOL as a derivation system over simply-typed $\lambda$-terms with constants and types for logical reasoning (like a type of truth-values and constant symbols like $\limp$ and $\forall$).
This is all standard.

\begin{defn}
\label{defn.hol.sort.sig}
A \deffont{HOL signature} is a set $\mathcal D$ of \deffont{base types}, which includes a distinguished base type of \deffont{truth-values} $o\in\mathcal D$.
$\basetype$ will range over base types.
A \deffont{type-language} is defined by
\begin{frameqn}
\beta ::= \basetype \mid 
(\beta,\ldots,\beta) \mid \beta\to\beta
.
\end{frameqn}
\end{defn}
It is not necessary to include products $(\beta_1,\ldots,\beta_n)$, but for the purposes of translating PNL into HOL doing this is convenient.

\begin{defn}
\label{defn.hol.term.signature}
A \deffont{term-signature} over a HOL signature $\mathcal D$ is a tuple $(\mathcal G,\type)$ where:
\begin{itemize*}
\item
$\mathcal G$ is a set of \deffont{constants}, which must contain elements $\bot$, $\limp$, and $\forall_\beta$ for every type $\beta$.
\item $\type$ assigns to each ${\tf g\in\mathcal G}$ a
type $\beta$
in the type-language
determined by $\mathcal D$, such that $\type(\bot)=o$, $\type(\limp)=o\to o\to o$, and $\type(\forall_\beta)=(\beta\to o)\to o$.
\end{itemize*}
A \deffont{signature} $\mathcal T$ is then a tuple $(\mathcal D,\mathcal G,\type)$.
\end{defn}

\begin{nttn}
We write $\tf g:\beta$ for $\type(\tf g)=\beta$.
\end{nttn}

\begin{rmrk}
We strongly deprecate referring to the constant $\forall_\beta:(\beta\to o)\to o$ as `higher-order abstract syntax' (HOAS), as sometimes happens.  
That term should refer to inductive types with binding constructed using constants of higher type like $(\Lambda\to\Lambda)\to\Lambda$ (strong HOAS) or $(\nu\to\Lambda)\to\Lambda$ (weak HOAS) \cite{despeyroux94higherorder,pfenning:hoas} (the study of inductive syntax with binding is not the topic of this paper). 

A term $\forall_\beta:(\beta\to o)\to o$ (plus axioms) expresses the \emph{meaning} of $\forall$ \cite[Section~2]{church:forstt} and would still have meaning if our syntax was, e.g. combinators.   In contrast, the \emph{syntax} of combinators could be represented without any need for higher-order syntax, since it does not have binders \cite[Section~2]{hindley:lamcci}.
\end{rmrk}

\begin{defn}
\label{defn.hol.terms.sorts}
For each signature $\mathcal T=(\mathcal D,\mathcal G,\type)$ and each type $\beta$ over $\mathcal D$ fix a countably infinite set of \deffont{variables} of that type.

$X,Y,Z$ will range over distinct HOL variables.\footnote{This means that if the reader sees `$X$' this could refer either to a HOL variable or---recalling Definition~\ref{defn.term.signature}---to a PNL unknown.  
We will make sure that it is always clear from context which is meant.} 
Write $\type(X)$ for the type of $X$.
\end{defn}

\begin{defn}
For each signature $\mathcal T$ define \deffont{HOL terms} over $\mathcal T$ by
$$
t::= X \mid \lam{X}t \mid tt \mid (t,\ldots,t) \mid \tf g
$$
and a \deffont{typing} relation by: 
\begin{frameqn}
\begin{array}{c@{\qquad}c@{\qquad}c@{\qquad}c}
\begin{prooftree}
\rawt:\beta\ \ (\type(X){=}\beta')
\justifies
\lam{X}\rawt:\beta'{\to}\beta
\end{prooftree}
&
\begin{prooftree}
\rawt':\beta'{\to}\beta\quad \rawt:\beta'
\justifies
\rawt'\rawt:\beta
\end{prooftree}
& 
\begin{prooftree}
\rawt_1:\beta_1 \ \ldots\ \rawt_n:\beta_n
\justifies
(\rawt_1,\ldots,\rawt_n):(\beta_1,\ldots,\beta_n)
\end{prooftree}
&
\begin{prooftree}
(\type(\tf g){=}\mu)
\justifies
\tf g:\mu
\end{prooftree}
\end{array}
\end{frameqn}
\end{defn}

We now define $\alpha$-equivalence.
We would not normally be so detailed about this, but when we map PNL terms and propositions to HOL later, it will be useful to have been precise here:
\begin{defn}
\label{defn.hol.perm}
A \deffont{permutation} of HOL variables is a bijection $\varpi$ such that $\f{nontriv}(\varpi)=\{X\mid \varpi(X)\neq X\}$ is finite.
Give HOL terms a permutation action $\varpi\act t$ defined by:
\maketab{tab4}{R{6.5em}@{}C{1em}@{}L{6em}@{\quad}R{4em}@{}C{1em}@{}L{5em}@{\quad}R{3em}@{}C{1em}@{}L{7em}}
\begin{tab4}
\varpi\act X&=&\varpi(X)
&
\varpi\act \lam{X}t&=&\lam{\varpi(X)}\varpi\act t
&
\varpi\act (t't)&=&(\varpi\act t')(\varpi\act t)
\\
\varpi\act(t_1,\dots,t_n)&=&(\varpi\act t_1,\dots,\varpi\act t_n)
&
\varpi\act \tf g&=&\tf g
\end{tab4}
Free variables are defined by:
\begin{tab4}
\fv(X)&=&\{X\}
&
\fv(\lam{X}t)&=&\fv(t)\setminus\{X\}
&
\fv(t't)&=&\fv(t')\cup\fv(t)
\\
\fv((t_1,\dots,t_n))&=&\bigcup_i\fv(t_i)
&
\fv(\tf g)&=&\varnothing
\end{tab4}
Call a relation $\somerel$ on HOL terms a \deffont{congruence} when it is closed under the following rules:
$$
\begin{prooftree}
t\somerel u
\justifies
\lam{X}t\somerel \lam{X}u
\end{prooftree}
\qquad
\begin{prooftree}
t'\somerel u'\quad t\somerel u
\justifies
t't\somerel u'u
\end{prooftree}
\qquad
\begin{prooftree}
t_i\somerel u_i \quad (1\leq i\leq n)
\justifies
(t_1,\dots,t_n)\somerel (u_1,\dots,u_n)
\end{prooftree}
$$
Define $\alpha$-equivalence to be the least congruence that is an equivalence relation and is such that:
$$
\begin{prooftree}
(X,Y\not\in\fv(t))
\justifies
(Y\ X)\act t\aeq t
\end{prooftree}
$$ 
We quotient terms by $\alpha$-equivalence and define \deffont{capture-avoiding substitution} $\rawt[X\ssm u]$ as usual.
\end{defn}

\begin{defn}
We write $\rawt{\,:\,}\beta$ for \emph{$\rawt$ is a term and has type $\beta$}.
We call $\rawt$ \deffont{typable} when $\rawt:\beta$ for some type $\beta$.

We call a term a \deffont{HOL proposition} when it has type $o$.
$\xi$ and $\chi$ will range over HOL propositions.
We may write $\forall_\beta\,\lam{X}\xi$ as $\Forall{X}\xi$. 
\end{defn}
 
\begin{defn}
\label{defn.hol.seq}
$\Xi$ and $\Chi$ will range over sets of HOL propositions. 
We may write $\xi,\Xi$ and $\Xi,\xi$ as shorthand for $\{\xi\}\cup\Xi$. 
 
Write 
$\fv(\Xi,\Chi)=\bigcup\{\fv(\xi)\mid \xi\in\Xi\}\cup\bigcup\{\fv(\chi)\mid\chi\in\Chi\}$.

A \deffont{sequent} is a pair $\Xi\holcent\Chi$.
\end{defn}

\begin{frametxt}
\begin{defn}[Derivable sequents]
The \deffont{derivable sequents} are defined in Figure~\ref{hol.Seq}. 
\end{defn}
\end{frametxt}

\begin{figure*}[t]
$$
\begin{array}{c@{\qquad}c}
\begin{prooftree}
\phantom{h}
\justifies
\Xi,\,\xi\holcent \xi,\,\Chi
\using\rulefont{hAx}
\end{prooftree}
&
\begin{prooftree}
\phantom{h}
\justifies
\Xi,\,\bot\holcent \Chi
\using\rulefont{h\bot L}
\end{prooftree}
\\[4ex]
\begin{prooftree}
\Xi\holcent \xi,\,\Chi
\quad
\Xi,\,\chi\holcent \Chi
\justifies
\Xi,\,\xi\limp\chi\holcent\Chi
\using\rulefont{h{\limp}L}
\end{prooftree}
&
\begin{prooftree}
\Xi,\,\xi\holcent \chi,\,\Chi
\justifies
\Xi\holcent \xi\limp\chi,\,\Chi
\using\rulefont{h{\limp}R}
\end{prooftree}
\\[4ex]
\begin{prooftree}
\Xi,\,\xi[X\ssm \rawt]\holcent \Chi
\quad
(\rawt{:}\type(X))
\justifies
\Xi,\,\Forall{X}\xi\holcent \Chi
\using\rulefont{h{\forall}L}
\end{prooftree}
&
\begin{prooftree}
\Xi\holcent \xi,\,\Chi\quad {\small (X\not\in\fv(\Xi,\Chi))}
\justifies
\Xi\holcent \Forall{X}\xi,\,\Chi
\using\rulefont{h{\forall}R}
\end{prooftree}
\\[3ex]
\end{array}
$$
\caption{Sequent derivation rules of Higher-Order Logic}
\label{hol.Seq}
\end{figure*}

\section{The translation from nominal to functional syntax, and its soundness}
\label{sect.translation.sound}

\subsection{Translation from PNL to higher-order logic}
\label{subsect.pnl.to.hol}

In this subsection we show how to translate a PNL signature $\mathcal S$ and propositions and terms in that signature, to a higher-order logic (HOL) signature and propositions and terms in that signature.
We start by translating a PNL signature $\mathcal S$ to a HOL signature $\mathcal T_{\mathcal S}$.
First, we set up some notation:

\begin{nttn}
\label{nttn.finite.lists}
Let $D$ range over finite lists of distinct atoms.
\begin{itemize*}
\item
Write $a\in D$ when $a$ occurs in $D$.
\item
Write $D'\subseteq D$ when every element in $D'$ occurs in $D$ (disregarding order).
Similarly if $S$ is a set of atoms write $D\subseteq S$ when every element in $D$ occurs in $S$.
\item
If $S$ is a set of atoms write $D\cap S$ for the list obtained by removing from $D$ just those atoms not in $S$.
\item
Write $\pi\act D$ for the list obtained by applying $\pi$ pointwise to the elements of $D$ in order.
\item
Write $D,a$ for the list obtained by appending $a$; when we write this we include an assumption that $a\not\in D$.
\item
Write $\lam{D}t$ for $\lam{d_1}\dots\lam{d_n}t$ where $D=[d_1,\dots,d_n]$.
\end{itemize*}
\end{nttn}

\begin{defn}
\label{defn.TS}
\label{defn.hol.translation}
From a PNL signature $\mathcal S$ determine a HOL signature $\mathcal T_{\mathcal S}$ by the following specification:
\begin{enumerate*}
\item
For every atoms-sort $\nu$ in $\mathcal S$ assume a HOL base type $\mu_\nu$.
\item
For every base sort $\tau$ assume a HOL type $\mu_\tau$.
\end{enumerate*}
Translate sorts in $\mathcal S$ to types in $\mathcal T_{\mathcal S}$ as follows:
\begin{frameqn}
\begin{array}{r@{\ }l@{\qquad}r@{\ }l@{\qquad}r@{\ }l}
\hol{}{\nu}=&\mu_\nu
&
\hol{}{\tau}=&\mu_\tau
&
\hol{}{(\alpha_1,\ldots,\alpha_n)}=&(\hol{}{\alpha_1},\cdots,\hol{}{\alpha_n})
\\
\hol{}{[\nu]\alpha}=&\mu_\nu\to\hol{}{\alpha}
\end{array}
\end{frameqn}
\begin{enumerate*}
\setcounter{enumi}{2}
\item
For every term-former $\tf f:(\alpha)\tau$ assume a HOL constant $\tf g_{\smtf f}:\hol{}{\alpha}\to\tau$.
\item
For every proposition-former $\tf P:\alpha$ assume a HOL constant $\tf g_{\smtf P}:\hol{}{\alpha}\to o$.
\item
For every atom $a:\nu$ assume a HOL variable $a:\mu_\nu$.

It is convenient to assume this correspondence is a literal identity; i.e. that $\mathbb A_\nu$ is actually a subset of the set of HOL variables of type $\mu_\nu$, and that there are countably infinitely many HOL variables of type $\mu_\nu$ that are not atoms. 

In particular, this means that every permutation $\pi$ in the sense of Definition~\ref{defn.permutation} is also a permutation $\varpi$ in the sense of Definition~\ref{defn.hol.perm}.
\item
For every unknown $X:\alpha$ and list $D$ assume a distinct HOL variable $\XD$ that is not an atom\footnote{So $X$ is one of the countably infinitely many HOL variables that are not atoms.} of type $\mu_{\nu_1}{\to}\dots{\to}\mu_{\nu_n}{\to}\hol{}\alpha$ where $\nu_i$ is the sort of the $i$th atom in $\GammaX$ (by convention and as standard, if $\GammaX$ is empty we take this to be $\hol{}{\alpha}$).
\end{enumerate*}
\end{defn}

\begin{frametxt}
\begin{defn}
\label{defn.translation}
Given a list $D$ translate PNL terms and propositions in $\mathcal S$ to HOL terms and propositions in $\mathcal T_{\mathcal S}$ (Definition~\ref{defn.TS}) by the rules in Figure~\ref{fig.hol.translation}.
\end{defn}
(The notation $\pi\act(\GammaX)$ is defined in Notation~\ref{nttn.finite.lists}.)
\end{frametxt}

\begin{figure}[t]
$$
\begin{array}{r@{\ }l@{\qquad}r@{\ }l@{\qquad}r@{\ }l}
\hol{D}{a}=&a
&
\hol{D}{(r_1,\ldots,r_n)}=&(\hol{D}{r_1} ,\ldots,\hol{D}{r_n} )
&
\hol{D}{\tf f(r)}=&\tf g_{\smtf f}\, \hol{D}{r} 
\\
\hol{D}{[a]r}=&\lam{a}\hol{D}{r}
&
\hol{D}{\pi\act X}=&\XD\pi\act(\GammaX)
\\
\hol{D}{\bot}=&\bot
&
\hol{D}{\phi\limp\psi}=&{\limp}\hol{D}{\phi}\hol{D}{\psi}
&
\hol{D}{\tf P(r)}=&\tf g_{\smtf P}\, \hol{D}{r}
\\
\hol{D}{\Forall{X}\phi}=&\forall\, \lam{X_D}\hol{D}{\phi}
\end{array}
$$
\caption{Translation from PNL to HOL}
\label{fig.hol.translation}
\end{figure}

\begin{xmpl}
\label{xmpl.why.capturable}
Suppose $\GammaX$ (Notation~\ref{nttn.finite.lists}) is the list $[a]$. 
Assume a proposition-former $\tf{equal}$ of appropriate arity. 
Then: 
$$
\begin{gathered}
\hol{D}{\id\act X}=\XD\,a
\quad
\hol{D}{(b\ a)\act X}=\XD\,b
\quad
\hol{D}{[a]\id\act X}=\lam{a}(\XD\,a)
\quad
\hol{D}{[b](b\ a)\act X}=\lam{b}(\XD\,b)
\\
\hol{D}{\Forall{X}\tf{equal}([a]X,[b](b\ a)\act X)}=\forall_{\hol{}{\type(\XD)}}\,\lam{X}(\tf{equal}(\lam{a}(\XD\,a))(\lam{b}(\XD\,b)))
\end{gathered}
$$
Assuming appropriate axioms for $\tf{equal}$, we would expect this to be true.
Now assume $\GammaY$ is the list $[a,b]$. 
Then:
$$
\begin{gathered}
\hol{D}{\id\act Y}=Y_D ab
\quad
\hol{D}{(b\ a)\act Y}=Y_D ba
\quad
\hol{D}{[a]\id\act Y}=\lam{a}(Y_D ab)
\quad
\hol{D}{[b](b\ a)\act Y}=\lam{b}(Y_D ba)
\\
\hol{D}{\Forall{Y}\tf{equal}([a]Y,[b](b\ a)\act Y)}=\forall_{\hol{}{\type(\YD)}}\,\lam{Y_D}(\tf{equal}(\lam{a}(Y_D ab))(\lam{b}(Y_D ba)))
\end{gathered}
$$
We would expect this to be false.  
What has changed with respect to the previous case, is that $b$ is fresh for $X$ but not for $Y$.
\end{xmpl}

\begin{rmrk}
\label{rmrk.skolem}
The translation of $\pi\act X$ to $\XD\pi\act(\GammaX)$ can be seen as a \emph{Skolemisation} or \emph{raising}, where $D$ is a finite collection of atoms/variables that are considered `relevant'.

See for instance Subsections~4.2 and~4.3 of Chapter~1 of \cite{gabbay:hantm} or Definition~53 of Chapter~3 of \cite{gabbay:hantm}, where Skolemisation is discussed in the context of tableau methods.
There, just as here, there is a notion of the `relevant' variables.
See also Section~5 of \cite{miller:uniump}, where \emph{raising} is discussed in the context of unification in the presence of mixed quantifiers.

In fact, Definition~\ref{defn.translation} is a modified version of \cite[Definition~8.3]{gabbay:perntu-jv} which translated between nominal unification and unification over a generalisation of the same notion of \emph{pattern} used by Miller in \cite{miller:logpll,miller:uniump}.
See also \cite{gabbay:unialt}, where similar ideas are applied to algebraic reasoning.
The idea of these `nominal' translations is that $X$ translates to a variable $\XD$ of higher order.
In Proposition~\ref{prop.capturable.minimal} we state a formal sense in which, provided $D$ contains all `relevant' atoms, this translation loses no information (the result itself is from \cite{gabbay:perntu-jv}).
What does `relevant' mean? 
We examine Figure~\ref{fig.capture.typings} and see that it means intuitively `is permuted by some $\pi$ acting on $\pmss(X)$'. 
In Subsection~\ref{subsect.trans.soundness} we apply this to PNL. 
\end{rmrk}

\begin{lemm}
\label{lemm.hol.gamma.fa}
\begin{itemize*}
\item
Suppose $a$ is an atom.
Then if $a\in\fv(\hol{D}{r})$ then $a\in\fa(r)$.
\item
$\hol{D}{\pi\act r}=\pi\act\hol{D}{r}$ (for $\pi$ on the right-hand side considered as a permutation of HOL variables).
\end{itemize*}
As a corollary, the translation $\hol{D}{r}$ is well-defined.
That is, if $r$ and $s$ are $\alpha$-equivalent then $\hol{D}{r}=\hol{D}{s}$.
\end{lemm}
\begin{proof}
By routine inductions on $r$.
The proof that $\fv(\hol{D}{\pi\act X})\subseteq\fa(\pi\act X)$ uses the fact that $\GammaX\subseteq\pmss(X)$.
The corollary follows; for more details see \cite[Section~8]{gabbay:perntu-jv}.
\end{proof}

\subsection{Capture typing}

We mentioned in Remark~\ref{rmrk.skolem} that the translation of Definition~\ref{defn.translation} requires us to declare a finite list $D$ of `relevant' atoms.
How large must $D$ be in order to capture all the important information in some $r$ or $\phi$?
This is calculated by a \emph{capture typing}, an idea going back to \cite{gabbay:perntu,gabbay:perntu-jv}.
\begin{defn}
\label{defn.capture.typing}
Define \deffont{capture typings} $D\cent r:A$ and $D\cent\phi:A$ inductively by the rules in Figure~\ref{fig.capture.typings}.
Here $D$ ranges over finite lists of distinct atoms as described in Notation~\ref{nttn.finite.lists}, and $A$ ranges over finite sets of atoms.

If $A=\varnothing$ then we may omit the `${:}A$' and write just $D\cent r$ and $D\cent\phi$.
Write $D\cent\Psi$ when $D\cent\psi$ for every $\psi\in\Psi$.
\end{defn}

\begin{rmrk}
We are only interested in the case $A=\varnothing$, but we need to consider nonempty $A$ just as part of the inductive definition.
Intuitively, $D\cent r:A$ can be read as `$D$ is relevant to $r$ in the context of abstractions over the atoms in $A$'.
\end{rmrk}

\begin{figure}
$$
\begin{gathered}
\begin{prooftree}
\phantom{h}
\justifies
D\cent a:A
\end{prooftree}
\qquad\qquad 
\begin{prooftree}
D\cent r:A
\justifies
D\cent \tf f(r):A
\end{prooftree}
\qquad\qquad 
\begin{prooftree}
D\cent r:A,a
\justifies
D\cent [a]r:A
\end{prooftree}
\\[3ex]
\begin{prooftree}
D\cent r_i:A\quad (1{\leq}i{\leq}n)
\justifies
D\cent (r_1,\ldots,r_n):A
\end{prooftree}
\qquad\qquad 
\begin{prooftree}
((\f{nontriv}(\pi)\cup A)\cap\pmss(X)\subseteq D)
\justifies
D\cent \pi\act X:A
\end{prooftree}
\\[3ex]
\begin{prooftree}
D\cent r:A
\justifies
D\cent\tf P(r):A
\end{prooftree}
\quad\qquad 
\begin{prooftree}
D\cent\phi:A\quad D\cent\psi:A
\justifies
D\cent\phi\limp\psi:A
\end{prooftree}
\quad\qquad 
\begin{prooftree}
\phantom{h}
\justifies
D\cent\bot:A
\end{prooftree}
\quad\qquad 
\begin{prooftree}
D\cent \phi:A
\justifies
D\cent\Forall{X}\phi:A
\end{prooftree}
\end{gathered}
$$
\caption{Capture typing}
\label{fig.capture.typings}
\end{figure}

\begin{rmrk}
\label{rmrk.capture.closed}
As we mentioned in Remark~\ref{rmrk.skolem}, the intuitive reading of $D\cent r$ is `the atoms in $D$ get permuted in some $X$ occurring in $r$'.

Thus, the interesting case in Figure~\ref{fig.capture.typings} is the rule for $\pi\act X$.
This ensures that $D$ is large enough to record all the important atoms in $\pi$ or abstracted further up in the term---that is, those permitted in $X$---so that we do not lose information when we form $\hol{D}{\pi\act X}=\XD\pi\act(\GammaX)$.
This is made formal in Proposition~\ref{prop.capturable.minimal}, which is Theorems~8.12 and~8.14 of \cite{gabbay:perntu-jv}: 
\end{rmrk}

\begin{prop}
\label{prop.capturable.minimal}
\begin{itemize*}
\item
If $D\cent r$ and $D\cent s$ then $\hol{D}{r}=\hol{D}{s}$ implies $r=s$ (note that $=$ denotes $\alpha$-equality, because we quotiented terms by this relation), and similarly for $\phi$ and $\psi$.
\item
If $D\not\cent r$ then there exists $s$ such that $\hol{D}{r}=\hol{D}{s}$ yet $r\neq s$, and similarly for $\phi$. 
\end{itemize*}
\end{prop}

Definition~\ref{defn.translation} maps PNL terms and predicates to typable HOL terms:
\begin{prop}
\label{prop.typable.hol.gamma.r}
If $r:\alpha$ then for any $D$,\ $\hol{D}{r}:\hol{}{\alpha}$,
\ and $\hol{D}{\phi}:o$.
\end{prop}
\begin{proof}
By inductions on $r$ and $\phi$.
\begin{itemize*}
\item
\emph{The case $a\in\mathbb A_\nu$.}\quad
By Definition~\ref{defn.hol.translation} $a:\mu_\nu$. 
\item
\emph{The case $[a]r$ where $a\in\mathbb A_\nu$.}\quad
By inductive hypothesis $\hol{D}{r}:\beta$ for some type $\beta$.
It follows that $\hol{D}{[a]r}=\lam{a}\hol{D}{r}:\mu_\nu\to\beta$.
\item
\emph{The case $\pi\act X$.}\quad
Suppose $D\cent \pi\act X$.
It is routine to check that $\XD\pi\act(\GammaX):\hol{}{\sort(X)}$.
\qedhere\end{itemize*}
\end{proof}

\subsection{Soundness of the translation}
\label{subsect.trans.soundness}

Recall that HOL terms have a permutation action $\pi\act t$ given by considering $\pi$ as a permutation on HOL variables and using Definition~\ref{defn.hol.perm}.
Then:
\begin{lemm}
\label{lemm.hol.pi}
If $\f{nontriv}(\pi)\cap\fv(t)\subseteq D$ then $(\lam{D}t)\pi\act D\abeq\pi\act t$ (see Notation~\ref{nttn.finite.lists}).
\end{lemm}
\begin{proof}
A fact of $\alpha\beta$-conversion \cite[Lemma~9.2]{gabbay:perntu-jv}.
\end{proof}

\begin{lemm}
\label{lemm.hol.sub}
Suppose $D\cent r$ and $D\cent\phi$.
Suppose $r':\sort(X)$ and $\fa(r')\subseteq\pmss(X)$.
Then:
\begin{itemize*}
\item
$\hol{D}{r[X\ssm r']} \abeq \hol{D}{r}[\XD\ssm \lam{(\GammaX)}\hol{D}{r'}]$.
\item
$\hol{D}{\phi[X\ssm r']} \abeq \hol{D}{\phi}[\XD\ssm \lam{(\GammaX)}\hol{D}{r'}]$.
\end{itemize*}
\end{lemm}
\begin{proof}
By routine inductions on $r$ and $\phi$.
We sketch two cases:
\begin{itemize}
\item
\emph{The case $(\pi\act X)[X\ssm r']$.}\quad
We must prove that 
$$
\hol{D}{\pi\act r'}\abeq \bigl(\lam{(\GammaX)}\hol{D}{r'}\bigr)\pi\act(\GammaX).
$$
This follows by Lemmas~\ref{lemm.hol.gamma.fa} and~\ref{lemm.hol.pi}.
\item
\emph{The case $\tf P(r)[X\ssm r']$.}\quad
We must prove that 
$$
\hol{D}{\tf P(r[X\ssm r'])}\abeq \tf g_{\smtf P}(\hol{D}{r})[X\ssm\lam{(\GammaX)}\hol{D}{r'}].
$$
This follows directly from the first part.
\qedhere\end{itemize}
\end{proof}

\begin{prop}
\label{prop.hol.forall.sound}
Suppose $D\cent\phi$ and $D\cent r'$ and $r':\sort(X)$ and $\fa(r')\subseteq\pmss(X)$.

Then $\hol{D}{\Forall{X}\phi}\holcent \hol{D}{\phi[X\ssm r']}$.
\end{prop}
\begin{proof}
Using Lemma~\ref{lemm.hol.sub} and \rulefont{h\forall L} from Figure~\ref{hol.Seq}.
\end{proof}

\begin{lemm}
\label{lemm.enough}
For any finite set of predicates $\{\phi_i\mid i\in I\}$ there exists some $D$ such that $D\cent\phi_i$ for every $i\in I$.
\end{lemm}
\begin{proof}
We calculate $\bigcup\nontriv(\pi)$ for every $\pi$ occurring in every $\phi_i$, in some order, and take this to be $D$.
It is not hard to verify that this sufficies. 
\end{proof}

\begin{frametxt}
\begin{thrm}
\label{thrm.soundness}
The interpretation is sound: if $\Phi\nopicent\Psi$ with a derivation ${\mathfrak B}$ and $D\cent\Phi'$ and $D\cent\Psi'$ for every sequent $\Phi'\nopicent\Psi'$ appearing in ${\mathfrak B}$, then $\hol{D}{\Phi}\holcent\hol{D}{\Psi}$ (by Lemma~\ref{lemm.enough}, some such $D$ always exists).
\end{thrm}
\end{frametxt}
\begin{proof}
By induction on the derivation ${\mathfrak B}$ of $\Phi\nopicent\Psi$.
It is routine to verify by induction on ${\mathfrak B}$ that $\hol{D}{\Phi'}\holcent\hol{D}{\Psi'}$ is derivable for each $\Phi'\nopicent\Psi'$ appearing in ${\mathfrak B}$; the case of \rulefont{\forall R} uses Proposition~\ref{prop.hol.forall.sound}.
So in particular $\hol{D}{\Phi}\holcent\hol{D}{\Psi}$.
\end{proof}

\begin{lemm}
\label{lemm.it.has.to.be}
The interpretation for full PNL (Figure~\ref{Seq}, with the stronger axiom rule) would not be sound.
That is, there exist $\Phi$ and $\Psi$ and $D$ such that $D\cent\Phi$, $D\cent\Psi$, and $\Phi\cent\Psi$, but $\hol{D}{\Phi}\not\holcent\hol{D}{\Psi}$.
\end{lemm}
\begin{proof}
Consider a name sort $\nu$ and a unary predicate $\tf P:\nu$.
Then $\tf P(a)\cent\tf P(b)$ in full PNL, but it is not the case that $\tf g_{\tf P}\,a \cent\tf g_{\tf P}\,b$ in HOL. 
\end{proof}

\section{Semantics}
\label{sect.semantics}

For the reader's convenience we will clarify one aspect of the coming notation now: if the reader sees $\ns X$ this is a set with a permutation action; if the reader sees $\rs X$ this is a set with a renaming action.
There is no particular connection between $\ns X$ and $\rs X$.

A typical renaming is $[a\ssm b]$ (instead of a typical permutation $(a\ b)$).
Formal definitions are in Definition~\ref{defn.permutation} and~\ref{defn.renaming}.

The reader may not be surprised by the use of sets with a permutation action---nominal techniques are based on these \cite{gabbay:newaas-jv}.
But why the renaming action?
We need renamings to make a function out of an atoms-abstraction, mirroring the clause $\hol{D}{[a]r}=\lam{a}\hol{D}{r}$ in Definition~\ref{defn.translation}.

In PNL models, an abstraction $[a]r$ is modelled as Gabbay--Pitts atoms-abstraction $[a]x$, a sets-based construction from \cite{gabbay:newaas-jv} (Definition~\ref{defn.abstraction.sets}, in this paper).
This is constructed like a pair, forming $[a]x$ from $a$ and $x$, but destructed like a \emph{partial function} the graph of which is evident in Definition~\ref{defn.abstraction.sets}.
It is defined on $[a]x$ for fresh $b$ but not for $b\in\supp(x)\setminus\{a\}$.

When we translate $[a]r$ to HOL we interpret $[a]r$ as a function using $\lambda$-abstraction.
This suggests of our models that we translate a \emph{partial} function $[a]x$ to a total function. 
But then we have to give meaning to $[a]x$ applied to $b$ where $b$ is not fresh.
This is where renaming sets are used.

We can then conclude by noting that every model of PNL can be transformed into a model of HOL, and in a compositional manner (Lemma~\ref{lemm.commuting.square}).
Completeness quickly follows.

\subsection{Categories of supported permutation and renaming sets}

\subsubsection{Permutation and renaming sets}

\begin{defn}
\label{defn.renaming}
Suppose $\rho$ is a map from $\mathbb A$ to $\mathbb A$.
Define $\dom(\rho)$ and $\img(\rho)$ by
$$
\dom(\rho)=\{a\mid \rho(a)\neq a\}
\quad\text{and}\quad
\img(\rho)=\{\rho(a)\mid a\in\dom(\rho)\} .
$$
Echoing Definition~\ref{defn.permutation}, a \deffont{renaming} is a map $\rho$ from $\mathbb A$ to $\mathbb A$ such that $a\in\mathbb A_\nu\liff \rho(a)\in\mathbb A_\nu$ and 
$\f{nontriv}(\rho)=\dom(\rho)\cup\img(\rho)$ is finite. 
Write $\mathbb R$ for the set of renamings.

For $a,b\in\mathbb A_\nu$ let an \deffont{atomic renaming} $[a\ssm b]$ map $a$ to $b$, $b$ to $b$, and other $c$ to themselves.

$\rho$ will range over renamings.
\end{defn}

\begin{nttn}
\label{nttn.renamings}
We use the following notation (analogously to Notation~\ref{nttn.permutations}):
\begin{itemize*}
\item
Write $\rho\circ\rho'$ for \deffont{functional composition}, so $(\rho\circ\rho')(a)=\rho(\rho'(a))$).
\item
Write $\id$ for the \deffont{identity renaming}, so $\id(a)=a$ always.
\end{itemize*}
\end{nttn}
Unlike in Notation~\ref{nttn.permutations}, we have no notation for inverse, because for renamings inverses need not exist.

\begin{frametxt}
\begin{defn}
\label{defn.perm.set}
\begin{itemize*}
\item
A \deffont{permutation set} is a pair $\ns X=(|\ns X|,\act)$ of an \deffont{underlying set} $|\ns X|$ and a \deffont{permutation action} $(\mathbb P\times|\ns X|)\to |\ns X|$ which is a group action; write it infix.

(So $\id\act x=x$ and $\pi\act(\pi'\act x)=(\pi\circ\pi')\act x$.) 
\item
A \deffont{renaming set} is a pair $\rs X=(|\rs X|,\act)$ of an \deffont{underlying set} $|\rs X|$ and a \deffont{renaming action} $(\mathbb R\times|\rs X|)\to |\rs X|$ which is a monoid action; write it infix.

(So $\id\act x=x$ and $\rho\bigact(\rho'\bigact x)=(\rho\circ\rho')\bigact x$.) 
\end{itemize*}
\end{defn}
\end{frametxt}

\begin{defn}
\label{defn.finsupp}
\begin{itemize*}
\item
Suppose $\ns X$ is a permutation set.
Say that $A\subseteq \mathbb A$ \deffont{supports} $x\in|\ns X|$ when for all $\pi,\pi'\in\mathbb P$, if $\Forall{a\in A}\pi(a)=\pi'(a)$ then $\pi\act x=\pi'\act x$.
\item
Suppose $\rs X$ is a renaming set.
Say that $A\subseteq \mathbb A$ \deffont{supports} $x\in|\rs X|$ when for all $\rho,\rho'\in\mathbb R$, if $\Forall{a\in A}\rho(a)=\rho'(a)$ then $\rho\bigact x=\rho'\bigact x$.
\end{itemize*}
\end{defn}

\begin{lemm}
If $x\in |\ns X|/|\rs X|$ has a supporting permission set (Definition~\ref{defn.atoms}) then it has a unique least supporting set which is equal to the intersection of all permission sets supporting $x$.
We call this the \deffont{support} of $x$ when it exists, and write it $\supp(x)$.
\end{lemm}

\begin{frametxt}
\begin{defn}
\begin{itemize*}
\item
Call $x\in |\ns X|/|\rs X|$ \deffont{supported} when $\supp(x)$ exists.
\item
Call $\ns X$/$\rs X$ \deffont{supported} when every element $x\in|\ns X|/|\rs X|$ is supported. 
\end{itemize*}
\end{defn}
\end{frametxt}

Recall from Definition~\ref{defn.pointwise} the \emph{pointwise} actions of $\pi$ and $\rho$ on sets of atoms. 
\begin{lemm}
\label{lemm.supp.subsets}
\begin{enumerate*}
\item
If $x\in|\ns X|$ then $\supp(\pi\act x)=\pi\act\supp(x)$.
\item
If $x\in|\rs X|$ then $\supp(\rho\bigact x)\subseteq\rho\act\supp(x)$.

As a corollary, if $\rho$ is injective on $\supp(x)$ then $\supp(\rho\bigact x)=\rho\act\supp(x)$.
\end{enumerate*}
\end{lemm}
\begin{proof}
By routine calculations using the group/monoid action.
\end{proof}

The reverse subset inclusion in part~2 of Lemma~\ref{lemm.supp.subsets} is not true in general: 
\begin{lemm}
There exists a renaming set $\rs X$, an element $x\in|\rs X|$, and a renaming $\rho$ such that $\rho\act\supp(x)\not\subseteq\supp(\rho\bigact x)$. 
\end{lemm}
\begin{proof}
Consider $(\mathbb A\times\mathbb A)\cup\{\ast\}$ with the non-standard \deffont{`exploding'} renaming action such that:
$$
\begin{array}{r@{\ }l@{\qquad}r@{\ }l@{\quad}l}
\rho(\ast)=&\ast
&
\rho\bigact(a,b)=(\rho(a),\rho(b))&\text{if }\rho(a)\neq\rho(b)
\\
\rho\bigact(a,a)=&(\rho(a),\rho(a))
&
\rho\bigact(a,b)=&\ast&\text{if }\rho(a)=\rho(b)
\end{array}
$$
(Recall from Definition~\ref{defn.atoms} that by our permutative convention, $a$ and $b$ are distinct.)
Then $\supp([a\ssm b]\bigact (a,b))=\varnothing\subsetneq \{b\}=[a\ssm b]\bigact\supp((a,b))$.
\end{proof}

\subsubsection{Equivariant elements and maps}

\begin{defn} 
\label{defn.equivariant.element}
Call an element $x$ in $|\ns X|/|\rs X|$ \deffont{equivariant} when $\supp(x)=\varnothing$.
\end{defn}
$x$ is equivariant when $\pi\act x=x$ for all $\pi$, or $\rho\bigact x=x$ for all $\rho$, respectively.

\begin{defn}
\label{defn.equivariant}
\begin{itemize*}
\item
Call a function $F\in |\ns X|\to|\ns Y|$ \deffont{equivariant} when 
$$
\Forall{\pi{\in}\mathbb P}\Forall{x{\in}|\ns X|}F(\pi\act x)=\pi\act F(x).
$$ 
\item
Call a function $G\in |\rs X|\to|\rs Y|$ \deffont{equivariant} when 
$$
\Forall{\rho{\in}\mathbb R}\Forall{x{\in}|\rs X|}G(\rho\bigact x)=\rho\bigact G(x). 
$$
\end{itemize*}
$F$ and $G$ will range over equivariant functions between pairs of permutation and renaming sets respectively.
\end{defn}

\begin{rmrk}
Equivariance is a characteristic feature of nominal techniques.
Equivariance means in words `symmetric under permuting atoms'; in this paper we are also interested in `symmetric under renaming atoms'.
Whichever meaning is appropriate, equivariance is a symmetry property.

For an element $x\in|\ns X|$ or $x\in|\rs X|$ equivariance means that $x$ has empty support (in fact, $\supp(x)$ is a measure of asymmetry in $x$).
For a function $F\in|\ns X|\to|\ns Y|$ equivariance means that $F$ commutes with $\pi$.
For a function $F\in|\rs X|\to|\rs Y|$ equivariance means that $F$ commutes with $\rho$.

Equivariance and support are abstract mathematical concepts, but they were originally derived from the study of syntax in \cite{gabbay:newaas-jv}.
For syntax, equivariance corresponds to `closed'; in usual informal usage the syntax $\lam{x}x$ is closed and is symmetric under changing $x$ to $y$.\footnote{For the rest of this remark $\lambda$ is a term-former.  So $\lam{x}x$ refers to the term, not the function.}
The reader might therefore find it useful to read `equivariant' as `closed'.

However, this is just an analogy and it can be quite treacherous.  
For instance, $\lam{x}\lam{x}x$ is closed as a term but it is not \emph{permutatively} symmetric with $\lam{x}\lam{y}y$.
(If we build our syntax using nominal abstract syntax then $\lambda [a]\lambda [a]a$ is equal to $\lambda [a]\lambda [b]b$, but the reason for this is that $[a]a$ is equal, via symmetry up to permutation, with $[b]b$.)

So, while an analogy of `equivariance' with `closure' is historically reasonable and may be helpful, bear in mind that what it really means is `symmetric'.
\end{rmrk}

\begin{lemm}
\label{lemm.equivar.reduces.supp}
\begin{enumerate*}
\item
Suppose $F\in |\ns X|\to|\ns Y|$ is equivariant.
Then $\supp(F(x))\subseteq \supp(x)$ for every $x\in|\ns X|$. 
\item
Suppose $G\in |\rs X|\to|\rs Y|$ is equivariant.
Then $\supp(G(x))\subseteq \supp(x)$ for every $x\in|\rs X|$. 
\end{enumerate*}
\end{lemm}
\begin{proof}
We consider only the second part.
Suppose $S$ supports $x$ so that for all $\rho$ and $\rho'$, if $\Forall{a\in S}\rho(a)=\rho'(a)$ then $\rho\bigact x=\rho'\bigact x$.
The result follows if we note that $\rho\bigact G(x)=G(\rho\bigact x)$ and $\rho'\bigact G(x)=G(\rho'\bigact x)$.
\end{proof}

\begin{frametxt}
\begin{defn}
\label{defn.fps}
\begin{itemize*}
\item
Write \theory{PmsPrm} for the category with objects supported permutation sets and arrows equivariant functions between them.

Henceforth, $\ns X$ and $\ns Y$ will range over objects in \theory{PmsPrm}.
\item
Write \theory{PmsRen} for the category with objects supported renaming sets and arrows equivariant functions between them.

Henceforth, $\rs X$ and $\rs Y$ will range over objects in \theory{PmsPrm}.
\end{itemize*}
\end{defn} 
\end{frametxt}

\begin{rmrk}
Both \theory{PmsPrm} and \theory{PmsRen} are categories of sets with a monoid action (in the case of \theory{PmsPrm} that monoid happens to be a group).

\theory{PmsPrm} can be thought of as the category of pullback-preserving presheaves on the category $\mathbb I'$ of permission sets and finite injections between them (so an object $S\in\mathbb I'$ is a permission set, and an arrow from $S$ to $T$ is a permutation $\pi$ such that $\pi\act S\subseteq T$).
\theory{PmsRen} can be thought of as the category of those presheaves on the category $\mathbb F'$ of permission sets and finite renamings between them (so an object $S\in\mathbb F'$ is a permission set, and an arrow from $S$ to $T$ is a renaming $\rho$ such that $\rho\bigact S\subseteq T$) that preserve pullbacks of monos.

For details on this see \cite{gabbay:nomrs}, and for a more wide-ranging survey of the applications of sets with an action and presheaves see \cite{gadducci:abopap}.
See also the discussion of presheaves in Subsection~\ref{subsect.pnl.in.perspective}.
\end{rmrk}

\subsection{The exponential in \theory{PmsRen}}
\label{subsect.exp}

\theory{PmsPrm} and \theory{PmsRen} are both cartesian closed, but we only discuss exponentials for \theory{PmsRen} in this paper.
The reader can find the constructions for \theory{PmsPrm} e.g. in \cite[Section~9]{gabbay:fountl}.

\theory{PmsPrm} is used to give denotation to PNL only, while \theory{PrmRen} is used to give a denotation to PNL and also to HOL. 
For this reason, the exponentials of \theory{PmsRen} are of specific and immediate importance to us, but not those of \theory{PmsPrm}.

\subsubsection{Functions}

Recall the definitions of $\dom$ and $\img$ from Definition~\ref{defn.renaming}.
\begin{frametxt}
\begin{defn}
\label{defn.exp.ren}
\begin{itemize*}
\item
Suppose $\ns X,\ns Y\in\theory{PmsPrm}$.
Suppose $f\in |\ns X|\to|\ns Y|$ ($f$ is not necessarily equivariant).

Call $f$ \deffont{supported} when there exists a permission set $S_f\subseteq\mathbb A$ such that for every $x\in |\ns X|$ and permutation $\pi\in\mathbb P$, if $\nontriv(\pi)\cap S_f=\varnothing$ then
$$
\pi\bigact(f(x)) = f(\pi\bigact x) .
$$
\item
Suppose $\rs X,\rs Y\in\theory{PmsRen}$.
Suppose $f\in |\rs X|\to|\rs Y|$ ($f$ is not necessarily equivariant).

Call $f$ \deffont{supported} when there exists a permission set $S_f\subseteq\mathbb A$ such that for every $x\in |\rs X|$ and renaming $\rho\in\mathbb R$, if $\dom(\rho)\cap S_f=\varnothing$ then
$$
\rho\bigact(f(x)) = f(\rho\bigact x) .
$$
\end{itemize*}
\end{defn}
\end{frametxt}

\begin{rmrk}
Definition~\ref{defn.exp.ren} uses a word `supported' for $f$, suggestive of Definition~\ref{defn.finsupp}, even though $f$ has no permutation/renaming action.
It \emph{will} have a permutation/renaming action (Remark~\ref{rmrk.conj.action} and Definition~\ref{defn.exp.ren.action}), and then the terminologies will coincide (see Lemma~\ref{lemm.OK}).
\end{rmrk}

\begin{rmrk}
\label{rmrk.conj.action}
It is a fact that \theory{PmsPrm} is cartesian closed and functions have the \emph{conjugation action} 
$$
\label{Conjugation action} (\pi\act f)(x)=\pi\act(f(\pi^\mone\act x)).
$$
and $f$ is supported in the sense of Definition~\ref{defn.exp.ren} if and only if it is supported as an element of $|\ns X|\to|\ns Y|$ with the conjugation action.
For more on this see \cite{gabbay:fountl,gabbay:newaas-jv}.

Renamings $\rho$ are not invertible, so we must work a little harder to define a renaming action.
This is Definition~\ref{defn.exp.ren.action}.
However, the end result is similar to the conjugation action, in a sense made formal in Lemma~\ref{lemm.renaming.distribute} which is similar to an immediate corollary of the conjugation action that $\pi\act (f(x)) =(\pi\act f)(\pi\act x)$.
\end{rmrk}

\begin{lemm}
\label{lemm.supp.supported.f.bound}
If $f$ is supported then $\supp(f(x))\subseteq S_f\cup\supp(x)$ for every $x\in|\rs X|$.
\end{lemm}
\begin{proof}
By contradiction.
Suppose there exists $a\in \supp(f(x))\setminus(S_f\cup\supp(x))$.
Choose $b$ fresh (so $b\not\in\supp(f(x))\cup S_f\cup\supp(x)$).
Then $(b\ a)\bigact (f(x))=f((b\ a)\bigact x)$ since $a,b\not\in S_f$
and $f((b\ a)\bigact x)=f(x)$ since $b,a\not\in\supp(x)$.
It follows by Lemma~\ref{lemm.supp.subsets} that $(b\ a)\bigact\supp(f(x))=\supp(f(x))$, which is impossible.
\end{proof}

\begin{defn}
\label{defn.freshening.pair}
Suppose $S\subseteq\mathbb A$ is a permission set and $A\subseteq\mathbb A$ is finite.
Call $\rho_1$ and $\rho_2$ a \deffont{freshening pair} of renamings for $A$ with respect to $S$ when:
\begin{itemize*}
\item
$\dom(\rho_1)=A$ and $\dom(\rho_2)=\img(\rho_1)$.
\item
$(\rho_2\circ\rho_1)(a)=a$ for all $a\in A$.
\item
$\dom(\rho_2)\cap (S\cup A)=\varnothing$.
\end{itemize*}
\end{defn}
In words, $\rho_1$ maps the atoms in $A$ to be outside $S$ (and $A$), and $\rho_2$ is an `inverse' to $\rho_1$ that puts them back.

\subsubsection{Renaming action}

\begin{defn}
\label{defn.exp.ren.action}
(We continue the notation of Definition~\ref{defn.exp.ren}.)
If $f$ is supported then define $\rho\bigact f$ by
\begin{frameqn}
(\rho\bigact f)(x) = (\rho_2\circ\rho)\bigact f(\rho_1\bigact x)
\end{frameqn}
\noindent \!\!\!\! for some/any freshening pair of renamings $\rho_1$ and $\rho_2$ for $\nontriv(\rho)$ (which is finite), with respect to $\supp(x)\cup S_f$. 
\end{defn}

\begin{lemm}
Definition~\ref{defn.exp.ren.action} is well-defined.
That is, it does not matter which freshening pair of renamings we choose.
\end{lemm}
\begin{proof}
Consider two freshening pairs of renamings $\rho_1,\rho_2$ and $\rho_1',\rho_2'$.
These exist because $\nontriv(\rho)$ is finite and $\mathbb A\setminus(\supp(x)\cup S_f)$ is infinite.

Let $\rho_1''$ map $\img(\rho_1)$ to $\img(\rho_1')$ and $\rho_2''$ map $\dom(\rho_2')=\img(\rho_1')$ to $\dom(\rho_2)=\img(\rho_1)$ in such a way that 
\begin{itemize*}
\item
$\rho_1'(a)=(\rho_1''\circ\rho_1)(a)$ for all $a\in\dom(\rho_1')$, 
\item
$\rho_2'(a)=(\rho_2\circ\rho_2'')(a)$ for all $a\in\dom(\rho_2')$, and 
\item
$\nontriv(\rho_1'')=\img(\rho_1)\cup\img(\rho_1')$ and $\nontriv(\rho_2'')=\dom(\rho_2')\cup\dom(\rho_2)$.
\end{itemize*}
We reason as follows:
\begin{tab7}
(\rho_2'\circ\rho)\bigact f((\rho_1'\circ\rho)\bigact x)=&
(\rho_2\circ\rho_2''\circ\rho)\bigact f((\rho_1''\circ\rho_1\circ\rho)\bigact x)
&\text{Lems.~\ref{lemm.supp.supported.f.bound} \& \ref{lemm.supp.subsets}, Def.~\ref{defn.finsupp}}
\\
=&(\rho_2\circ\rho_2''\circ\rho\circ\rho_1'')\bigact f((\rho_1\circ\rho)\bigact x)
&\dom(\rho_1'')\cap S_f=\varnothing
\\
=&(\rho_2\circ\rho_2''\circ\rho_1''\circ\rho)\bigact f((\rho_1\circ\rho)\bigact x)
&\nontriv(\rho_1'')\cap\nontriv(\rho)=\varnothing 
\\
=&(\rho_2\circ\rho)\bigact f((\rho_1\circ\rho)\bigact x)
&\text{Lems.~\ref{lemm.supp.supported.f.bound} \& \ref{lemm.supp.subsets}, Def.~\ref{defn.finsupp}}
\end{tab7}
\end{proof}

\begin{lemm}
\label{lemm.renaming.distribute}
Suppose $x\in|\rs X|$ and $\rho$ is a renaming.
Suppose $f\in|\rs X|\to|\rs Y|$ is supported.

Then $\rho\bigact(f(x))=(\rho\bigact f)(\rho\bigact x)$.
\end{lemm}
\begin{proof}
Let $\rho_1$ and $\rho_2$ be a freshening pair of renamings of $\nontriv(\rho)$ with respect to $S_f\cup\supp(x)$.

Let $\rho'$ be a renaming with $\nontriv(\rho')=\img(\rho_1)$ such that $\rho_1\circ\rho=\rho'\circ\rho_1$; this exists since $\rho_1$ is injective on $\nontriv(\rho)$ and `freshens' this set to some fresh set of atoms.

We reason as follows: 
\begin{tab7}
(\rho\bigact f)(\rho\bigact x)=&(\rho_2\circ\rho)\bigact f((\rho_1\circ\rho)\bigact x)
&\text{Definition~\ref{defn.exp.ren.action}}
\\
=&(\rho_2\circ\rho)\bigact f((\rho'\circ\rho_1)\bigact x)
&\text{Definition~\ref{defn.finsupp}}
\\
=&(\rho_2\circ\rho\circ\rho')\bigact f(\rho_1\bigact x)
&\nontriv(\rho')\cap S_f=\varnothing
\\
=&(\rho\circ\rho_2)\bigact f(\rho_1\bigact x)
&\text{Lem.~\ref{lemm.supp.supported.f.bound}, Def.~\ref{defn.finsupp}}
\\
=&\rho\bigact f((\rho_2\circ\rho_1)\bigact x)
&\dom(\rho_2)\cap S_f=\varnothing
\\
=&\rho\bigact f(x)
&\text{Definition~\ref{defn.finsupp}}
\end{tab7}
\end{proof}

\subsubsection{Definition of the exponential}

\begin{frametxt}
\begin{defn}
\label{defn.frs.exp}
Write $\rs X\Rightarrow\rs Y$ for the renaming set with underlying set those $f\in|\rs X|\to|\rs Y|$ that are supported in the sense of Definition~\ref{defn.exp.ren}, and renaming action as defined in Definition~\ref{defn.exp.ren.action}.
\end{defn} 
\end{frametxt}
 
\begin{lemm}
\label{lemm.OK}
If $f$ is supported in the sense of Definition~\ref{defn.exp.ren} then it is supported by $S_f$ in the sense of Definition~\ref{defn.finsupp}.
Thus, $\rs X\Rightarrow\rs Y$ is indeed a permissive-nominal renaming set.
\end{lemm}
\begin{proof}
It suffices to show that if $a\not\in S_f$ then $([a\ssm b]\bigact f)(x)=f(x)$.
This follows by routine calculations. 
\end{proof}

\begin{lemm}
\theory{PmsRen} (Definition~\ref{defn.fps}) is cartesian closed: 
\begin{itemize*}
\item
The exponential is $\rs X\Rightarrow\rs Y$ from Definition~\ref{defn.frs.exp}. 
\item
Products are given pointwise as in Definition~\ref{defn.times}.
\item
The terminal object $\rs 1$ is the singleton set $\{0\}$ with the trivial action $\rho\bigact 0=0$. 
\end{itemize*}
\end{lemm}
\begin{proof}
The bijection between $(\rs X\times\rs Y)\to\rs Z$ and $\rs X\to (\rs X\Rightarrow\rs Y)$ is given by currying and uncurrying as usual.
Thus $G:(\rs X\times\rs Y)\to\rs Z$ maps to $x\mapsto \lam{y}G(x,y)$.
It is not hard to verify that if $\dom(\rho)\cap \supp(x)=\varnothing$ then
$$
(\rho\bigact \lam{y}F(x,y))(y) = \rho\bigact F(x,y) = F(x,\rho\bigact y) = (\lam{y}F(x,y))(\rho\bigact y) .
$$
Thus $\lam{y}G(x,y)$ is supported by $\supp(x)$ and is in $\rs Y\Rightarrow\rs Z$.
\end{proof}

We take a moment to build a particular exponential which will be useful later.
\begin{defn}
\label{defn.lambda.a}
Suppose $x\in|\rs X|$ and $a\in\mathbb A_\nu$.
Write $\lam{a}x\in|\mathbb A_\nu|\to|\rs X|$ for the function mapping $a$ to $x$ and $b$ to $[a\ssm b]\bigact x$.
\end{defn}

\begin{lemm}
$\lam{a}x\in |{\mathbb A_\nu\Rightarrow\rs X}|$.
\end{lemm}
\begin{proof} 
It suffices to show that $\lam{a}x$ is supported by $\supp(x)$ (in fact, it is also supported by $\supp(x){\setminus}\{a\}$).
Suppose $\dom(\rho)\cap\supp(x)=\varnothing$ and $z\in\mathbb A_\nu$ ($z$ is not necessarily distinct from $a$).
Write $\rho\text{-}a$ for the renaming such that $(\rho\text{-}a)(b)=\rho(b)$ and $(\rho\text{-}a)(a)=a$.
We sketch the relevant reasoning:
$$
\rho\bigact((\lam{a}x)z) = (\rho\circ [a\ssm z])\bigact x
=([a\ssm\rho(z)]\circ(\rho\text{-}a))\bigact x=[a\ssm\rho(z)]\bigact x=(\lam{a}x)(\rho\bigact z) 
$$
\end{proof}

\subsection{Atoms, products, atoms-abstraction, and functions out of atoms}
\label{subsect.atoms.ren.example}

\subsubsection{Atoms}

\begin{defn}
\label{defn.bool}
Write $\mathbb B$ for the nominal set and the permutation/renaming set with underlying set $\{0,1\}$ and the \deffont{trivial} permutation/renaming action such that $\pi\act x=x$/$\rho\bigact x=x$ always.

We will be lax and write $x\in\mathbb B$ for $x\in|\mathbb B|$.

Write $\mathbb A_\nu$ for the permutation set and the renaming set with underlying set $\mathbb A_\nu$ and the natural permutation/renaming action such that $\pi\act x=\pi(x)$/$\rho\bigact x=\rho(x)$ always.

We will be lax and write $x\in\mathbb A_\nu$ for $x\in|\mathbb A_\nu|$.
\end{defn}

\subsubsection{Atoms-abstraction in permutation and renaming sets}

\begin{defn}
\label{defn.abstraction.sets}
Suppose $\ns X$ is a supported permutation set.
Suppose $x\in |\ns X|$ and $a\in\mathbb A_\nu$.
Define \deffont{atoms-abstraction} $[a]x$ and $[\mathbb A_\nu]\ns X$ by:
\begin{frameqn}
\begin{array}{r@{\ }l}
[a]x =& \{(a,x)\}\cup \{(b,(b\ a)\act x)\mid b\in\mathbb A_\nu{\setminus} \f{supp}(x)\}
\\
|[\mathbb A_\nu]\ns X| =& \{[a]x\mid a\in\mathbb A_\nu,\ x\in|\ns X|\} 
\\
\pi\act [a]x =& [\pi(a)]\pi\act x
\end{array}
\end{frameqn}
\end{defn}

\begin{lemm}
\label{lemm.supp.abstraction}
Suppose $\ns X$ is a supported permutation set.
\begin{enumerate*}
\item
$[\mathbb A_\nu]\ns X$ is a supported permutation set.
\item
$[a]x{=}[a]x'$ if and only if $x{=}x'$, for $a{\in}\mathbb A_\nu$ and $x{\in} |\ns X|$.
\item
$[a]x{=}[a']x'$ if and only if $a'{\not\in}\f{supp}(x)$ and $(a'\, a)\act x{=}x'$, for $a,a'{\in}\mathbb A_\nu$ and $x,x'{\in}|\ns X|$.
\end{enumerate*}
\end{lemm}

We do not need Definition~\ref{defn.abstraction.sets'} for the completeness proof but we include it for the interested reader to compare and contrast with Definition~\ref{defn.abstraction.sets}.
\begin{defn}
\label{defn.abstraction.sets'}
Suppose $\rs X$ is a supported renaming set.
Suppose $x\in |\rs X|$ and $a\in\mathbb A_\nu$.
Define \deffont{atoms-abstraction} $[a]x$ and $[\mathbb A_\nu]\rs X$ by:
\begin{frameqn}
\begin{array}{r@{\ }l}
[a]x =& \{(a,x)\}\cup \{(b,[a\ssm b]\bigact x)\mid b\in\mathbb A_\nu{\setminus} \supp(x)\}
\\
|[\mathbb A_\nu]\rs X| =& \{[a]x\mid a\in\mathbb A_\nu,\ x\in|\rs X|\} 
\\
\rho\bigact [a]x =& [a]\rho\bigact x\quad (a\not\in\f{nontriv}(\rho))
\end{array}
\end{frameqn}
\end{defn}

\begin{rmrk}
\label{rmrk.total-partial}
Definitions~\ref{defn.abstraction.sets} and~\ref{defn.abstraction.sets'} look similar; both define graphs of partial functions defined on $\supp(x)\setminus\{a\}$.
However, the critical difference is that in renaming sets, this partial function can be extended to a total function in $\mathbb A_\nu\to\rs X$.

That is, $[a]x\in[\mathbb A_\nu]\rs X$ determines the total function $\lam{a}x$ from Definition~\ref{defn.lambda.a}, mapping $a$ to $x$ and any other $b$ to $[a\ssm b]\bigact x$.
We return to this in Lemma~\ref{lemm.non-iso} where we show that the natural map from $[\mathbb A_\nu]\rs X$ to ${\mathbb A_\nu\Rightarrow\rs X}$ which we construct in a moment in Definition~\ref{defn.some.functors}, is not surjective.
So Definition~\ref{defn.abstraction.sets'} identifies a `small' and `well-behaved' subset of the function space.
\end{rmrk}

A cognate of Lemma~\ref{lemm.supp.abstraction} also holds for $[\mathbb A_\nu]\rs X$:
\begin{lemm}
\label{lemm.supp.abstraction'}
Suppose $\rs X$ is a supported renaming set.
\begin{enumerate*}
\item
$[\mathbb A_\nu]\rs X$ is a supported renaming set.
\item
$[a]x{=}[a]x'$ if and only if $x{=}x'$, for $a{\in}\mathbb A_\nu$ and $x{\in} |\rs X|$.
\item
$[a]x{=}[a']x'$ if and only if $a'{\not\in}\f{supp}(x)$ and $(a'\, a)\bigact x{=}x'$ (or equivalently $[a\ssm a']\bigact x{=}x'$), for $a,a'{\in}\mathbb A_\nu$ and $x,x'{\in}|\rs X|$.
\end{enumerate*}
\end{lemm}

It may be useful to organise Definitions~\ref{defn.lambda.a} and~\ref{defn.abstraction.sets'} and Lemma~\ref{lemm.supp.abstraction'} into two functors and a natural transformation: 
\begin{defn}
\label{defn.some.functors}
Write $[\mathbb A_\nu]\text{-}$ for the functor taking $\rs X$ to $[\mathbb A_\nu]\rs X$ and taking $G:\rs X\longrightarrow\rs Y$ to $[\mathbb A_\nu]G:\mathbb A_\nu{\Rightarrow}\rs X\longrightarrow\mathbb A_\nu{\Rightarrow}\rs Y$ which maps $[a]x$ to $[a]G(x)$.

Write $\mathbb A_\nu{\Rightarrow}\text{-}$ for the functor taking $\rs X$ to $\mathbb A_\nu{\Rightarrow}\rs X$ and taking $G:\rs X\longrightarrow\rs Y$ to $\mathbb A_\nu{\Rightarrow}G:\mathbb A_\nu{\Rightarrow}\rs X\longrightarrow\mathbb A_\nu{\Rightarrow}\rs Y$ which maps $f$ to $G\circ f=\lam{n{\in}\mathbb A_\nu}G(f(n))$.

Finally, write $\f{AbsFun}$ for the natural transformation from $[\mathbb A_\nu]\text{-}$ to $\mathbb A_\nu{\Rightarrow}\text{-}$ such that $\f{AbsFun}(\rs X):[\mathbb A_\nu]\rs X\longrightarrow\mathbb A_\nu{\Rightarrow}\rs X$ maps $[a]x$ to $\lam{a}x$ (Definition~\ref{defn.lambda.a}).
\end{defn}

\begin{lemm}
The functors $[\mathbb A_\nu]\text{-}$ and $\mathbb A_\nu{\Rightarrow}\text{-}$ and the natural transformation $\f{AbsFun}$ are indeed well-defined.
\end{lemm}
\begin{proof}
By routine calculations using the fact that $G:\rs X\longrightarrow\rs Y$ is equivariant (Definition~\ref{defn.equivariant}).
We consider two of the relevant calculations:
\begin{itemize*}
\item
\emph{If $G:\rs X\longrightarrow\rs Y$ then $[\mathbb A_\nu]G$ is well-defined.}\quad
Suppose $[a]x=[b]y\in|[\mathbb A_\nu]\rs X|$.
By Lemma~\ref{lemm.supp.abstraction'} $b\not\in\supp(x)$ and $(b\ a)\act x=y$.
By equivariance $G((b\ a)\bigact x)=(b\ a)\bigact G(x)$.
By part~2 of Lemma~\ref{lemm.equivar.reduces.supp} $b\not\in\supp(G(x))$.
By Lemma~\ref{lemm.supp.abstraction'} $[a]G(x)=[b]G(y)$.
\item
\emph{$\f{AbsFun}$ is a natural transformation.}\quad
We need that $\lam{a}(G(x))=\lam{n{\in}\mathbb A_\nu}G((\lam{a}x)n)$.
Unpacking Definition~\ref{defn.lambda.a} $(\lam{a}x)n=[a\ssm n]\bigact x$ (here $n{\in}\mathbb A_\nu$ is not necessarily distinct from $a$).
By equivariance of $G$,\ $\lam{n{\in}\mathbb A_\nu}G((\lam{a}x)n)=\lam{n{\in}\mathbb A_\nu}[a\ssm n]\bigact G(x)$.
This is exactly equal to $\lam{a}G(x)$ as required.
\end{itemize*}
\end{proof}
In part~4 of Lemma~\ref{lemm.non-iso} we prove that $\f{AbsFun}$ does not have an inverse.
This merely points out the obvious: there are more functions from $\mathbb A_\nu$ to $\rs X$ than can be represented in the form $\lam{a}x=\lam{n{\in}\mathbb A_\nu}[a\ssm n]\bigact x$.

\subsubsection{Product}

\begin{defn}
\label{defn.times}
If $\ns X_i$ and $\rs X_i$ are supported permutation and renaming sets respectively for $1\leq i\leq n$ then define $\ns X_1\times\ldots\times \ns X_n$ and $\rs X_1\times\ldots\times\rs X_n$ by:
$$
\begin{array}{r@{\ }l}
|\ns X_1\times\ldots\times\ns X_n|=&|\ns X_1|\times\ldots\times|\ns X_n|
\\
\pi\act (x_1,\ldots,x_n)=&(\pi\act x_1,\ldots,\pi\act x_n)
\end{array}
\quad\quad
\begin{array}{r@{\ }l}
|\rs X_1\times\ldots\times\rs X_n|=&|\rs X_1|\times\ldots\times|\rs X_n|
\\
\rho\bigact (x_1,\ldots,x_n)=&(\rho\bigact x_1,\ldots,\rho\bigact x_n)
\end{array}
$$
\end{defn}

\begin{lemm}
\label{lemm.properties.of.support}
\begin{itemize*}
\item
$\f{supp}(a)=\{a\}$.
\item
$\f{supp}([a]x)=\f{supp}(x)\setminus\{a\}$.
\item
$\f{supp}((x_1,\ldots,x_n))=\bigcup\{\f{supp}(x_i)\mid 1\leq i\leq n\}$.
\end{itemize*}
\end{lemm}
\begin{proof}
By routine arguments like those in \cite{gabbay:newaas-jv} or \cite[Corollary~2.30 \& Theorem~3.11]{gabbay:fountl}. 
\end{proof}

\subsection{The free extension of a permutation set to a renaming set}
\label{subsect.free.ext}

We now show how to construct a renaming set $\Ren{\ns X}$ out of a nominal set.
At the start of Section~\ref{sect.semantics} we noted that an atoms-abstraction $[a]x\in|[\mathbb A_\nu]\ns X|$ can be viewed as a partial function and that in renaming sets atoms-abstraction also exists but can be completed to a total function (Remark~\ref{rmrk.total-partial}). 
We can view our free construction as a canonical way to move from a world in which atoms-abstraction is partial, to a world in which it is total.

\begin{nttn}
If $\sim$ is an equivalence relation, $[\text{-}]_\sim$ will denote the equivalence class of $\text{-}$ in $\sim$. 
\end{nttn}

\begin{defn}
\label{defn.free.ren}
We define a functor $\Ren{\text{-}}$ from \theory{PmsPrm} to \theory{PmsRen} as follows:
\begin{itemize*}
\item
\emph{Action of $\Ren{\text{-}}$ on objects.}

$\ns X$ maps to $\Ren{\ns X}=((\mathbb R\times|\ns X|)/{\sim},\bigact)$ where $\rho\bigact [(\rho',x)]_\sim = [(\rho\circ\rho',x)]_\sim$ and $\sim$ is the least equivalence relation such that:
\begin{frametxt}
\begin{enumerate*}
\item
If $\rho(a)=\rho'(a)$ for every $a\in\supp(x)$ then $(\rho,x)\sim (\rho',x)$.
\item
$(\rho\circ\pi,x)\sim(\rho,\pi\act x)$.
\end{enumerate*} 
\end{frametxt}
For convenience we will write $\rho\bigact x$ as shorthand for $[(\rho,x)]_\sim$.\footnote{$\rho\bigact x$ is not `$\rho$ acting on $x$' and cannot be, since $x\in|\ns X|$ only has a permutation action. 
However this notation gives us cleaner-looking maths and the nice equality $\rho\bigact(\id\bigact x)=\rho\bigact x$ (in long form: $\rho\bigact[(\id,x)]_\sim=[(\rho,x)]_\sim$). 
The notation follows the nominal terms tradition of writing $\pi\act X$ both for `$\pi$ acting on the moderated unknown $\id\act X$', and for `the moderated unknown $\pi\act X$' \cite{gabbay:nomu-jv}.}
\item
\emph{Action of $\Ren{\text{-}}$ on arrows.}

An arrow $F:\ns X\longrightarrow\ns Y$ maps to $\Ren F:\Ren{\ns X}\longrightarrow\Ren{\ns Y}$ given by:
$$
\Ren F(\rho\bigact x)=\rho\bigact F(x)
$$
\end{itemize*}
\end{defn}

\begin{lemm}
$\Ren F$ is well-defined; that is, that if $(\rho,x)\sim(\rho',x')$ then $\Ren F((\rho,x))\sim \Ren F((\rho',x'))$.
\end{lemm}
\begin{proof}
Induction on the derivation that $(\rho,x){\sim}(\rho',x')$.
We consider the two base cases:
\begin{itemize*}
\item
\emph{The case $\rho(a)=\rho'(a)$ for every $a\in\supp(x)$.}\quad
By part~2 of Lemma~\ref{lemm.equivar.reduces.supp} also $\rho(a)=\rho'(a)$ for every $a\in\supp(F(x))$.
\item
\emph{The case $(\rho\circ\pi,x)\sim (\rho,\pi\act x)$.}\quad
Then also $(\rho\circ\pi,F(x))\sim (\rho,\pi\act F(x))$ and by equivariance $\pi\act F(x)=F(\pi\act x)$.
\qedhere\end{itemize*}
\end{proof}

\begin{rmrk}
\label{rmrk.alpha}
Rules~2 and~1 of Definition~\ref{defn.free.ren} can be viewed as $\alpha$-conversion and garbage-collection respectively.
Thus in $\rho\bigact x\in\Ren{\ns X}$ we may without loss of generality (using rule~2) assume that $\dom(\rho)\cap S=\varnothing$ for any permission set $S$, and we may also assume (using rule~1) that $\dom(\rho)\subseteq\supp(x)$.
\end{rmrk}

\begin{lemm}
\label{lemm.BA.ren.isos}
\begin{enumerate*}
\item
$\Ren{\mathbb B}$ (for $\mathbb B$ considered a set with the trivial permutation action) is isomorphic to $\mathbb B$ (for $\mathbb B$ considered a set with a trivial renaming action).
\item
$\Ren{\mathbb A_\nu}$ (for $\mathbb A_\nu$ with its natural permutation action) is isomorphic to $\mathbb A_\nu$ (for $\mathbb A_\nu$ with its natural renaming action).
\end{enumerate*}
\end{lemm}
\begin{proof}
We consider only the second part.
This follows if we note that according to the rules for $\sim$ in Definition~\ref{defn.free.ren},\ 
$$
(\rho,a)\stackrel{\text{\it rule~1}}{\sim} ((\rho(a)\ a),a)\stackrel{\text{\it rule~2}}{\sim} (\id,\rho(a)).
$$
\end{proof}

Where we are dealing with more than zero or one atoms at a time, isomorphisms like those in Lemma~\ref{lemm.BA.ren.isos} may fail:
\begin{lemm}
$\Ren{\mathbb A_\nu{\times}\mathbb A_\nu}$ is not isomorphic to $\Ren{\mathbb A_\nu}{\times}\Ren{\mathbb A_\nu}$ (which is isomorphic to $\mathbb A_\nu{\times}\mathbb A_\nu$).
\end{lemm}
\begin{proof}
Consider the element $[a\ssm b]\bigact(a,b)$.
\end{proof}

\section{Interpretation of permissive-nominal logic}
\label{sect.permissive-nominal.sets}

\subsection{Interpretation of signatures} 
 
\begin{defn}
\label{defn.interpretation}
Suppose $(\mathcal A,\mathcal B)$ is a sort-signature (Definition~\ref{defn.sort.sig}).

A \deffont{PNL interpretation} $\mathcal I$ for $(\mathcal A,\mathcal B)$ consists of an assignment of a nonempty supported permutation set $\basesort^\iden$ to each $\basesort\in\mathcal B$.

We extend an interpretation $\mathcal I$ to sorts by:
\begin{frameqn}
\begin{array}{r@{\ }l@{\qquad}r@{\ }l}
\model{\basesort}=&\basesort^\iden
&
\model{(\alpha_1,\ldots,\alpha_n)}=&\model{\alpha_1}\times\ldots\times\model{\alpha_n}
\\
\model{\nu}=&\mathbb A_\nu
&
\model{[\nu]\alpha}=&[\mathbb A_\nu]\model{\alpha}
\end{array}
\end{frameqn}
\end{defn}

\begin{defn}
\label{defn.interpret.I}
Suppose $\mathcal S=(\mathcal A,\mathcal B,\mathcal F,\mathcal P,\f{ar},\mathcal X)$ is a signature (Definition~\ref{defn.signature}).

A \deffont{(non-equivariant) PNL interpretation} $\mathcal I$ for $\mathcal S$ consists of the following data:
\begin{itemize*}
\item
An interpretation for the sort-signature $(\mathcal A,\mathcal B)$ (Definition~\ref{defn.interpretation}).
\item
For every $\tf f\in\mathcal F$ with $\f{ar}(\tf f)=(\alpha')\alpha$ an equivariant function $\tf f^\iden$ from $\model{\alpha'}$ to $\model{\alpha}$ (Definition~\ref{defn.equivariant}).
\item
For every $\tf P\in\mathcal P$ with $\f{ar}(\tf P)=\alpha$ a supported function $\tf P^\iden$ from $\model{\alpha}$ to $\{0,1\}$. 
\end{itemize*}
If every $\tf P^\iden$ is equivariant, then call $\mathcal I$ a \deffont{fully equivariant} interpretation.\footnote{A non-equivariant PNL interpretation still interprets term-formers equivariantly. Only the predicates might not be equivariant.  
We do this in order to completely model \rulefont{Ax^\nopi} from Figure~\ref{rSeq}, so that $\tf P(r)\not\liff\tf P(\pi\act r)$; see Theorem~\ref{thrm.reduced.pnl.completeness}.
Of course it is possible to imagine a notion of non-equivariant interpretation where term-formers are interpreted as non-equivariant functions.  This would correspond to something else: namely, to losing the property that $\pi\act\tf f(r)=\tf f(\pi\act r)$.} 
\end{defn}

\subsection{Interpretation of terms}
\label{subsect.interpret.pnl.terms}

\begin{defn}
\label{defn.valuation}
Suppose $\mathcal I$ is an interpretation for $\mathcal S$.
A \deffont{valuation} $\varsigma$ to $\mathcal I$ is a map on unknowns such that for each unknown $X$,\ 
\begin{itemize*}
\item
$\varsigma(X)\in\model{\sort(X)}$,\ and\ 
\item
$\f{supp}(\varsigma(X))\subseteq \pmss(X)$.
\end{itemize*}
$\varsigma$ will range over valuations.
\end{defn}

Having interpreted sorts $\alpha$ in Definition~\ref{defn.interpretation} as permissive-nominal sets $\model{\alpha}$, we now interpret nominal terms $r:\alpha$ as elements of those sets $\denot{\mathcal I}{\varsigma}{r}$:
\begin{defn}
\label{defn.interpret.terms}
Suppose $\mathcal I$ is an interpretation of a signature $\mathcal S$.
Suppose $\varsigma$ is a valuation to $\mathcal I$.

Define an \deffont{interpretation} 
$\denot{\mathcal I}{\varsigma}{r}$ in $\mathcal S$ by:
\begin{frameqn}
\begin{array}{r@{\ }l@{\qquad}r@{\ }l}
\denot{\mathcal I}{\varsigma}{a} =& a
&
\denot{\mathcal I}{\varsigma}{[a]r} =& [a]\denot{\mathcal I}{\varsigma}{r}
\\
\denot{\mathcal I}{\varsigma}{\tf f(r)} =& 
\tf f^\iden(\denot{\mathcal I}{\varsigma}{r})
&
\denot{\mathcal I}{\varsigma}{\pi\act X} =& \pi\act\varsigma(X)
\\
\denot{\mathcal I}{\varsigma}{(r_1,\ldots,r_n)} =& 
(\denot{\mathcal I}{\varsigma}{r_1},\ldots,\denot{\mathcal I}{\varsigma}{r_n})
\end{array}
\end{frameqn}
\end{defn}

\begin{lemm}
\label{lemm.sort.r}
If $r:\alpha$ then $\denot{\mathcal I}{\varsigma}{r}\in\model{\alpha}$.
\end{lemm}
\begin{proof}
By a routine induction on $r$.
\end{proof} 

\begin{lemm}
\label{lemm.pi.r.model}
$\pi\act\denot{\mathcal I}{\varsigma}{r} = \denot{\mathcal I}{\varsigma}{\pi\act r}$.
\end{lemm}
\begin{proof}
By a routine induction on $r$. 
We consider one case:
\begin{itemize}
\item
\emph{The case $\pi'\act X$.}\quad
By Definition~\ref{defn.interpret.terms} 
$\denot{\mathcal I}{\varsigma}{\pi'\act X} = \pi'\act \varsigma(X)$.
Therefore $\pi\act
\denot{\mathcal I}{\varsigma}{\pi'\act X} = \pi\act(\pi'\act \varsigma(X))$.
It is a fact of the group action (Definition~\ref{defn.perm.set}) that $\pi\act(\pi'\act\varsigma(X))=(\pi\circ\pi')\act\varsigma(X)$, and of the permutation action (Definition~\ref{defn.permutation.action}) that $\pi\act(\pi'\act X)\equiv (\pi\circ\pi')\act X$.
The result follows.
\qedhere
\end{itemize}
\end{proof}

\begin{lemm}
\label{lemm.supp.r}
$\f{supp}(\denot{\mathcal I}{\varsigma}{r})\subseteq\fa(r)$.
\end{lemm}
\begin{proof}
By a routine induction on $r$.
We consider one case in detail:
\begin{itemize}
\item
\emph{The case $\pi\act X$.}\quad
$\fa(\pi\act X)=\pi\act\pmss(X)$ by Definition~\ref{defn.fa}.
By assumption in Definition~\ref{defn.valuation} $\f{supp}(\varsigma(X))\subseteq\pmss(X)$.
\end{itemize} 
The cases of $a$, $[a]r$, and $(r_1,\dots,r_n)$ use parts~1, 2, and 3 of Lemma~\ref{lemm.properties.of.support}.
The case of $\tf f$ uses part~1 of Lemma~\ref{lemm.equivar.reduces.supp}.
\end{proof}

\subsection{Interpretation of propositions} 
\label{subsect.interpret.prop}

\begin{defn}
\label{defn.varsigma.sub}
Suppose $\varsigma$ is a valuation to an interpretation $\mathcal I$.
Suppose $X$ is an unknown and $x\in \model{\sort(X)}$ is such that $\f{supp}(x)\subseteq \pmss(X)$.
Define $\varsigma[X\ssm x]$ by
$$
(\varsigma[X\ssm x])(Y)=\varsigma(Y)
\quad\text{and}\quad
(\varsigma[X\ssm x])(X)=x .
$$
\end{defn}
It is easy to verify that $\varsigma[X\ssm x]$ is also a valuation to $\mathcal I$.

\begin{defn}
\label{defn.truth}
Suppose $\mathcal I$ is an interpretation.
Define an \deffont{interpretation of propositions} by:
\begin{frameqn}
\begin{array}{r@{\ }l}
\denot{\mathcal I}{\varsigma}{\tf P(r)} =&
\tf P^\iden(\denot{\mathcal I}{\varsigma}{r})
\\
\denot{\mathcal I}{\varsigma}{\bot}=& 0
\\
\denot{\mathcal I}{\varsigma}{\phi\limp\psi}=& 
\f{max}\{1{-}\denot{\mathcal I}{\varsigma}{\phi},\denot{\mathcal I}{\varsigma}{\psi}\}
\\
\denot{\mathcal I}{\varsigma}{\Forall{X}\phi}=&\f{min}\{\denot{\mathcal I}{\varsigma[X{\ssm} x]}{\phi}\mid x{\in} \model{\sort(X)},\, \f{supp}(x){\subseteq} \pmss(X)
\}
\end{array}
\end{frameqn}
\end{defn}

We may identify $\denot{\mathcal I}{}{\phi}$ with a set of valuations $\{\varsigma\mid\denot{\mathcal I}{\varsigma}{\phi}=1\}$.
We discuss soundness and completeness in Appendix~\ref{sect.completeness}.

\begin{lemm}
\label{lemm.denotsub}
\begin{itemize*}
\item
$\denot{\mathcal I}{\varsigma[X\ssm \denot{\mathcal I}{\varsigma}{r'}]}{r}
=\denot{\mathcal I}{\varsigma}{r[X\ssm r']}$.
\item
$\denot{\mathcal I}{\varsigma[X\ssm \denot{\mathcal I}{\varsigma}{r'}]}{\phi} =
\denot{\mathcal I}{\varsigma}{\phi[X\ssm r']}$.
\end{itemize*}
\end{lemm}
\begin{proof}
By routine inductions on the definitions of 
$\denot{\mathcal I}{\varsigma}{r}$ and 
$\denot{\mathcal I}{\varsigma}{\phi}$ in 
Definitions~\ref{defn.interpret.terms} and~\ref{defn.truth}.
We consider two cases:
\begin{itemize*}
\item
The case of $\denot{\mathcal I}{\varsigma[X\ssm r']}{\pi\act X}$.\quad
We reason as follows:
\begin{tab3}
\denot{\mathcal I}{\varsigma[X\ssm \denot{\mathcal I}{\varsigma}{r'}]}{\pi\act X}
=& \pi\act \denot{\mathcal I}{\varsigma}{r'} 
&\text{Definition~\ref{defn.interpret.terms}}
\\
=& \denot{\mathcal I}{\varsigma}{\pi\act r'}
&\text{Lemma~\ref{lemm.pi.r.model}}
\\
=& \denot{\mathcal I}{\varsigma}{(\pi\act X)[X\ssm r']}
&\text{Definition~\ref{defn.subst.action}} .
\end{tab3}
\item
The case of $\denot{\mathcal I}{\varsigma[X\ssm r']}{\tf P(r)}$.
\ \  
We reason as follows:
\begin{tab3}
\denot{\mathcal I}{\varsigma[X\ssm \denot{\mathcal I}{\varsigma}{r'}]}{\tf P(r)}
=&
{\tf P}^\iden(\denot{\mathcal I}{\varsigma[X\ssm \denot{\mathcal I}{\varsigma}{r'}]}{r})
&\text{Definition~\ref{defn.truth}}
\\
=&
{\tf P}^\iden(\denot{\mathcal I}{\varsigma}{r[X\ssm r']}) 
&\text{Part~1 of this result}
\\
=&
\denot{\mathcal I}{\varsigma}{\tf P(r)[X\ssm r']} 
&\text{Definition~\ref{defn.truth}} .
\end{tab3}
\end{itemize*}
\end{proof}

\begin{lemm}
\label{lemm.fV.denot}
If $\varsigma(X)=\varsigma'(X)$ for all $X\in\fU(r)$ then $\denot{\mathcal I}{\varsigma}{r}=\denot{\mathcal I}{\varsigma'}{r}$, and similarly for $\phi$.
\end{lemm}
\begin{proof}
By a routine induction on $r$ and $\phi$.
\end{proof}

\section{Interpretation of HOL}
\label{sect.interp.hol}

For this section fix some PNL interpretation $\mathcal I$ of a PNL signature $\mathcal S$.
Recall from Definition~\ref{defn.TS} the definition of the corresponding HOL signature $\mathcal T_{\mathcal S}$. 

We have our interpretation of PNL and we have from Definition~\ref{defn.translation} a translation of PNL syntax to HOL syntax.
We also have a functor from nominal sets to renaming sets (Definition~\ref{defn.free.ren}).
It remains to interpret HOL in renaming sets consistent with these interpretations and translations.
This is Definitions~\ref{defn.hol.interpretation} and~\ref{defn.hol.interpret.terms}, and the key technical result Lemma~\ref{lemm.commuting.square}.
Completeness follows quickly as a corollary (Theorem~\ref{thrm.PNL.HOL.complete}).

Note that in the interpretation (Definition~\ref{defn.hol.interpretation}) the type $\mu_\nu\to\beta$ is not necessarily interpreted as the set of all functions; it may be interpreted as a small subset of this function space.
This is an old idea: since Henkin, models of HOL have been constructed to cut down on the full function-space (e.g. to create a complete semantics \cite[Section~55]{andrews:intmlt}; see also \cite{benzmuller:higose} for a survey of non-standard semantics for HOL).

What we need to prove completeness of the syntactic translation $\hol{D}{\text{-}}$ is the existence of \emph{some} interpretation of HOL with certain properties.
This should not be mistaken as a commitment of nominal techniques to using this model of HOL always (unless we want to).

\subsection{Interpretation of types}

Recall the definition of a valuation $\varsigma$ (Definition~\ref{defn.valuation}) to an interpretation $\mathcal I$ for the PNL signature $\mathcal S$. 
Recall the definition of $\varsigma[X\ssm x]$ (Definition~\ref{defn.varsigma.sub}), and the interpretations of terms $\denot{\mathcal I}{\varsigma}{r}$ (Definition~\ref{defn.interpret.terms}) and propositions $\denot{\mathcal I}{\varsigma}{\phi}$ (Definition~\ref{defn.truth}).

We give similar definitions for HOL and renaming sets, culminating with Theorem~\ref{thrm.hol.soundness} (soundness).

\begin{defn}
\label{defn.hol.interpretation}
We provide an interpretation $\mathcal H$ of $\mathcal T_{\mathcal S}$ by:
\begin{frameqn}
\begin{array}{r@{\ }l@{\qquad}l}
\holmodel{\hol{}{\alpha}}=&\Ren{\model{\alpha}}
\\
\holmodel{o}=&\mathbb B
\\
\holmodel{(\beta_1,\ldots,\beta_n)}=&\holmodel{\beta_1}\times\ldots\times\holmodel{\beta_n}
& (\beta_i\text{ not of the form }\hol{}{\alpha}\text{ for at least one }i) 
\\
\holmodel{\beta'\to\beta}=&\holmodel{\beta'}\Rightarrow\holmodel{\beta} 
& (\beta'\text{ or }\beta\text{ not of the form }\hol{}{\alpha}) 
\end{array}
\end{frameqn}
\end{defn}
Recall $\rs X\Rightarrow\rs Y$ from Definition~\ref{defn.frs.exp} and $\rs X\times\rs Y$ from Definition~\ref{defn.times}.

\begin{rmrk}
\label{rmrk.case-split}
Not all function types are interpreted equally by Definition~\ref{defn.hol.interpretation}. 

If a type is the image of a PNL sort then we handle it using the first clause by wrapping it up in $\Ren{\text{-}}$.
Otherwise the interpretation is as standard: pairs to product; function types to the (supported) function set.
This case-split makes Lemma~\ref{lemm.commuting.square} work, which is central to Corollary~\ref{corr.notsubseteq} and to Completeness (Theorem~\ref{thrm.PNL.HOL.complete}).

Why Lemma~\ref{lemm.commuting.square} \emph{could not} work if we did not do this, is indicated in Lemma~\ref{lemm.non-iso}.
Briefly, 
${\mathbb A_\nu\Rightarrow\text{-}}$ contains `exotic elements' making it bigger than $[\mathbb A_\nu]\text{-}$, which readers familiar with higher-order abstract syntax would expect \cite[\emph{exotic terms}]{despeyroux94higherorder}.

Perhaps less familiar from Lemma~\ref{lemm.commuting.square} is that 
the free construction $\Ren{\text{-}}$ does not commute with atoms-abstraction or even with cartesian product (which tells us that the associated functor does not have a left adjoint). 
So even e.g. $\mathbb A_\nu\times\mathbb A_\nu$ in \theory{PmsRen} has an `exotic element'. 

The natural maps below in Lemma~\ref{lemm.non-iso} are functors (some are defined in Definition~\ref{defn.some.functors}).
Technically, Lemma~\ref{lemm.non-iso} implies that various natural transformations between these functors do not have inverses.
More loosely, this is a way of putting some formal measure to the intuition that the translation from (restricted) PNL to HOL cannot be surjective:
\end{rmrk}

\begin{lemm}
\label{lemm.non-iso}
\begin{enumerate*}
\item
The natural map from $\Ren{\mathbb A_\nu}$ to $\mathbb A_\nu$ mapping $\rho\bigact a$ to $\rho(a)$, is a bijection (cf. Lemma~\ref{lemm.BA.ren.isos}).
\item
The natural map from $\Ren{\ns X\times\ns Y}$ to $\Ren{\ns X}\times\Ren{\ns Y}$ mapping $\rho\bigact(x,y)$ to $(\rho\bigact x,\rho\bigact y)$ is neither surjective nor injective.
\item
The natural map from $\Ren{[\mathbb A_\nu]\ns X}$ to $[\mathbb A_\nu]\Ren{\ns X}$ mapping $\rho\bigact [a]x$ where $a\not\in\f{nontriv}(\rho)$ to $[a]\rho\bigact x$, is not surjective.
\item
The natural map from  $[\mathbb A_\nu]\rs Y$ to ${\mathbb A_\nu\Rightarrow\rs Y}$ mapping $[a]x$ to $\lam{a}x$ (see Definitions~\ref{defn.lambda.a} and~\ref{defn.some.functors}), is not surjective.
\end{enumerate*}
\end{lemm}
\begin{proof}
\begin{enumerate*}
\item
By rule~2 of Definition~\ref{defn.free.ren}.
\item
Take $\ns X=\ns Y=\mathbb A_\nu$.
The natural map from $\Ren{\ns X\times\ns Y}$ to $\Ren{\ns X}\times\Ren{\ns Y}$ takes $\id\bigact (a,b)$ to $(\id\bigact a,\id\bigact b)$.
By equivariance it must map $[a\ssm b]\bigact(a,b)$ to $(\id\bigact b,\id\bigact b)$.
But then it is not injective, since $[a\ssm b]\bigact(a,b)\neq \id\bigact(b,b)$ in $\Ren{\ns X\times\ns Y}$.

Now take $\ns X=\ns Y=\mathbb A_\nu\times\mathbb A_\nu$.
It is not hard to see that $([a\ssm b]\bigact(a,b),[b\ssm a]\bigact(a,b))$ is not in the image of the natural map, so the map is also not surjective.
\item
Take $\ns X=\mathbb A_\nu\times\mathbb A_\nu$ and consider $[a][a\ssm b]\bigact(a,b)\in[\mathbb A_\nu]\Ren{\ns X}$.
\item
Take $\rs Y=\mathbb A_\nu$ for $\mathbb A_\nu$ considered a renaming set as in Definition~\ref{defn.bool}.
Consider $(b\ a)\in{\mathbb A_\nu\Rightarrow\mathbb A_\nu}$, mapping $a$ to $b$, $b$ to $a$, and all other $c$ to $c$.
\qedhere\end{enumerate*}
\end{proof}

\subsection{Interpretation of terms}

\begin{defn}
\label{defn.hol.valuation}
A \deffont{(HOL) valuation} $\varrho$ to $\mathcal H$ is a map on variables $X:\beta$ 
such that 
\begin{itemize*}
\item
$\varrho(X)\in\holmodel{\beta}$ for every variable $X$.
\item
$\{a\in\mathbb A\mid \varrho(a)\neq\id\bigact a\}$ is finite.\footnote{In other words, $\varrho$ restricted to those HOL variables in $\mu_\nu$ that are PNL atoms (condition~5 of Definition~\ref{defn.hol.translation}) is a renaming $\rho$ (Definition~\ref{defn.renaming}).}
\end{itemize*} 
$\varrho$ will range over valuations.
\end{defn}

\begin{defn}
\label{defn.varrho.update}
Suppose $\varrho$ is a valuation.
Suppose $X$ is a variable and $x\in \holmodel{\type(X)}$. 
Define a function $\varrho[X\ssm x]$ by:
\begin{frameqn}
(\varrho[X\ssm x])(Y)=\varrho(Y)
\quad\text{and}\quad
(\varrho[X\ssm x])(X)=x 
\end{frameqn}
\end{defn}
It is easy to verify that $\varrho[X\ssm x]$ is also a valuation to $\mathcal H$.

\begin{defn}
\label{defn.hol.interpret.terms} 
Extend $\mathcal H$ to terms as follows:
\begin{itemize*}
\item
$\holmodel{a}(\varrho) = \varrho(a)$.
\item
$\holmodel{X}(\varrho) = \varrho(X)$.
\item
$\holmodel{\tf g_{\smtf f}} = \Ren{\tf f^\iden}$ and $\holmodel{\tf g_{\smtf P}} = \Ren{\tf P^\iden}$ (Definition~\ref{defn.free.ren}).
\item
$\holmodel{\bot}(\varrho)=0$.
\item
$\holmodel{\limp}(\varrho)=\lam{x\in\mathbb B,y\in\mathbb B}\f{max}\{1{-}x,y\}$.
\item
$\holmodel{\forall_{\beta})}(\varrho)=
\lam{x\in\holmodel{\beta\Rightarrow\mathbb B}}\f{min}\{xy\mid y\in\holmodel{\beta}\}$.
\item
$\holmodel{\lam{a}t}(\varrho)=\rho\bigact [a]x$ where $\holmodel{t}(\varrho)=\rho\bigact x$ provided that
$t:\hol{}{\alpha}$ for some PNL sort $\alpha$ and 
$a\in\mathbb A_\nu$ for some name sort $\nu$ and ($\alpha$-converting if necessary) 
$a\not\in\bigcup_{X\in\fv(t)\setminus\{a\}}\supp(\varrho(X))$ and
$\varrho(a)=\id\bigact a$.
\item
$\holmodel{\lam{X}t}(\varrho) = \lam{x}\holmodel{t}(\varrho[X\ssm x])$ provided that $\lam{X}t:\beta'\to\beta$ where $\beta'\to\beta$ is not equal to $\hol{}{[\mathbb A_\nu]\alpha}$ for any $\nu$ or $\alpha$.
\item
$\holmodel{tu}(\varrho) = ([a\ssm b]\circ\rho)\bigact x$ provided that 
$t:\hol{}{[\nu]\alpha}$ for some PNL name sort $\nu$ and sort $\alpha$, where $\holmodel{u}(\varrho)=\id\bigact b$ (by construction some such $b\in\mathbb A_\nu$ always exists) and $\holmodel{t}(\varrho)=\rho\bigact [a]x$, and (renaming if necessary) $a\not\in\f{nontriv}(\rho)\cup\{b\}$. 
\item
$\holmodel{tu}(\varrho) = \holmodel{t}(\varrho)\holmodel{u}(\varrho)$ provided that $t:\beta$ for $\beta$ not equal to $\hol{}{\alpha}$ for any PNL sort $\alpha$ (it is a fact that this case and the previous case exhaust all the possibilities for $tu$). 
\item
$\holmodel{(t_1,\ldots,t_n)}(\varrho) = (\bigcup \rho_i)\bigact (x_1,\ldots,x_n)$ provided that $t_i:\hol{}{\alpha_i}$ for $1\leq i\leq n$, where $\holmodel{t_i}=\rho_i\bigact x_i$, and we choose representatives such that $\dom(\rho_i)\cap\dom(\rho_j)=\varnothing$ for all $1\leq i\neq j\leq n$. 
\item
$\holmodel{(t_1,\ldots,t_n)}(\varrho) = (\holmodel{t_1}(\varrho),\ldots,\holmodel{t_n}(\varrho))$ provided that there exists some $i$ and $\beta$ such that $t_i:\beta$ and $\beta$ is not equal to $\hol{}{\alpha}$ for any PNL sort $\alpha$.
\end{itemize*}
\end{defn}

\begin{rmrk}
\label{rmrk.outline}
Definition~\ref{defn.hol.interpret.terms} propagates to terms the case-split noted in Remark~\ref{rmrk.case-split}. 
We treat terms differently depending on whether they populate the translation of a PNL sort, or not.
We must do this because of how we interpreted types in Definition~\ref{defn.hol.interpretation}.

Just to locate where we are, here is an schematic of the overall structure of the proof of completeness:
$$
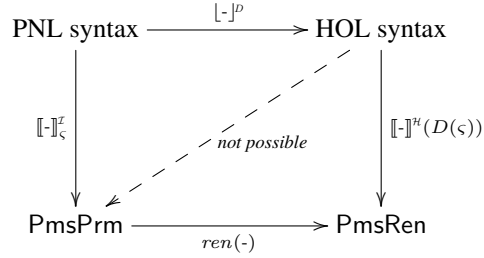
\xymatrix@=6em{
\text{PNL syntax}  \ar[r]^{\hol{D}{\text{-}}}\ar[d]_{\denot{\mathcal I}{\varsigma}{\text{-}}} & \text{HOL syntax} \ar[d]^{\holmodel{\text{-}}(D(\varsigma))}\ar@{-->}[dl]^{\text{\em not possible}} 
\\
\theory{PmsPrm}  \ar[r]_{\Ren{\text{-}}} & \theory{PmsRen} 
}
$$
We translated PNL to HOL using $\hol{D}{\text{-}}$ in Definition~\ref{defn.translation}.
Ideally, to prove completeness we would now give HOL a denotation directly in \theory{PmsRen}.

Unfortunately this is not possible (the dashed arrow) because $[a]r$ translates to $\lam{a}\hol{D}{r}$ and has nominal denotation as an atoms-abstraction $[a]\denot{\mathcal I}{\varsigma}{r}$.
Atoms-abstraction (Definition~\ref{defn.abstraction.sets}) is the graph of a partial function, whereas $\lam{a}\hol{D}{r}$ `wants' to take denotation as a total function.

So instead we use a commuting square as illustrated; in \theory{PmsRen} atoms-abstraction \emph{can} be viewed as a total function, as noted in Remark~\ref{rmrk.total-partial}.
Definition~\ref{defn.hol.interpret.terms} uses this, and fills in the right-hand arrow.
We also need to convert the PNL valuation $\varsigma$ to a HOL valuation $D(\varsigma)$; this comes later in Definition~\ref{defn.epsilond}.

Note that by forming this diagram we give a new semantics to PNL in \theory{PmsRen}, and thus in particular give a semantics to nominal atoms-abstraction in which it becomes interpreted as a total function. 

The top arrow is Definition~\ref{defn.translation}; the left-hand arrow is Definition~\ref{defn.interpret.terms}; and the bottom arrow is Definition~\ref{defn.free.ren}.
Lemma~\ref{lemm.abs.conc.pi} proves commutativity of the square.
\end{rmrk}

\begin{lemm}
\label{lemm.hol.denotren}
Suppose $a\in\mathbb A_\nu$ and $b\in\mathbb A_\nu$. 
Suppose $a\not\in\supp(\varrho(X))$ for every $X\in\fv(t){\setminus}\{a\}$ (including $b$). 
Suppose $\varrho(a)=\id\bigact a$.

Then $\holmodel{t}(\varrho[a\ssm \id\bigact b]) =[a\ssm b]\bigact(\holmodel{t}(\varrho))$.
\end{lemm}
\begin{proof}
By a routine induction on $t$.
We mention two cases:
\begin{itemize*}
\item
The case $t$ is $a$.\quad
Using the fact that $\id\bigact b=[a\ssm b]\bigact a$ in $\mathbb A_\nu$ with the action described in Definition~\ref{defn.bool}.
\item
The case $t$ is $X$ for some HOL variable that is not an atom.\quad
By assumption $a\not\in\supp(\varrho(X))$ and so by Definition~\ref{defn.finsupp},\ $\varrho(X)=[a\ssm b]\bigact\varrho(X)$.
The result follows. 
\qedhere
\end{itemize*}
\end{proof}

\begin{rmrk}
Lemma~\ref{lemm.hol.denotren} may fail if $a\in\supp(\varrho(X))$.
For instance, if $\varrho(X)=a$ where $a\in\mathbb A_\nu$ and $\type(X)=\mu_\nu$ and $X$ is not itself an atom, then $\holmodel{X}(\varrho[a\ssm \id\bigact b])=\id\bigact a$ yet $[a\ssm b]\bigact(\holmodel{X}(\varrho))=[a\ssm b]\bigact(\id\bigact a)=\id\bigact b$. 
\end{rmrk}

We need to check that the denotation of terms populates the denotation of their types, and that $\beta$-equivalent terms receive equal denotations. 
\begin{lemm}
\label{lemm.type.t}
If $t:\beta$ then $\holmodel{t}(\varrho)\in\holmodel{\beta}$.
\end{lemm}

\begin{thrm}
$\holmodel{(\lam{X}t)u}(\varrho) = \holmodel{t}(\varrho[X\ssm\holmodel{u}(\varrho)])$.
\end{thrm}
\begin{proof}
There are two cases, depending on whether $\lam{X}t:\hol{}{[\mathbb A_\nu]\alpha}$ for some PNL sort, or not.
\begin{itemize*}
\item
\emph{The case $t:\hol{}{\alpha}$.}\quad
By Definition~\ref{defn.hol.interpret.terms} $\holmodel{u}(\varrho)=\id\bigact b$ and $\holmodel{\lam{X}t}(\varrho) =\rho\bigact [a]x$, for some $b$, $a$, and $x$. 
$\alpha$-converting if necessary assume $X$ is equal to $a$ which we choose fresh (so $a\not\in\f{nontriv}(\rho)\cup\{b\}$ and $a\not\in\supp(\varrho(Y))$ for every $Y\in\fv(t){\setminus}\{a\}$ and $\varrho(a)=\id\bigact a$).
Then also by definition $\holmodel{(\lam{a}t)u}(\varrho) = ([a\ssm b]\circ\rho)\bigact x$.

Thus it suffices to check that $([a\ssm b]\circ\rho)\bigact x=\holmodel{t}(\varrho[a\ssm b])$.
This follows using Lemma~\ref{lemm.hol.denotren}.
\item
\emph{The case $t:\beta$ where $\beta$ is not equal to $\hol{}{\alpha}$ for any PNL sort $\alpha$.}\quad
This is as standard.
\qedhere
\end{itemize*}
\end{proof}

\subsection{Soundness}

\begin{lemm}
\label{lemm.fV.hol.denot}
If $\varrho(X)=\varrho'(X)$ for all $X\in\fv(t)$ then $\holmodel{t}(\varrho)=\holmodel{t}(\varrho')$.
\end{lemm}
\begin{proof}
By a routine induction on terms.
\end{proof}

\begin{lemm}
\label{lemm.hol.denotsub}
$\holmodel{\rawt}(\varrho[X\ssm \holmodel{\rawu}(\varrho)]) =\holmodel{\rawt[X\ssm \rawu]}(\varrho)$.
\end{lemm}
\begin{proof}
By a routine induction on $\rawt$.
We mention two cases (bearing in mind that in HOL, a variable $X:\beta$ may be an atom in $\mathbb A_\nu$ if $\beta=\mu_\nu$):
\begin{itemize*}
\item
\emph{The case $\rawt$ equals $X$ equals $a\in\mathbb A_\nu$ for some atom $a$.}\quad

By Definition~\ref{defn.hol.interpret.terms},\ $\holmodel{a}(\varrho[a\ssm\holmodel{\rawu}(\varrho)])= \holmodel{\rawu}(\varrho)$.
\item
\emph{The case $\rawt$ equals $\lam{Y}\rawt'$.}\quad

We assume ${Y\not\in\fv(\rawu)}$, so $(\lam{Y}\rawt')[X\ssm \rawu]=\lam{Y}(\rawt'[X\ssm \rawu])$,\ and use the inductive hypothesis.
\qedhere
\end{itemize*}
\end{proof}

\begin{defn}[Validity]
\label{defn.hol.ment}
Call the proposition $\xi$ \deffont{valid} in ${\mathcal H}$ when 
$\holmodel{\xi}(\varrho) = 1$ for all $\varrho$. 

Call the sequent $\xi_1, ..., \xi_n \holcent \chi_1, ..., \chi_p$ \deffont{valid} 
in ${\mathcal H}$ when 
$(\xi_1 \wedge ... \wedge \xi_n) \Rightarrow 
(\chi_1 \vee ... \vee \chi_p)$ is valid.

If this is true for all ${\mathcal H}$ then write $\xi_1,\dots,\xi_n\holment\chi_1,\dots,\chi_p$. 
\end{defn}

\begin{thrm}[Soundness]
\label{thrm.hol.soundness}
If $\Xi\holcent\Chi$ is derivable then $\Xi\holment\Chi$.
\end{thrm} 
\begin{proof}
Fix some interpretation $\mathcal H$.
We work by induction on derivations (Figure~\ref{rSeq}).
We sketch the two non-trivial cases:
\begin{itemize*}
\item
\emph{The case of \rulefont{h\forall L}.}\quad
We check that $u:\type(X)$ implies $\holmodel{\Forall{X}\xi}(\varrho)\leq \holmodel{\xi[X\ssm u]}(\varrho)$.
We reason as follows:
$$
\begin{array}{r@{\ }l@{\quad}l}
\holmodel{\Forall{X}\xi}(\varrho)=&\f{min}\{\holmodel{\lam{X}\xi}(\varrho)y \mid y\in\holmodel{\type(X)}\}
&\text{Definition~\ref{defn.hol.interpret.terms}}
\\
=&\f{min}\{\holmodel{\xi}(\varrho[X\ssm y]) \mid y\in\holmodel{\type(X)}\}
&\text{Definition~\ref{defn.hol.interpret.terms}}
\\
\leq&\holmodel{\xi}(\varrho[X\ssm\holmodel{u}(\varrho)])
&\text{Fact}
\\
=&\holmodel{\xi[X\ssm u]}(\varrho)
&\text{Lemma~\ref{lemm.hol.denotsub}}
\end{array}
$$
In the second use of Definition~\ref{defn.hol.interpret.terms} above, note that $[\mathbb A_\nu]o$ is never of the form $\hol{}{[\mathbb A_\nu]\alpha}$ for any $\alpha$.
\item
\emph{The case of \rulefont{h\forall R}.}\quad
We use Lemma~\ref{lemm.fV.hol.denot} and routine calculations on truth-values.
\qedhere
\end{itemize*}
\end{proof}

\section{Completeness of the translation of PNL to HOL}
\label{sect.pnl.hol.complete}

We are now ready to prove completeness (Theorem~\ref{thrm.PNL.HOL.complete}) of the translation from Definition~\ref{defn.translation}.
The proof is subtle; notably Lemma~\ref{lemm.rho.varrho} and the case of $\Forall{X}\phi$ in Lemma~\ref{lemm.commuting.square} are non-trivial.
Some mathematical action also takes place in Lemma~\ref{lemm.abs.conc.pi} and the case of $\pi\act X$ in Lemma~\ref{lemm.commuting.square}.

\subsection{Renamings and HOL propositions}

We need a few technical observations about how renamings interact with the denotations of HOL propositions:
\begin{lemm}
\label{lemm.equivar.to.triv}
Suppose $G:\rs X\longrightarrow\mathbb B$.
Then for every $\rho$, $G(x)=1$ implies $G(\rho\bigact x)=1$. 
\end{lemm}
\begin{proof}
From equivariance and the fact that $\rho\bigact 1=1$ in $\mathbb B$.
\end{proof}

\begin{corr}
\label{corr.unren.prop}
Suppose $F:\ns X\longrightarrow\mathbb B$.
Then $\Ren{F}(\rho\bigact x)=F(x)$.
\end{corr}

\begin{nttn}
Write $\rho\bigact\varrho$ for the valuation mapping $X$ to $\rho\bigact\varrho(X)$. 
\end{nttn}

\begin{lemm}
\label{lemm.rho.varrho}
Suppose $\xi$ is a HOL proposition.
Then 
\begin{itemize*}
\item
$\holmodel{\xi}(\rho\bigact\varrho)=\holmodel{\xi}(\varrho)$ for every $\rho$ and $\varrho$, and
\item
as a corollary, if $X:\beta$ and $x\in\holmodel{\beta}$
then $\holmodel{\xi}(\varrho[X\ssm x])=\holmodel{\xi}(\varrho[X\ssm\rho\bigact x])$.
\end{itemize*}
\end{lemm}
\begin{proof}
We work by induction on $\xi$. 
For each $\xi$ the corollary follows from the first part using a freshening pair of renamings (see Definition~\ref{defn.freshening.pair}).
For the first part, the case of $\tf g_{\smtf P}$ is by Corollary~\ref{corr.unren.prop}.
The case of $\forall$ follows using the second part and some routine calculations.
The cases of $\bot$ and $\limp$ are immediate. 
\end{proof}

\begin{rmrk}
Lemma~\ref{lemm.rho.varrho} expresses that $\holmodel{\xi}$ does not examine atoms for inequality across its arguments (if it did then Lemma~\ref{lemm.rho.varrho} could not hold, because $\rho$ can identify atoms---make them become equal---in the denotations of variables in $\xi$). 
The corollary is even more powerful: we can even apply renamings to the denotations of individual free variables, and still not affect validity.

We use this in the case of $\Forall{X}\phi$ in Lemma~\ref{lemm.commuting.square} to `jettison' unwanted $\rho$ in the denotation of the quantified variable.
\end{rmrk}

\subsection{The completeness proof}

\begin{nttn}
\label{nttn.D}
Suppose $D=[d_1,\ldots,d_n]$ is a finite list of distinct atoms in $\mathbb A_{\nu_1}$, \ldots, $\mathbb A_{\nu_n}$ respectively.
Suppose $r:\alpha$ is a PNL term.
Then:
\begin{itemize*}
\item
Write $[D]r$ for the PNL term $[d_1]\ldots[d_n]r$.
\item
Write $[\mathbb A_D]\alpha$ for the PNL sort $[\mathbb A_{\nu_1}]\ldots[\mathbb A_{\nu_n}]\alpha$. 
\end{itemize*}
\end{nttn}

\begin{defn}
\label{defn.epsilond}
Given a finite list of distinct atoms $D$, map a PNL valuation $\varsigma$ to a HOL valuation $D(\varsigma)$ defined by 
\begin{frameqn}
D(\varsigma)\quad\text{maps}\quad
\begin{array}[t]{l@{\quad\text{to}\quad}l} 
X:\alpha & \id\bigact [\GammaX]\varsigma(X)\in\holmodel{\hol{}{[\mathbb A_{\GammaX}]\alpha}}\quad\text{and}
\\
a:\nu & a\in\mathbb A_\nu
\end{array}
\end{frameqn}
\end{defn}

Compare Lemma~\ref{lemm.abs.conc.pi} with Lemma~\ref{lemm.hol.pi}:
\begin{lemm}
\label{lemm.abs.conc.pi}
Suppose ${x\in\model{\alpha}}$ and ${\f{nontriv}(\pi)\cap\supp(x)\subseteq D'}$.
Then 
$$
(\id\bigact[D']x)\pi\act D' =\id\bigact (\pi\act x).
$$
\end{lemm}
\begin{proof}
From Definition~\ref{defn.hol.interpret.terms} and rule~2 of Definition~\ref{defn.free.ren}.
\end{proof}

Lemma~\ref{lemm.commuting.square} proves that the schematic diagram of Remark~\ref{rmrk.outline} does indeed commute:
\begin{lemm}
\label{lemm.commuting.square}
Suppose $r:\alpha$ and $\phi$ is a proposition.
Then:
\begin{itemize*}
\item
If $D\cent r$ then $\holmodel{\hol{D}{r}}(D(\varsigma))=\id\bigact\model{r}(\varsigma)$.
\item
If $D\cent\phi$ then $\holmodel{\hol{D}{\phi}}(D(\varsigma))=\model{\phi}(\varsigma)$.
\end{itemize*}
\end{lemm}
\begin{proof}
By inductions on $r$ and $\phi$.
\begin{itemize*}
\item
\emph{The case $\pi\act X$.}\quad
We reason as follows: 
\begin{tab3}
\holmodel{\hol{D}{\pi\act X}}(D(\varsigma))
=& \holmodel{\XD\pi\act(\GammaX)}(D(\varsigma))
&\text{Definition~\ref{defn.translation}}
\\
=& D(\varsigma)(\XD)\pi\act(\GammaX)
&\text{Definition~\ref{defn.hol.interpret.terms}}
\\
=& (\id\bigact[\GammaX]\varsigma(X))\pi\act(\GammaX)
\\
&&\text{Definition~\ref{defn.epsilond}}
\\
=& \id\bigact\pi\act\varsigma(X)
&\text{Lemma~\ref{lemm.abs.conc.pi}},\ 
\\ & & \quad\supp(\varsigma(X)){\subseteq}\pmss(X) 
\\
=& \id\bigact \model{\pi\act X}(\varsigma)
&\text{Definition~\ref{defn.interpret.terms}}
\end{tab3}
Note of the penultimate step that by assumption $D\cent r$, so by Definition~\ref{defn.capture.typing} $\f{nontriv}(\pi)\cap \pmss(X)\subseteq \GammaX$.
\item
\emph{The case $[a]r$.}\quad
By Definition~\ref{defn.translation}
\begin{tab3}
&{\holmodel{\hol{D}{[a]r}}(D(\varsigma))=\holmodel{\lam{a}\hol{D}{r}}(D(\varsigma)).}\hspace{-20em}
\end{tab3}
By inductive hypothesis 
$\holmodel{\hol{D}{r}}(D(\varsigma))=\id\bigact\model{r}(\varsigma)$
and so by Definition~\ref{defn.hol.interpret.terms} (eliding routine freshness reasoning for $a$)
\begin{tab3}
&{\holmodel{\lam{a}\hol{D}{r}}(D(\varsigma))=\id\bigact[a]\model{r}(\varsigma).}\hspace{-20em}
\end{tab3}
Finally by Definition~\ref{defn.interpret.terms} 
\begin{tab3}
&{\id\bigact [a]\model{r}(\varsigma)=\id\bigact \model{[a]r}(\varsigma)}\hspace{-20em}
\end{tab3}
and by transitivity of equality we are done.
\item
\emph{The case $\tf P(r)$.}\quad
We reason as follows:
\begin{tab3}
\holmodel{\hol{D}{\tf P(r)}}(D(\varsigma)) 
=& \holmodel{\tf g_{\smtf P}(\hol{D}{r})}(D(\varsigma))
&\text{Definition~\ref{defn.translation}}
\\
=& \tf g_{\smtf P}^\hiden(\holmodel{\hol{D}{r}}(D(\varsigma))) 
&\text{Definition~\ref{defn.truth}}
\\
=& \tf g_{\smtf P}^\hiden(\id\bigact\model{r}(\varsigma))
&\text{part~1}
\\
=& \Ren{\tf P^\iden}(\id\bigact\model{r}(\varsigma))
&\text{Definition~\ref{defn.hol.interpret.terms}} 
\\
=& \tf P^\iden(\model{r}(\varsigma))
&\text{Corollary~\ref{corr.unren.prop}}
\\
=& \model{\tf P(r)}(\varsigma)
&\text{Definition~\ref{defn.truth}}
\end{tab3}
\item
\emph{The case $\Forall{X}\phi$.}\quad
Write $\alpha=\sort(X)$. 
From Definition~\ref{defn.hol.interpret.terms}
$$
\holmodel{\hol{D}{\Forall{X}\phi}}(D(\varsigma)) 
=\f{min}\{\holmodel{\hol{D}{\phi}}(D(\varsigma)[X\ssm x])\mid x\in\holmodel{\hol{}{[\mathbb A_{\GammaX}]\alpha}}\} 
$$
By construction in Definition~\ref{defn.hol.interpretation} every $x\in\holmodel{\hol{}{[\mathbb A_{\GammaX}]\alpha}}$ has the form $\rho\bigact x'$ for $x'\in [\GammaX]\model{\alpha}$.
By Lemma~\ref{lemm.rho.varrho} we have
\begin{multline*}
\f{min}\{\holmodel{\hol{D}{\phi}}(D(\varsigma)[X\ssm x])\mid x\in\holmodel{\hol{}{[\mathbb A_{\GammaX}]\alpha}}\} 
\\
=
\f{min}\{\holmodel{\hol{D}{\phi}}(D(\varsigma)[X\ssm \id\bigact x'])\mid x'\in\model{[\mathbb A_{\GammaX}]\alpha}\}  
\end{multline*}
Using Lemma~\ref{lemm.rho.varrho} again we assume without loss of generality that $\supp([\GammaX]x')\subseteq\pmss(X)\setminus\GammaX$, and so:
\begin{multline*}
\f{min}\{\holmodel{\hol{D}{\phi}}(D(\varsigma)[X\ssm \id\bigact [\GammaX]x'])\mid x'\in\model{[\mathbb A_{\GammaX}]\alpha}\}  
\\
=\f{min}\{\holmodel{\hol{D}{\phi}}(D(\varsigma)[X\ssm \id\bigact x''])\mid x''\in\model{\alpha},\ \supp(x''){\subseteq}\pmss(X)\}  
\end{multline*}
Now we unfold definitions and use the inductive hypothesis which tells us that $D(\varsigma)[X\ssm \id\bigact [\GammaX]x'']=D(\varsigma[X\ssm x''])$, and we obtain:
$$
\hspace{-2em}\begin{array}{r@{}l}
\f{min}\{\holmodel{\hol{D}{\phi}}(D(\varsigma)[X\ssm \id&\bigact [\GammaX]x''])\mid x''\in\model{\alpha},\ \supp(x''){\subseteq}\pmss(X)\}  
\\
&=\f{min}\{\holmodel{\hol{D}{\phi}}(D(\varsigma[X\ssm x'']))\mid x''\in\model{\alpha},\ \supp(x''){\subseteq}\pmss(X)\}  
\\
&=\f{min}\{\model{\phi}(\varsigma[X\ssm x''])\mid x''\in\model{\alpha},\ \supp(x''){\subseteq}\pmss(X)\}  
\\
&=\model{\Forall{X}\phi}(\varsigma)
\end{array}
$$
\end{itemize*}
\end{proof}

\begin{corr}
\label{corr.notsubseteq}
Suppose $\Phi=\{\phi_1,\dots,\phi_n\}$ and $\Psi=\{\psi_1,\dots,\psi_p\}$ and $D\cent\Phi$, and $D\cent\Psi$ (Definition~\ref{defn.capture.typing}).
Suppose $\mathcal I$ is a PNL interpretation and suppose $\phi_1,\dots,\phi_n\nopicent\psi_1,\dots,\psi_p$ is not valid in $\mathcal I$.

Then $\mathcal H$ from Definition~\ref{defn.hol.interpretation} is a HOL interpretation and
$\hol{D}{\phi_1},\dots,\hol{D}{\phi_n}\holcent\hol{D}{\psi_1},\dots,\hol{D}{\psi_p}$ is not valid in $\mathcal H$.
\end{corr}
\begin{proof}
Suppose $\varsigma$ is such that $\model{\phi_1\land\dots\land\phi_n}(\varsigma)=1$ and $\model{\psi_1\lor\dots\lor\psi_p}(\varsigma)=0$.
We use Lemma~\ref{lemm.commuting.square} for $D(\varsigma)$ (Definition~\ref{defn.epsilond}).
\end{proof}

\begin{thrm}[Completeness]
\label{thrm.PNL.HOL.complete}
Suppose $D\cent\Phi$ and $D\cent\Psi$.
If $\Phi\not\nopicent\Psi$ then $\hol{D}{\Phi}\not\holcent\hol{D}{\Psi}$.
\end{thrm}
\begin{proof}
We use the contrapositive of completeness of restricted PNL (Theorem~\ref{thrm.reduced.pnl.completeness}), then Corollary~\ref{corr.notsubseteq}, then the contrapositive of HOL soundness (Theorem~\ref{thrm.hol.soundness}).
\end{proof} 

\section{Conclusions}

We have translated a logic with its own proof-theory, syntax, and sound and complete semantics.
Any formal theory specified in the PNL fragment of this paper can be systematically, soundly, and completely translated to HOL.

For the reader interested in nominal techniques, the main contribution of this paper is that in 
proving completeness of the translation, we have given another 
semantics of permissive nominal logic, besides the `obvious' one
in nominal sets. In this new semantics, a term of
the form $[a]t$ is interpreted as a function, like $\lam{a}t$ would be in higher-order logic. 
This shows at the semantic level an implicit similarity between PNL and HOL (we discuss presheaves in the next Subsection).

For the reader interested in higher-order logic, this paper is of interest because its image is readily identified with the \emph{higher-order patterns} developed by Miller \cite{miller:logpll} (so that, intuitively, restricted PNL could be thought of as a compact first-order logic and nominal semantics for higher-order patterns).

In this semantics the sort $[\mathbb A]\alpha$ is not interpreted as the set of all
functions from atoms to the interpretation of $\alpha$,
but as a small subset of this function space. 
This is an old idea: since Henkin, models of HOL have been constructed to cut down on the full function-space (e.g. to create a complete semantics \cite[Section~55]{andrews:intmlt}).
Moreover in weak HOAS to avoid so-called \emph{exotic terms}, function 
existence axioms must be weakened in HOL: for instance, the
description axiom that entails the existence of a function for all
functional relations has to be dropped (an alternative is to introduce an explicit modality \cite{despeyroux:prirh-jv}).
We now have a new view of these `smaller' function-spaces as being the image of nominal atoms-abstractions via the semantic operations considered in this paper.

\subsection{The big picture}

It is quite interesting to understand this paper via a contrast between what nominal foundations and `ordinary' foundations provide.

The denotation of an open term in PNL is an open element of a nominal set (i.e. an element with atoms in its support): so for instance an atom maps just to itself in the denotation, rather than being assigned some other element by a valuation.
PNL has binding term-formers because its open denotation in nominal sets provides constructions like atoms, permutations, and atoms-abstraction (Definition~\ref{defn.abstraction.sets}); see Section~\ref{sect.permissive-nominal.sets}.
 
In contrast, in HOL there is no concept of open element---there are open \emph{terms}, but their denotation is closed by a valuation.  
If we want the effect of a free variable then we Skolemise/raise.
Indeed, \emph{any} mathematics using technology based on the HOL/ZF (Zermelo--Fraenkel sets) foundations of mathematics \emph{must} handle names either as something like functional arguments or something like numbers; simply because numbers and functions are what HOL/ZF foundations provide.

The translation from PNL to HOL in Section~\ref{sect.translation.sound} works by raising.\footnote{If we translated PNL to first-order logic then we would probably map atoms to numbers instead. This is future work.} 
However, we can only raise $n$ variables, in order. 
So the translation has to be on a per-derivation basis, including the (finitely many) atoms of interest in that (finite) derivation.
Furthermore, we lose the equivariance (name symmetry) of full PNL.
So we can only naturally translate \emph{individual} derivations in \emph{restricted} PNL.

This matters because in losing symmetry we lose what makes nominal techniques so distinctive. 
So although we show how to translate a complete `nominal' proof to a complete `HOL' proof, we also see how the way in which nominal and HOL proofs are manipulated and combined, are different.

The translation is not entirely trivial to define and prove sound, but the technically hardest part in this paper is clearly the proof of its completeness.
For this we build a hybrid denotation for HOL in nominal renaming sets which is sound but in which certain function spaces are restricted to be `not too large'.
That motivates the bulk of the technical mathematics of this paper.

\subsection{Permissive nominal logic in perspective}
\label{subsect.pnl.in.perspective}

Permissive-nominal logic is the endpoint---so far---of an evolution as follows: 
\begin{itemize*}
\item
Fraenkel-Mostowski set theory and a first-order axiomatisation by Pitts introduced and described the underlying nominal sets models in first-order logic \cite{gabbay:newaas-jv,pitts:nomlfo-jv}.
\item
Nominal terms introduced a dedicated syntax with two-levels of variable and freshness side-conditions \cite{gabbay:nomu-jv}.
\item
Nominal algebra and $\alpha$Prolog inserted nominal terms syntax into formal reasoning systems \cite{gabbay:nomuae,cheney:alppl}.
\item
Permissive-nominal terms introduced permission sets \cite{gabbay:perntu-jv}.
\item
PNL introduced a proof-theory and universal quantifier for nominal terms unknowns \cite{gabbay:pernl,gabbay:pernl-jv}. 
\end{itemize*}
Meanwhile in the semantics
\begin{itemize*}
\item
Nominal renaming sets extended nominal sets from a permutation action to a renaming action \cite{gabbay:nomrs}.
\item
A permissive version of nominal algebra (an equality fragment of PNL) was given semantics in \theory{PmsPrm} and theories were translated from HOL \cite{gabbay:unialt}, but this was done purely syntactically without using nominal renaming sets and without considering universal quantification.
\end{itemize*}

The categories \theory{PmsPrm} and \theory{PmsRen} from Definition~\ref{defn.fps} are identical to the categories of nominal sets and nominal renaming sets from \cite{gabbay:newaas-jv} and \cite{gabbay:nomrs}, except that here we insist on supporting \emph{permission} sets instead of supporting \emph{finite} sets.
 
The reader familiar with presheaf techniques will see in \theory{PmsRen} the category $\mathsf{Sets}^{\mathbb F}$ (presheaves over the category of finite sets and functions between them).
\theory{PmsRen} corresponds to presheaves (not quite over $\mathbb F$, as discussed in the previous paragraph) that preserve pullbacks of pairs of monos \cite{gabbay:nomrs} and because of this it admits an arguably preferable sets-based presentation.
(In the same sense, \theory{PmsPrm} corresponds to $\mathsf{Sets}^{\mathbb I}$.) 

If for the sake of argument we set aside the issues of finiteness and preserving pullbacks of monos, then this paper can be summed up as follows: PNL, and thus nominal terms, can be given a semantics in something that looks like $\mathsf{Sets}^{\mathbb F}$.
This semantics is functional in that atoms-abstractions in $\mathsf{Sets}^{\mathbb F}$ can be naturally identified with total functions, though not all of them, which is good.
HOL can also be given a semantics in something that looks like $\mathsf{Sets}^{\mathbb F}$, and in such a way that it overlaps with the semantics of PNL, as described in Definition~\ref{defn.hol.interpret.terms} and~\ref{lemm.commuting.square}.
We describe and exploit that overlap, in this paper.

\theory{PmsRen} from Definition~\ref{defn.fps} is related to the category of (finitely-supported) nominal renaming sets from \cite{gabbay:nomrs}.
Here, the difference that $x\in|\rs X|$ need not have finite support is significant because it is impossible with a finite renaming to rename $\supp(x)$ to be entirely disjoint for some other permission set $S$.
The definitions and proofs in Subsection~\ref{subsect.exp} are delicately revised with respect to those in \cite[Section~3]{gabbay:nomrs}.
Thus this paper contributes to the use of non-finitely-supported objects in nominal techniques, building on \cite{gabbay:nomrs} and also on Cheney's and the second author's considerations of infinitely supported permutation sets \cite{cheney:comhtn,gabbay:genmn}.

A similar construction as in Subsection~\ref{subsect.free.ext} has been considered, also in the context of names, though tersely, in Fiore and Turi's paper on the semantics of name and value passing \cite{fiore:semnvp}.
The reader can compare for example the final two paragraphs of Subsection~1.3 in \cite{fiore:semnvp} with Definition~\ref{defn.free.ren} from Subsection~\ref{subsect.free.ext}.
Fiore and Turi want substitutions to model bisimulation in the presence of name-generation and message-passing; we want renamings to model function application on names.
The underlying technical demands overlap and are similar.

Fiore and Turi's framework includes the possibility of arbitrary substitutions for atoms (not just what we call renamings: substitution of atoms for atoms).
This was apparent in \cite{fiore:semnvp} and is developed greatly in subsequent work by Fiore and Hur \cite{fiore:secoel}.
We hypothesise that from the point of view of PNL, their logic and semantics correspond to PNL enriched with substitution actions like those in \cite{gabbay:pernl,gabbay:capasn}, but this remains to be checked.\footnote{Conversely, Fiore and Hur would view PNL as a restriction of their logic \emph{without} substitution.  The two points of view are consistent with each other, of course, and it is interesting that different authors are converging on similar systems.  
It might be worth mentioning that \emph{deduction modulo} by the first author with Hardin and Kirchner was designed to mediate between these kinds of design decisions while retaining proof-theory \cite{dowek:dedm}.}

Levy and Villaret translated nominal unification problems to higher-order unification problems \cite{levy:nomufh}.
A similar but more detailed analysis, translating solutions and introducing the same notion of capturable atoms as used in the capture typings in this paper, appears in the paper which introduced permissive nominal terms \cite{gabbay:perntu-jv}.
See also a journal version of Levy and Villaret's paper \cite{levy:nomufh-jv}, which expanded on their previous work by eliminating freshness contexts (in a similar spirit to PNL, we feel, though the details are different).
This paper can be viewed as a very considerable extension, refinement, and generalisation of these works: this paper is their grandchild, so to speak, via two other papers \cite{gabbay:pernl,gabbay:unialt}.

The extension of nominal sets to nominal renaming sets is free.
This is touched on in Lemma~\ref{lemm.non-iso} when we note that $[a\ssm b]\bigact(a,b)$ and $\id\bigact(b,b)$ are distinct elements in $\Ren{\mathbb A_\nu\times\mathbb A_\nu}$ in \theory{PmsRen}; this happens because the free construction `suspends the non-injectivity' of $[a\ssm b]$ on $(a,b)$.
This is as things should be, in order to obtain completeness.
The second author has considered a more radical non-free construction \cite{gabbay:stusun}, which has the effect of extending atoms-abstraction to a total function and in which $[a\ssm b]\act x$ really does identify $a$ with $b$ in $x$ in a suitable sense.

As we have emphasised, we translate a fragment of PNL to HOL.
In \cite{gabbay:pernl} we considered full PNL with \emph{equivariance}, which corresponds to strengthening the axiom rule \rulefont{Ax^{\nopi}} in Figure~\ref{rSeq} from 
$\begin{prooftree}
\justifies
\Phi,\,\phi\nopicent \phi,\,\Psi
\end{prooftree}
$
to
$
\begin{prooftree}
\justifies
\Phi,\,\phi\cent \pi\act\phi,\,\Psi 
\end{prooftree}
$
as illustrated in Figure~\ref{Seq}.
This internalises the equivariance assumed in Definition~\ref{defn.interpret.I} and allows us to derive e.g. $\tf P(a) \cent \tf P(b)$.

In the journal version \cite{gabbay:pernl-jv} of \cite{gabbay:pernl} we strengthen PNL further by allowing a \emph{shift}-permutation.
This is a non-finitely-supported bijection on $\mathbb A$ similar to a \emph{de Bruijn shift function} $\uparrow$ \cite[Subsection~2.2]{abadi:exps}. 
Its effect in this paper is to make all permission sets isomorphic up to bijection (e.g. $\atomsdown\cup\{a\}=\pi\act\atomsdown$ for some $\pi$, where $a\not\in\atomsdown$) and this deals with a subtle restriction in the power of universal quantification discussed for instance in \cite[Example~2.29]{gabbay:pernl}.
Briefly, \emph{shift} lets us derive $\Forall{X}\tf P(X)\cent \tf P(Z)$ where $\pmss(X)=\atomsdown$ and $\pmss(Z)=\atomsdown\cup\{a\}$ where $a\not\in\atomsdown$, which was not possible in the PNL from \cite{gabbay:pernl}. 

Neither equivariance nor \emph{shift} are translated to HOL in this paper; more on this in the next subsection.

\subsection{Future work}

We have translated Permissive-Nominal Logic
to Higher-Order Logic. The translation is not surjective: all variables are at most
second-order; all constants are at most third-order; higher types
are not used; and in fact all terms in the image of the translation are \emph{$\lambda$-patterns} \cite{miller:logpll}.
In addition, the translation is not total: we have dropped equivariance.

This is with good reason.
We have not been able to simulate equivariance in HOL---not without `cheating' by simply adding it (and causing a blowup in the size of propositions). 
We have not proved this impossible, but we hypothesise that it cannot be done.
We further hypothesise (based on preliminary calculations not included in this paper) that HOL augmented with the $\nabla$-quantifier from \cite{Miller:protgj} would allow us to express equivariance. 

It is not currently clear how to extend HOL with a \emph{shift}-like
permutation as discussed in \cite{gabbay:pernl-jv,gabbay:nomtnl}.
This seems reasonable since $\f{shift}$ would correspond to an infinite renaming. 

Some natural theories in PNL might correspond to other fragments of HOL.
Notably, it is not known what relation exists between HOL and PNL with the theory of atoms-substitution from \cite{gabbay:capasn-jv,gabbay:pernl-jv}.

\hyphenation{Mathe-ma-ti-sche}

{\small
\ \\
\noindent \emph{Acknowledgements:}\quad
 Thanks to the careful attention of two anonymous referees.
We acknowledge the support of DIGITEO in Paris, of grant RYC-2006-002131 at the Polytechnic University of Madrid, and of the Leverhulme Trust in the UK.}

\appendix

\section{Soundness and completeness of restricted PNL with respect to non-equivariant models}
\label{sect.completeness}

\subsection{Validity and soundness} 

\begin{defn}[Validity]
\label{defn.pnl.ment}
Suppose $\mathcal I$ is a non-equivariant interpretation of a signature $\mathcal S$ (Definition~\ref{defn.interpret.I}).
Call the proposition $\phi$ \deffont{valid} in ${\mathcal I}$ when 
$\denot{\mathcal I}{\varsigma}{\phi} = 1$ for all $\varsigma$. 

Call the sequent 
$\phi_1, ..., \phi_n \cent \psi_1, ..., \psi_p$ 
\deffont{valid} 
in ${\mathcal I}$ when 
$(\phi_1 \wedge ... \wedge \phi_n) \Rightarrow (\psi_1 \vee ... \vee \psi_p)$ is valid.

If this is true for all non-equivariant ${\mathcal I}$ then write $\phi_1,\dots,\phi_n\nopiment\psi_1,\dots,\psi_p$. 
If this is true for all equivariant ${\mathcal I}$ then write $\phi_1,\dots,\phi_n\ment\psi_1,\dots,\psi_p$. 
\end{defn}

\begin{thrm}[Soundness]
\label{thrm.pnl.soundness}
\begin{enumerate*}
\item
If $\Phi\nopicent\Psi$ is derivable then $\Phi\nopiment\Psi$.
\item
If $\Phi\cent\Psi$ is derivable then $\Phi\ment\Psi$.
\end{enumerate*}
\end{thrm} 
\begin{proof}
Fix some interpretation $\mathcal I$.
We work by induction on derivations.
The case of \rulefont{\forall L} uses Lemma~\ref{lemm.denotsub}.
The case of \rulefont{\forall R} uses Lemma~\ref{lemm.fV.denot}.
Other rules are routine by unpacking definitions.

If the interpretation $\mathcal I$ is fully equivariant then it can further be proved that $\denot{\mathcal I}{\varsigma}{\phi}=\denot{\mathcal I}{\varsigma}{\pi\act\phi}$ always, so that \rulefont{Ax} is valid.  
If $\mathcal I$ is not fully equivariant, then just \rulefont{Ax^{\nopi}} is valid.
\end{proof}

\begin{thrm}
\label{thrm.rPNL.cut}
\rulefont{Cut} is admissible in both full and restricted PNL.
\end{thrm}
\begin{proof}
The proof for full PNL is in \cite[Section~7]{gabbay:pernl-jv} or \cite[Subsection~11.2]{gabbay:nomtnl}; the derivation rules are almost exactly those of first-order logic, and so is the proof of cut-elimination.
The argument for restricted PNL is identical; we note that none of the cut-eliminating transformations add $\pi$ to axiom rules unless they are already there, so the same reductions on derivations work also for the restricted system.
\end{proof}

\subsection{Completeness}

In \cite{gabbay:pernl-jv,gabbay:nomtnl} we prove completeness of full PNL with respect to equivariant models, by means of a Herbrand construction (a model built out of syntax). 
We can leverage this result to concisely prove completeness of restricted PNL with respect to non-equivariant models, without having to repeat the model constructions.

For this subsection, fix the following data:
\begin{itemize*}
\item
A signature $\mathcal S=(\mathcal A,\mathcal B,\mathcal F,\mathcal P,\f{ar},\mathcal X)$.
\item
A formula $\phi$ such that $\not\nopicent\phi$.
\end{itemize*}

\begin{defn}
\label{defn.S.pi}
Define a new signature $\mathcal S^\pi$ as follows:
\begin{itemize*}
\item
$\mathcal A^\pi=\mathcal A$ and $\mathcal B^\pi=\mathcal B\cup\{\tau^\pi\}$ (so we have the same atom sorts and the same base sorts, plus one extra base sort $\tau^\pi$).
\item
$\mathcal F^\pi=\mathcal F$ and $\mathcal P^\pi=\mathcal P$ (so we have the same term- and proposition-formers).
\item
If $\tf f\in\mathcal F$ then $\f{ar}^\pi(\tf f)=\f{ar}(\tf f)$ (the term-formers are identical).
\item
If $\tf P\in\mathcal P$ and $\f{ar}(\tf P)=\alpha$ then $\f{ar}^\pi(\tf P)=(\tau^\pi,\alpha)$ (so proposition-formers take one extra argument of sort $\tau^\pi$). 
\item
$\mathcal X^\pi=\mathcal X\cup\{Z_{i,S}^\pi\mid i\in\mathbb N,\ S\text{ a permission set}\}$ where $\sort(Z_{i,S}^\pi)=\tau^\pi$ (so we add unknowns of sort $\tau^\pi$).
\end{itemize*}
Now fix some particular unknown $Z^\pi$ with $\sort(Z^\pi)=\tau^\pi$ and such that $\fa(\phi)\subseteq\pmss(Z^\pi)$.
\end{defn}

\begin{defn}
Define a translation $\text{-}^\pi$ from PNL propositions in the signature $\mathcal S$ to PNL propositions in the signature $\mathcal S^\pi$ by mapping $\tf P(r)$ to $\tf P(Z^\pi,r)$ and extending this in the natural way to all predicates.
\end{defn}

Our proof depends on the following technical lemma about restricted PNL:
\begin{lemm}
\label{lemm.fa.restrict}
If $\Phi\cent\Psi$ is derivable in full PNL then there exists a derivation ${\mathfrak B}$ such that every sequent $\Phi'\cent\Psi'$ in ${\mathfrak B}$ satisfies $\fa(\Phi')\cup\fa(\Psi')\subseteq\fa(\Phi)\cup\fa(\Psi)$.
\end{lemm}
\begin{proof}
By cut-elimination of restricted PNL (Theorem~\ref{thrm.rPNL.cut}) if a derivation of $\Phi\cent\Psi$ exists then a cut-free derivation exists.
We now examine the derivation rules in Figure~\ref{rSeq} and the definition of free atoms in Definition~\ref{defn.fa} and note that the rules \rulefont{{\limp}L}, \rulefont{{\limp}R}, \rulefont{\forall L}, and \rulefont{\forall R} do not increase the free atoms moving from below the line to above the line.\footnote{\rulefont{\forall R} and \rulefont{\forall L} can increase the free \emph{unknowns}---but not the free atoms.} 
\end{proof}

\begin{lemm}
\label{lemm.pi.r.fa}
$\pi\act r=\pi'\act r$ if and only if $\pi(a)=\pi'(a)$ for every $a\in\fa(r)$, and similarly for $\phi$.
\end{lemm}
See \cite[Lemma~3.2.9]{gabbay:nomtnl} or \cite[Lemma~4.15]{gabbay:perntu-jv}.

\begin{prop}
\label{prop.completeness.lemma}
If $\Phi^\pi\cent\Psi^\pi$ in PNL and $\fa(\Phi)\cup\fa(\Psi)\subseteq \pmss(Z^\pi)$ then $\Phi\nopicent\Psi$. 
\end{prop}
\begin{proof}
Using cut-elimination of full PNL (Theorem~\ref{thrm.rPNL.cut}) assume a cut-free PNL derivation ${\mathfrak B}$ of $\Phi^\pi\cent\Psi^\pi$. 
Because of Lemma~\ref{lemm.fa.restrict}, the condition on free atoms holds of every sequent in ${\mathfrak B}$.
Because of the form of the derivation rules in Figure~\ref{Seq}, ${\mathfrak B}$ cannot instantiate $Z^\pi$.

So we can go through the entire syntax of ${\mathfrak B}$ and delete $Z^\pi$ to obtain a structure that is a candidate for being a derivation in restricted PNL of $\Phi\nopicent\Psi$.

The only non-trivial thing to check is that valid instances of \rulefont{Ax} are transformed to valid instances of \rulefont{Ax^\pi}.
Suppose we deduce $\Phi^\pi,\psi^\pi\cent\pi'\act\psi^\pi,\Psi^\pi$ using \rulefont{Ax}.
By assumption $\pi'\act\psi^\pi={\psi'}^\pi$ for some $\psi'$.
It follows that $\pi'\act Z^\pi=\id\act Z^\pi$ (recall from Subsection~\ref{subsect.aeq} that we quotient by $\alpha$-equivalence) and so by Lemma~\ref{lemm.pi.r.fa} that $\pi'(a)=a$ for all $a\in\pmss(Z^\pi)$.
By assumption $\fa(\Phi)\cup\fa(\Psi)\cup\fa(\psi)\cup\fa(\psi')\subseteq \pmss(Z^\pi)$ and so by Lemma~\ref{lemm.pi.r.fa} $\psi=\psi'$, and we are done. 
\end{proof}

\begin{thrm}
\label{thrm.reduced.pnl.completeness}
If $\Phi\nopiment\Psi$ then $\Phi\nopicent\Psi$.
\end{thrm}
\begin{proof}
We prove the contrapositive, that if $\Phi\not\nopicent\Psi$ then $\Phi\not\nopiment\Psi$.
Suppose $\Phi\not\nopicent\Psi$.
Using the constructions above we augment to a signature $\mathcal S^\pi$ (Definition~\ref{defn.S.pi}) with some $Z^\pi$ with $\fa(\Phi)\cup\fa(\Psi)\subseteq\pmss(Z^\pi)$.
Thus by Proposition~\ref{prop.completeness.lemma} $\Phi^\pi\not\cent\Psi^\pi$.
 
By completeness of full PNL with respect to equivariant models (\cite[Theorem~3.45]{gabbay:pernl-jv}, \cite[Theorem~9.4.15]{gabbay:nomtnl}) we have that $\Phi^\pi\not\ment\Psi^\pi$.
So there exists an equivariant model $\mathcal I$ and valuation $\varsigma$ to $\mathcal I$ such that $\denot{\mathcal I}{\varsigma}{\Phi}=1$ and $\denot{\mathcal I}{\varsigma}{\Psi}=0$. 
It is now routine to convert $\mathcal I$ into a non-equivariant model of the original signature $\mathcal S$ by taking $\tf P^\hden(x)=\tf P^\iden(\varsigma(Z^\pi),x)$. 
\end{proof}

\end{document}